%% file: main.tex
\documentclass[sigconf]{acmart}

\usepackage[color,secparam=\kappa]{crypto}

\usepackage{tikz}
\usetikzlibrary{arrows.meta, positioning}
\usetikzlibrary{decorations.pathreplacing}
\usetikzlibrary{patterns}
\usepackage{framed}
\usepackage{tabularx}
\usepackage{multirow}
\usepackage{amsthm}
\usepackage{amsmath}
\usepackage{enumitem}
\usepackage{amsfonts}
\usepackage{graphicx}
\usepackage{stmaryrd}
\usepackage{hyperref}
\usepackage{algorithm2e}
\usepackage{subcaption}
\usepackage{diagbox}
\usepackage[table]{xcolor}
\usepackage{xcolor}
\usepackage{tcolorbox}
\usepackage{marvosym}

\setlist[enumerate]{leftmargin=2.5em}

\usetikzlibrary{
    positioning,
    arrows.meta,
    calc,
    fit
}

\newtheorem{definition}{Definition}[section]
\newtheorem{lemma}{Lemma}[section]

\newtheorem{theorem}{Theorem}[section]
\AtBeginDocument{%
  }

\copyrightyear{2026}
\acmYear{2026}
\setcopyright{cc}
\setcctype{by}
\acmConference[CCS '26]{Proceedings of the 2026 ACM SIGSAC Conference on Computer and Communications Security}{November 15--19, 2026}{The Hague, Netherlands}
\acmBooktitle{Proceedings of the 2026 ACM SIGSAC Conference on Computer and Communications Security (CCS '26), November 15--19, 2026, The Hague, Netherlands}
\acmDOI{10.1145/3830454.3846569}
\acmISBN{979-8-4007-2871-6/2026/11}

\begin{document}

\title{Efficient Fuzzy PSI under One-Sided Assumptions}

\author{Xinpeng Yang}
\affiliation{%
  \institution{Nanyang Technological University}
  \country{Singapore}
}
\email{XINPENG004@e.ntu.edu.sg}

\author{Meng Hao}
\authornote{Meng Hao is the corresponding author}
\affiliation{%
  \institution{Singapore Management University}
  \country{Singapore}
}
\email{menghao303@gmail.com}

\author{Yanxue Jia}
\affiliation{%
  \institution{Illinois Institute of Technology}
  \country{Chicago, Illinois, USA}
}
\email{yjia12@illinoistech.edu}

\author{Chenkai Weng}
\affiliation{%
  \institution{Arizona State University}
  \country{Tempe, Arizona, USA}
}
\email{chenkai.weng@asu.edu}

\author{Yonggang Wen}
\affiliation{%
  \institution{Nanyang Technological University}
  \country{Singapore}
}
\email{ygwen@ntu.edu.sg}

\author{Tianwei Zhang}
\affiliation{%
  \institution{Nanyang Technological University}
  \country{Singapore}
}
\email{tianwei.zhang@ntu.edu.sg}

\renewcommand{\shortauthors}{Xinpeng Yang et al.}

\begin{abstract}
Fuzzy private set intersection (PSI) enables two parties to identify approximately matching elements between their input sets, where two elements are considered a match if their distance is at most a threshold $\delta$ under a given metric. Although substantial progress has been made, existing constructions for \textit{general Minkowski distances} either rely on strong two-sided geometric separation assumptions or incur substantial overhead under one-sided assumptions.

In this work, we present the first concretely efficient fuzzy PSI protocols for general $L_{p\in[1,\infty]}$ distances under \textit{one-sided assumptions}, relying solely on lightweight symmetric-key primitives. Our constructions support both sender-sided and receiver-sided settings.  We further study sparser input distributions and present more efficient protocols tailored to this case. To reduce the overhead scaling with $\delta$, we non-trivially incorporate prefix trie techniques into our protocols, achieving $O(\log\delta)$ complexity for general $L_{p\in[1,\infty]}$ distances for the first time, improving upon $O((\log\delta)^d)$ or $O(\delta)$ complexities of prior works.

Extensive experiments, across a wide range of parameter settings, show that our protocols significantly outperform prior works under the same assumptions. Specifically, against van Baarsen and Pu {(EUROCRYPT'24)}, our protocols achieve up to $248\times$ faster computation and up to $20\times$ lower communication. Against Dang et al. {(CCS'25)}, we achieve up to $568\times$ speedup and up to $63\times$ communication reduction. Against Bui et al. {(ASIACRYPT'25)}, we achieve up to $4978\times$ faster computation and up to $282\times$ lower communication.

\end{abstract}




\maketitle

\input{intro}

\input{prelimilary}

\input{overview}

\input{method}

\input{experiment}

\input{conclusion}

\section*{Acknowledgments}
This work was supported by the Nanyang Technological University Centre in Computational Technologies for Finance (NTU-CCTF). It was also supported by the National Research Foundation, Singapore, and Cyber Security Agency of Singapore under its National Cybersecurity R\&D Programme and CyberSG R\&D Cyber Research Programme Office. Any opinions, findings, and conclusions or recommendations expressed in this material are those of the author(s) and do not necessarily reflect the views of NTU-CCTF, National Research Foundation, Singapore, Cyber Security Agency of Singapore as well as CyberSG R\&D Programme Office, Singapore.

\bibliographystyle{ACM-Reference-Format}
\bibliography{sample-base}

\appendix

\input{appendix}

\end{document}

%% file: intro.tex
\section{Introduction}

Standard private set intersection (PSI)~\cite{meadows1986more,freedman2004efficient,dong2013private,pinkas2015phasing,pinkas2020psi,raghuraman2022blazing,chen2017fast,wu2023efficient,kales2019mobile,ion2017private,hao2024unbalanced} has been studied for decades, with significant progress made in both efficiency and scalability. It enables two parties to compute their common elements without revealing any information about non-matching ones. PSI has found applications in contact tracing~\cite{wu2023efficient,kales2019mobile}, ad conversion measurement~\cite{ion2017private}, and others~\cite{kiss2017private,alamati2021laconic}. However, its exact-matching requirement limits applicability in settings where data are inherently noisy or imprecise, such as biometric identification or location services.

Fuzzy PSI~\cite{garimella2022structure,garimella2024computation,van2024fuzzy,gao2025efficient,zhang2025fast,dang2025ccs,piske2025distance,van2025,richardson2024fuzzy,fss25,yang2026} generalizes this functionality to support \textit{approximate matching}. Specifically, in a $d$-dimensional space, let the sender hold a set $Q = \{\vecq_1, \ldots, \vecq_m\}$ and the receiver hold a set $W = \{\vecw_1, \ldots, \vecw_n\}$. Fuzzy PSI outputs to the receiver all elements $\vecq \in Q$ for which there exists some $\vecw \in W$ satisfying $\mathsf{dist}(\vecq, \vecw) \le \delta$, where $\delta$ is the matching threshold and  $\mathsf{dist}(\cdot,\cdot)$ denotes a predefined distance metric, e.g., general $L_{p\in[1, \infty]}$ distances.

To ensure practical efficiency, most existing fuzzy PSI protocols rely on \textit{geometric separation assumptions} over the input sets. These works can be broadly divided into two categories based on the underlying assumptions: \textit{one-sided} and \textit{two-sided}. Protocols based on two-sided assumptions~\cite{van2024fuzzy,gao2025efficient,dang2025ccs,van2025,yang2026,zhang2025fast} generally achieve better efficiency. For example, under the assumption that points held by both the sender and receiver are pairwise separated by at least $2\delta$ or $4\delta$, the state-of-the-art construction~\cite{van2025} realizes concretely efficient fuzzy PSI using lightweight OPRF. However, such two-sided assumptions may be unrealistic in practice and can limit its applicability to real-world scenarios.

To relax these constraints, recent works~\cite{dang2025ccs,van2024fuzzy,fss25,garimella2024computation} have proposed fuzzy PSI protocols under \textit{one-sided assumptions}, where only one party's inputs are required to satisfy certain assumptions while the other party may hold arbitrary inputs. However, existing constructions either rely on heavyweight cryptographic primitives (e.g., AHE~\cite{paillier1999public,NR97}) or do not support general $L_{p\in[1, \infty]}$ distances.
Specifically, van Baarsen and Pu~\cite{van2024fuzzy}, later improved by Dang et al.~\cite{dang2025ccs}, proposed the first fuzzy PSI protocols for general $L_{p\in[1, \infty]}$ distances using AHE. Concurrently, another line of work~\cite{garimella2024computation,fss25} achieves fuzzy PSI using symmetric-key techniques (e.g., function secret sharing~\cite{boyle2015function,boyle2016function}), but supports only the $L_\infty$ distance.

Moreover, these works incur undesirable asymptotic communication and computation overhead in the distance threshold $\delta$. The protocol of~\cite{van2024fuzzy} scales linearly as $O(\delta)$, while the constructions in~\cite{dang2025ccs,fss25,garimella2024computation} incur computation or communication complexity of $O((\log \delta)^d)$. Consequently, the efficiency of these protocols deteriorates with larger $\delta$, particularly in real-world applications where $\delta$ represents a normalized integer of some high-precision floating-point value and may thus be large.

In summary, existing works for {general $L_{p\in[1, \infty]}$ distances} either (1) rely on strong geometric separation assumptions on both parties' inputs, or (2) incur substantial computation and communication overhead under one-sided assumptions. 
Motivated by these limitations, we raise the following question:

\begin{center}
        \begin{tcolorbox}[
            colback=lightgray!15, 
            size=title, 
            colframe=black, 
            boxrule=0.5pt,
            width=\linewidth,
            valign=center,
            ]
        
            Can we construct efficient fuzzy PSI protocols for general $L_{p\in[1, \infty]}$ distances under one-sided assumptions?
          
        \end{tcolorbox}
\end{center}

\subsection{Our Contributions}

To address these limitations, we propose a fuzzy PSI for general $L_{p\in[1,\infty]}$ distances with two key features: (1) it operates under weaker one-sided assumptions (for either sender or receiver), and (2) it relies solely on lightweight symmetric-key primitives, thus achieving significant improvements in both communication and computation efficiency. We summarize our contributions as follows.

\begin{enumerate}[leftmargin=*]
    \item \textbf{General fuzzy PSI under one-sided assumptions.} We propose concretely efficient fuzzy PSI protocols for general $L_{p\in[1,\infty]}$ distances under one-sided assumptions, relying solely on lightweight symmetric-key primitives. 
    {The core intermediate techniques are efficient protocols for multi-point fuzzy matching, which support one party holding multiple inputs without imposing any assumption. 
    Building on this, we present new constructions of fuzzy PSI supporting both the sender-sided and receiver-sided assumptions.}

    \item \textbf{Logarithmic complexity in threshold $\delta$.} We non-trivially incorporate prefix-trie techniques into our protocols. Unlike prior works that incur superlinear complexity $O((\log\delta)^d)$, we design an efficient dimension-by-dimension filtering mechanism to prune the exponential search space. As a result, our optimized protocol achieves logarithmic complexity $O(\log \delta)$ in the distance threshold $\delta$.

    \item \textbf{Extension for sparser distributions.} We further consider sparser input distributions under one-sided assumptions and design more efficient fuzzy PSI protocols tailored to this setting. The key idea is to reverse the roles of the sender and receiver in the spatial hashing procedure, which balances their respective workloads and yields improved concrete performance.
    
    \item \textbf{Implementation and evaluation.} We implement all proposed protocols and conduct extensive experiments across a wide range of parameter settings. Our protocols significantly outperform all prior works under assumptions of comparable strength, achieving one to three orders of magnitude improvement in efficiency.
\end{enumerate}

\section{Related Work}

{We introduce related works of fuzzy PSI under one-sided assumptions and defer other works to Appendix~\ref{appendix: other work}. Table~\ref{Tab: complexities} compares our protocols with prior works in terms of techniques, assumptions, and asymptotic complexity.}

Specifically, Garimella et al.~\cite{garimella2022structure} proposed the first fuzzy PSI protocol for $L_\infty$ distance, using a customized scheme of Function Secret Sharing (FSS)~\cite{boyle2015function,boyle2016function} to test whether the sender's points fall within the receiver's $\delta$-radius $L_\infty$ balls. Their construction requires the receiver's balls to be disjoint, while the sender's points could be arbitrary.
Depending on the distances between the ball centers, they consider different settings, but the computational complexity remains $O(\delta^d)$ in all cases.
Following works~\cite{garimella2024computation,fss25} improved upon this via prefix trie optimizations, reducing the complexity to $O((\log\delta)^d)$ for both communication and computation. These two constructions rely on the \textit{mini-universe} assumption, under which each ball maps uniquely to a distinct vertex of the space and the sender's points remain unconstrained. A common limitation of this line of work~\cite{garimella2022structure,garimella2024computation,fss25} is that it supports only $L_1$ and $L_\infty$ distances, with no known efficient instantiation for general $L_p$ distances.

Subsequently, van Baarsen and Pu~\cite{van2024fuzzy} studied fuzzy PSI under the one-sided $2\delta$-\textit{apart} assumption, where only the receiver's points are required to be at least $2\delta$ apart from each other. Their construction is built on a DDH-based set membership test, achieving $O(\delta)$ complexity in threshold $\delta$. To support arbitrary $L_p$ distances, they require a stronger $4\delta$-\textit{apart} assumption, and the overhead scales as $O(\delta^p)$ since the sender must enumerate all possible $p$-th power distances. Similarly, Dang et al.~\cite{dang2025ccs} further optimized this construction via prefix trie techniques under the same $4\delta$-\textit{apart} assumption, reducing the complexity to $O((\log\delta)^d)$.

{
More recently, Richardson et al.~\cite{richardson2024fuzzy} and Piske et al.~\cite{piske2025distance} proposed fuzzy PSI protocols based on symmetric-key techniques under the \textit{disjoint hash} assumption, which requires each point to map to a unique representation. In particular, Richardson et al.~\cite{richardson2024fuzzy} map elements into bins via spatial hashing and, assuming that each bin of both parties contains at most one element, compare items within corresponding bins. However, as illustrated by Piske et al.~\cite{piske2025distance}, their constructions based on Garbled Circuits incur substantial communication and computational overhead.
Subsequently, Piske et al.~\cite{piske2025distance} introduced more efficient fuzzy PSI protocols based on a new primitive called distance-aware OT, which transfers values conditioned on the distance between inputs. However, their theoretical construction under one-sided assumptions additionally leaks the density of the arbitrarily distributed receiver set and incurs complexity with an additional multiplicative dependence on this density. Moreover, their concrete instantiation ultimately relies on a two-sided assumption, requiring both the sender's and receiver's points to map to unique blocks or cells.
}

\begin{table*}[!t]
\centering
\caption{Comparisons of existing fuzzy PSI protocols under one-sided assumptions. The sender holds $m$ points and the receiver holds $n$ points in a $d$-dimensional space, and $\delta$ is the distance threshold. Note that our assumptions are either equivalent to or strictly weaker than existing ones; a detailed analysis is provided in Appendix~\ref{appendix: analysis of assumptions}.}

\vspace{-3pt}
\label{Tab: complexities}
\resizebox{\textwidth}{!}{%
\begin{tabular}{|c|c|c|c|c|cc|}
\hline
\multirow{2}{*}{Metric}     & \multirow{2}{*}{Protocol}            & \multirow{2}{*}{Technique} & \multirow{2}{*}{Assumption}                                             & \multirow{2}{*}{Communication}                                       & \multicolumn{2}{c|}{Computation}                                                                                         \\ \cline{6-7} 
                            &                                      &                                                                                     &                  &                                                    & \multicolumn{1}{c|}{Sender}                                              & Receiver                                      \\ \hline
\multirow{8}{*}{$L_\infty$} & \cite{fss25}                    & FSS & $\mathcal{R},\text{mini-universe}$                                                & $O(dn \log \delta + m2^d \textcolor{red}{(\log \delta)^d})$                            & \multicolumn{1}{c|}{$O(m2^d\textcolor{red}{(\log \delta)^d})$}                              & $O(\textcolor{red}{(\log \delta)^d} n + dm\log \delta)$        \\ \cline{2-7} 
                            & \multirow{2}{*}{\cite{van2024fuzzy}} & AHE &$\mathcal{R}, 2\delta\text{-apart}$                                                  & $O(\textcolor{red}{\delta} d n + 2^d m)$                                              & \multicolumn{1}{c|}{$O(2^dd m)$}                                         & $O(\textcolor{red}{\delta} d n + 2^d m)$                       \\
                            \cline{3-7}
                            &                                      & AHE &$\mathcal{R}, 4\delta\text{-apart}$                                                  & $O(\textcolor{red}{\delta}2^ dd n + m)$                                               & \multicolumn{1}{c|}{$O(dm)$}                                             & $O(\textcolor{red}{\delta}2^dd n + m)$                         \\ \cline{2-7} 
                            & \cite{dang2025ccs}                   & AHE & $\mathcal{R},4\delta\text{-apart}$                                                     & $O(d \log \delta (n 2^d + m))$                                       & \multicolumn{1}{c|}{$O(2^d d n \log \delta + m \textcolor{red}{(\log \delta)^d})$}        & $O(d m \log\delta)$                           \\ \cline{2-7} 
                            & \cite{richardson2024fuzzy}           & GC &$\mathcal{S}^*, \text{disj. hash}$                                                & $O(d \log \delta (n 2^s + m 2^{d-s}))$                               & \multicolumn{1}{c|}{$O(2^{d-s} dm \log \delta)$}                         & $O(2^{s} dn \log \delta)$                     \\ \cline{2-7} 
                            & \cite{piske2025distance}             & OPRF & $ \mathcal{S}^*$, disj. hash                                   & $O(d (\textcolor{red}{\delta} m+2^dn))$                                               & \multicolumn{1}{c|}{$O(\textcolor{red}{\delta} d m)$}                                     & $O(d 2^dn + m)$                               \\ \cline{2-7}   
                            & Ours                               & so-OPPRF & $\mathcal{R}/\mathcal{S}$, unique cell                & $O(d(\delta n+2^d  m))$ & \multicolumn{1}{c|}{$O(2^d d m)$}                              &    $O(d(\delta n+2^d  m))$                \\ \cline{2-7}   
                            & Ours-Px                               & so-OPPRF & $\mathcal{R}/\mathcal{S}$, unique cell                & $O(\textcolor{blue}{\log\delta}d(n+2^dm))$ & \multicolumn{1}{c|}{$O(\textcolor{blue}{\log\delta}2^ddm)$}                              &    $O(\textcolor{blue}{\log\delta}d(n+2^dm))$                \\
                            \cline{2-7}   
                            & Ours                               & so-OPPRF & $\mathcal{R}/\mathcal{S}$, unique block                & $O(d(m + \delta 2^dn))$  & \multicolumn{1}{c|}{$O(dm)$}                              &   $O(d(m + \delta 2^dn))$                 \\
                            \cline{2-7}   
                            & Ours-Px                                & so-OPPRF & $\mathcal{R}/\mathcal{S}$, unique block                  & $O(\textcolor{blue}{\log\delta}d(2^dn+m))$  & \multicolumn{1}{c|}{$O(\textcolor{blue}{\log\delta}dm)$}                              &   $O(\textcolor{blue}{\log\delta}d(2^dn+m))$                 \\
                            \cline{1-7} 
\multirow{6}{*}{$L_{p}$}    & \cite{van2024fuzzy}                  & AHE & $\mathcal{R}, 2\delta(d^{1/p}+1)\text{-apart}$                                       & $O(\textcolor{red}{\delta}{2^d} d n + \textcolor{red}{\delta^p} m)$                                    & \multicolumn{1}{c|}{$O((d + \textcolor{red}{\delta^p}) m)$}                               & $O(\textcolor{red}{\delta}{2^d} d n + m)$                      \\ \cline{2-7}  
                            & \cite{dang2025ccs}                   & AHE & $\mathcal{R},2\delta(d^{1/p}+1)\text{-apart}$                                          & $O(2^{d} p (d n \log \delta + m\textcolor{red}{(\log \delta)^{d}}) + d m\log \delta)$ & \multicolumn{1}{c|}{$O(2^{d} (d p n \log \delta + m\textcolor{red}{(\log \delta)^{d}}))$} & $O(m (d p \log \delta + 2^d\textcolor{red}{(\log \delta)^{d}}))$ \\ \cline{2-7} 
                            & \cite{richardson2024fuzzy}           & GC & $\mathcal{S}^*, \text{disj. hash}$                                                & $O(d \log(d\delta)(n 2^{d} + m 2^{d-s}))$                            & \multicolumn{1}{c|}{$O(d m 2^{d-s} \log(d\delta) )$}                     & $O(d n 2^{s} \log(d\delta))$                  \\ \cline{2-7} 
                            & \cite{piske2025distance}             & OPRF & $ \mathcal{S}^*$, disj. hash                                    & $O(d (\textcolor{red}{\delta} m+2^dn))$                                               & \multicolumn{1}{c|}{$O(\textcolor{red}{\delta} d m)$}                                     & $O(d 2^dn + m)$                               \\ \cline{2-7}   
                            & Ours                              & so-OPPRF & $\mathcal{R}/\mathcal{S}$, unique cell                & $O(d(\delta n+2^d  m)+p\log\delta 2^dm)$ & \multicolumn{1}{c|}{$O((d+p\log\delta) 2^dm)$}                              &      $O(d(\delta n+2^d  m)+p\log\delta 2^dm)$              \\ \cline{2-7}   
                            & Ours-Px                             & so-OPPRF & $\mathcal{R}/\mathcal{S}$, unique cell                & $O(\textcolor{blue}{\log\delta}(dn+2^ddm+2^dpm))$ & \multicolumn{1}{c|}{$O(\textcolor{blue}{\log\delta}(2^ddm+2^dpm))$}                              &     $O(\textcolor{blue}{\log\delta}(dn+2^ddm+2^dpm))$               \\
                            \cline{2-7}   
                            & Ours                               & so-OPPRF & $\mathcal{R}/\mathcal{S}$, unique block                  & $O(d(m + \delta 2^dn)+p\log\delta m )$ & \multicolumn{1}{c|}{$O((d+p\log\delta)m)$}                              &     $O(d(m + \delta 2^dn)+p\log\delta m )$               \\
                            \cline{2-7}   
                            & Ours-Px                               & so-OPPRF & $\mathcal{R}/\mathcal{S}$, unique block                  & $O(\textcolor{blue}{\log\delta}(2^ddn+dm+pm))$ & \multicolumn{1}{c|}{$O(\textcolor{blue}{\log\delta}(dm+pm))$}                              &    $O(\textcolor{blue}{\log\delta}(2^ddn+dm+pm))$                \\
                            \hline

\end{tabular}%
}
\captionsetup{justification=raggedright, singlelinecheck=false}
\caption*{\footnotesize 
-- $\mathcal{R}/\mathcal{S}$ denotes that the set of receiver or sender satisfies the assumption. 

-- $\mathcal{S}^*$ denotes that the protocols of \cite{piske2025distance,richardson2024fuzzy} require an additional assumption on the receiver's inputs to instantiate the concrete protocol.

-- Our protocols for sender-sided and receiver-sided assumptions have the same asymptotic complexity when $m = n$, we report one complexity.

-- Mini-universe assumption requires that each $L_\infty$ ball maps uniquely to a distinct mini-universe.

-- $2\delta$/$4\delta$-apart means the distance between any two points of the set is at least $2\delta$ or $4\delta$ apart. $2\delta(d^{1/p}+1)$-apart is the generalization of the $4\delta$-apart for $L_p$ distances.

-- Disj. hash is short for disjoint hash assumption, which means the spatial hashing scheme maps at most one point to the same cell.

-- Unique block/unique cell means in spatial hashing, each cell intersects with at most one $L_\infty$ ball and each cell contains at most one point, respectively.
}
\end{table*}

%% file: prelimilary.tex
\section{Preliminary}
\subsection{Notation}

We use $\kappa$ and $\lambda$ to denote the computational and statistical security parameters, respectively. We use $[a, b]$ to denote the set $\{a, \ldots, b\}$ and $[a]$ to denote the set $\{1, \ldots, a\}$. For a set $S$, $|S|$ denotes its cardinality. We write $r \overset{\$}{\leftarrow} \FF$ to denote that $r$ is sampled uniformly at random from space $\FF$. We use $\mathbf{1}\{\mathsf{event}\}$ to denote the indicator function that equals $1$ if $\mathsf{event}$ occurs and $0$ otherwise. All protocols in this work are proven secure in the semi-honest model and the formal definition is deferred to Appendix~\ref{appendix: Threat Model}.

\subsection{Fuzzy PSI}

We formally define the ideal functionality of fuzzy PSI in Figure~\ref{Func:FPSI}, following prior works~\cite{dang2025ccs,van2025,van2024fuzzy}. Given a threshold $\delta$ and a distance metric $\mathsf{dist}(\cdot,\cdot)$, the functionality takes as input a set from each party and outputs to the receiver all sender elements within distance $\delta$ of some receiver element.

\begin{figure}[!h]
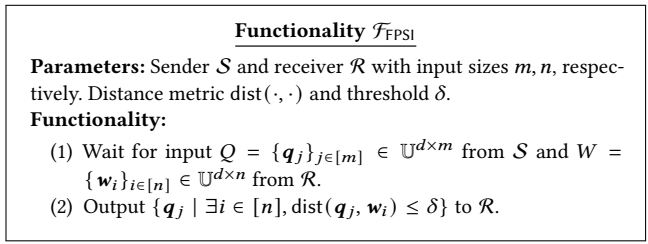

\begin{nffunc}{\Func[FPSI]}

\noindent \textbf{Parameters:} Sender \SSS and receiver \RRR with input sizes $m, n$, respectively. Distance metric $\mathsf{dist}(\cdot,\cdot)$ and threshold $\delta$.

\noindent \textbf{Functionality:}

\begin{enumerate}

    \item Wait for input $Q = \{\vecq_j\}_{j \in [m]} \in \mathbb{U}^{d\times m}$ from \SSS and $W = \{\vecw_i\}_{i \in [n]} \in \mathbb{U}^{d\times n}$ from \RRR.

    \item Output $\{\vecq_j \mid \exists i\in[n], \mathsf{dist}(\vecq_j, \vecw_i) \le \delta \}$ to \RRR.
    
\end{enumerate}

\end{nffunc}

\caption{Functionality of fuzzy PSI.}
\label{Func:FPSI}
\end{figure}

\subsection{Spatial Hashing and Assumptions}\label{sec:pre:assumptions}

\noindent\textbf{Spatial Hashing.} To avoid the quadratic complexity of naively comparing every pair of inputs, similar to prior fuzzy PSI under one-sided assumptions \cite{garimella2022structure,van2024fuzzy,van2025,dang2025ccs}, we adopt the spatial hashing technique, which partitions the space $\mathbb{U}^d$ into cells of side length $2\delta$. For a point $\vecx$, we define two algorithms:
\begin{enumerate}
    \item $\mathsf{cell}_{2\delta}(\vecx) \rightarrow \CCC$ maps $\vecx$ to the identifier of the cell containing it, where $\CCC = id_1 \Vert id_2 \Vert \cdots \Vert id_d$ and $id_k = \lfloor \frac{x_k}{2\delta} \rfloor$ for $k\in[d]$.
    \item $\mathsf{neigh}_{2\delta}(\vecx) \rightarrow \{\CCC_{z}\}_{z\in[2^d]}$ maps $\vecx$ to the set of identifiers of all cells intersecting the $L_{p\in[1,\infty]}$ ball of radius $\delta$ centered at $\vecx$. By default, the output is padded to $2^d$ cells using random values.
\end{enumerate}

\begin{figure}[!h]
    \centering
    \begin{tikzpicture}[
        every node/.style = {
            draw, rounded corners,
            minimum height = 5mm,
            inner sep = 2pt,
            align = center,
            font = \normalsize
        },
        arr/.style  = {Stealth-, semithick, shorten >=6pt, shorten <=6pt},
        cite/.style = {draw=none, minimum height=0pt,
                       inner sep=0pt, font=\small, text=black},
    ]
        \node (BC) at (0.0, 0)   {unique cell};
        \node (L)  at (2.85, 0) {unique block};
        \node (A) at (-3,  1.65) {mini-universe};
        \node (D) at (-3,  -0.55) {$2\delta$-apart};
        \node (E) at (-3, -1.65) {$2\delta$-disj.\ proj.};
        \node (F) at (2.85, -1.65) {$4\delta$-apart};
        \node (C) at (-3, 0.55)   {disj. hash}; 
        \draw[Stealth-Stealth, semithick, shorten >=6pt, shorten <=6pt] ([xshift=-3pt]A.south east) -- (BC.north west);
        \draw[arr] (D.east) -- ([xshift=-1pt, yshift=-3pt]BC.west);
        \draw[Stealth-Stealth, semithick, shorten >=6pt, shorten <=6pt] (C.east) -- ([xshift=-1pt, yshift=3pt]BC.west);
        \draw[arr] ([xshift=-2pt]E.north east) -- (BC.south west);
        \draw[arr] (F) -- (L);
        \draw[arr] (L.west) -- (BC.east);
        \node[cite, below=1pt of A]  {\cite{garimella2024computation,fss25}};
        \node[cite, below=1pt of C] {\cite{piske2025distance,richardson2024fuzzy}};
        \node[cite, above=2.8pt of BC] {{{Section~\ref{Sec: fuzzy psi 2delta}}}};
        \node[cite, below=1pt of D]  {\cite{van2024fuzzy,garimella2022structure}};
        \node[cite, below =1pt of E]  {\cite{van2024fuzzy,gao2025efficient,dang2025ccs,van2025,yang2026}};
        \node[cite, below=1pt of F]  {\cite{dang2025ccs,van2024fuzzy}};
        \node[cite, above=2.8pt of L] {{Section~\ref{Sec: 4delta}}};
    \end{tikzpicture}
    \vspace{-6pt}
    \caption{Relationships of existing assumptions. An arrow indicates the direction from the weaker to the stronger assumption, and a double-headed arrow indicates equivalence.}
    \label{Fig: assumptions}
\end{figure}
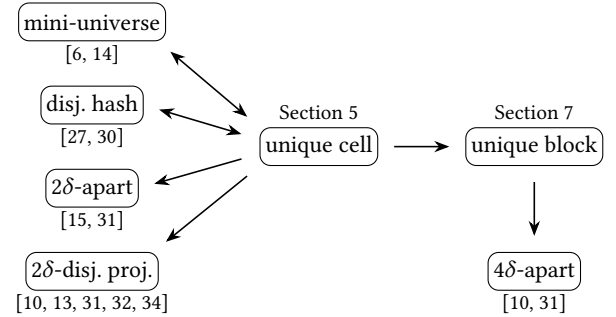

\noindent\textbf{Assumptions.} {
Our protocols operate under two different one-sided assumptions: \textit{unique cell} and \textit{unique block}. The \textit{unique cell} assumption requires that each cell contains at most one point, i.e., every point occupies a distinct cell. 

\begin{definition}[Unique Cell]\label{def: unique cell}
A set $W \in \mathbb{U}^{n \times d}$ satisfies the {unique cell} assumption if every cell in spatial hashing contains at most one point $\vecw \in W$. In other words, for any two distinct points $\vecw, \vecw' \in W$, $\mathsf{cell}_{2\delta}(\vecw) \neq \mathsf{cell}_{2\delta}(\vecw')$.
\end{definition}

The \textit{unique block} assumption is sparser, requiring that each cell intersects with at most one $L_{p\in[1,\infty]}$ ball, i.e., every point's $L_\infty$ ball occupies a distinct block of side length $4\delta$ comprising $2^d$ cells.

\begin{definition}[Unique Block]\label{def: unique block}
A set $W \in \mathbb{U}^{n \times d}$ satisfies the \textit{unique block} assumption if each cell intersects with at most
one $L_{p\in[1,\infty]}$ ball centered at $\vecw \in W$. In other words, for any two distinct points $\vecw, \vecw' \in W$, $\mathsf{neigh}_{2\delta}(\vecw) \cap \mathsf{neigh}_{2\delta}(\vecw') = \emptyset$.
\end{definition}

We provide a comparison of existing assumptions in Figure~\ref{Fig: assumptions}. For clarity, all assumptions are stated with respect to the $L_\infty$ distance; their relationships extend analogously to general $L_p$ distances. Notably, all existing assumptions in the literature are either equivalent to or strictly stronger than the \textit{unique cell} assumption. We defer the formal definitions and detailed analysis of these assumptions to Appendix~\ref{appendix: analysis of assumptions}.
}

\subsection{Prefix Trie}
\label{Sec: prefix-trie prelim}

Chakraborti et al.~\cite{chakraborti2023distance} first introduced the prefix trie technique, which was subsequently studied and formalized by~\cite{van2025,dang2025ccs,garimella2024computation,fss25}. Let $x = x_{\ell} x_{\ell-1} \cdots x_{1} \in \{0, 1\}^\ell$ be a binary string. We adopt the notation of~\cite{dang2025ccs}:
\begin{itemize}
    \item $\mathsf{UpBound}(x_{\ell} x_{\ell-1}\cdots x_{k})=x_{\ell} x_{\ell-1}\cdots x_{k} \| 11\cdots 1$
    \item $\mathsf{LowBound}(x_{\ell} x_{\ell-1}\cdots x_{k})=x_{\ell} x_{\ell-1}\cdots x_{k} \| 00\cdots 0$
    \item $\mathsf{Interval}(x_{\ell} x_{\ell-1}\cdots x_{k}) = \{x_{\ell}\cdots x_{k} \| x^*\}_{x^* \in \{0,1\}^{k-1}}$
\end{itemize}

van Baarsen and Pu~\cite{van2025} further summarized the technique and gave the following theorem.

\begin{theorem}[\cite{van2025}]
\label{theorem: prefix}
Given an integer interval $[q-\delta, q+\delta] \in \ZZ_{2^\ell}^{2\delta+1}$ and a point $w \in \ZZ_{2^\ell}$,
there are two algorithms with $\mu, \mu^\prime = O(\log \delta)$:

\begin{enumerate}
    \item $\mathsf{PxTrie}(q-\delta, q+\delta) \to \{\tilde{q}_1,\ldots,\tilde{q}_\mu\}$: The $\mathsf{PxTrie}$ algorithm succinctly encodes the interval $[q-\delta, q+\delta]$ into $\mu$ prefix nodes $\{\tilde{q}_1,\ldots,\tilde{q}_\mu\} \in \ZZ_{2^\ell}^\mu$ that cover the entire interval without overlap. Specifically, $\mu \le 2 + \log\delta$ when $\delta$ is a power of $2$. 

    \item $ \mathsf{PxPath}(w, \delta) \to \{\tilde{w}_1,\ldots,\tilde{w}_{\mu^\prime}\}$: The $\mathsf{PxPath}$ algorithm expands each query point $w \in \ZZ_{2^\ell}$ into a path of prefix nodes $\{\tilde{w}_1,\ldots,\tilde{w}_{\mu^\prime}\} \in \ZZ_{2^\ell}^{\mu^\prime}$, where $\mu^\prime := 2 + \log\delta$. 
\end{enumerate}
Then, it holds that $\tilde{w}_t \in \{\tilde{q}_1,\ldots,\tilde{q}_\mu\}$ for some unique $t \in [\mu^\prime]$ if and only if $w \in [q-\delta, q+\delta]$. 
\end{theorem}

For clarity, we assume throughout this paper that $\delta$ is a power of 2. We set $\mu$ to its maximum value, giving $\mu = \mu' = 2 + \log\delta$. The output of $\mathsf{PxTrie}(q-\delta, q+\delta)$ is always padded to $\mu$ prefixes using the default random values. Additionally, we give the following property: the outputs of $\mathsf{PxTrie}(\cdot,\cdot)$ and $\mathsf{PxPath}(\cdot)$ share at most one common element. A formal proof is provided in Appendix~\ref{appendix: proof of prefix}.

\begin{lemma}\label{lemma: prefix}
For any point $w \in \{0, 1\}^\ell$ and interval $I = [q-\delta, q+\delta]$, the following holds:
\[
\left| \mathsf{PxPath}(w,\delta) \cap \mathsf{PxTrie}(q-\delta,q+\delta) \right| =
\begin{cases}
1, & \text{if } w \in [q-\delta,q+\delta],\\
0, & \text{otherwise.}
\end{cases}
\]
\end{lemma}

\subsection{Oblivious Programmable PRF (with Secret-
shared Outputs)}
\label{Sec: pre so-opprf}

An oblivious programmable pseudorandom function (OPPRF)~\cite{kolesnikov2017practical} allows the sender to program a pseudorandom function with a set of key-value pairs. When evaluated at a programmed point, the receiver obtains the corresponding programmed value directly; otherwise, it obtains a uniformly random value.

\begin{figure}[!h]
\begin{nffunc}{\Func[so\text{-}OPPRF]}

\noindent \textbf{Parameters:} Sender \SSS and receiver \RRR with input sizes $m, n$, respectively. 

\noindent \textbf{Functionality:}

\begin{enumerate}

    \item Wait for input $L = \{(q_{j}, z_j)\}_{j \in [m]} \subseteq \UU \times \FF$ from \SSS and $X = \{x_i\}_{i \in [n]} \subseteq \UU$ from \RRR.

    \item Sample a random function $F : \UU \to \FF$ such that $F(q) = z$ for $(q, z) \in L$. 
    
    \item For $i \in [n]$, compute $y_i := F(x_i)$ and sample $y_i^\SSS, y_i^\RRR \from \FF$ such that $y_i^\SSS + y_i^\RRR = y_i \in \FF$.

    \item Output $\{y_i^\SSS\}_{i \in [n]}$ and $\OOO^{F}$ to \SSS and $\{y_i^\RRR\}_{i \in [n]}$ to \RRR.

\end{enumerate}

\end{nffunc}
\vspace{0.5em}
\caption{Functionality of oblivious programmable PRF with secret-shared outputs.}
\label{Func: ss-opprf}
\end{figure}

Yang et al.~\cite{yang2026} extended the standard OPPRF to support secret-shared outputs, yielding the so-OPPRF. Rather than revealing the outputs directly to the receiver, the so-OPPRF splits it into additive secret shares: when evaluated at a programmed point, both parties obtain shares of the corresponding programmed value; otherwise, they obtain uniformly random shares. The so-OPPRF can be efficiently instantiated from a secret-shared oblivious pseudorandom function (ss-OPRF)~\cite{alamati2024improved} and an oblivious key-value store (OKVS)~\cite{raghuraman2022blazing}. We give the ideal functionality of so-OPPRF in Figure~\ref{Func: ss-opprf}.

\subsection{Other Building Blocks}

Our protocols additionally rely on several standard functionalities, including interval test $\Func[Interval]$, private equality test $\Func[Eq]$, secure comparison $\Func[ssCMP]$, multiplexer \Func[MUX], and oblivious transfer $\Func[OT]$. We defer the formal functionalities to Appendix~\ref{appendix: func other}.

%% file: overview.tex
\section{Technical Overview}

We first introduce the notation used throughout this overview. Let $\mathcal{S}$ (the sender) hold a set $Q = \{\vecq_j\}_{j \in [m]}$ of size $m$, and let $\mathcal{R}$ (the receiver) hold a set $W = \{\vecw_i\}_{i \in [n]}$ of size $n$. Each element is a vector in a $d$-dimensional space, i.e., $\vecq_j, \vecw_i \in \mathbb{U}^d$. We denote by $q_{j,k}$ the $k$-th coordinate of vector $\vecq_j$. For clarity, we focus on the $L_\infty$ distance in this technical overview and detail how to extend to general $L_p$ distances in Section~\ref{Sec: fuzzy psi 2delta}.

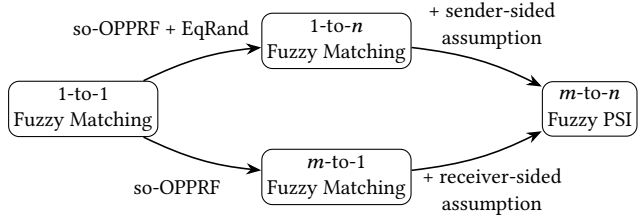
\begin{figure}[t]
  \centering
  \begin{tikzpicture}[
    every node/.style={
      draw, rounded corners,
      minimum height=5mm,
      inner sep=2pt,
      align=center,
      font=\small
    },
    arr/.style={-Stealth, semithick},
    lbl/.style={draw=none, inner sep=1pt, font=\small},
  ]
    \node (oto1) at (0,    0) {1-to-1 \\ Fuzzy Matching};
    \node (oton) at (3.35,  0.9) {1-to-$n$ \\ Fuzzy Matching};
    \node (mto1) at (3.35, -0.9) {$m$-to-1 \\ Fuzzy Matching};
    \node (mton) at (6.7,  0)   {$m$-to-$n$ \\ Fuzzy PSI};

    \draw[arr] (oto1) to[overlay, bend left=10]  node[lbl, above, xshift = -15pt, yshift = 3pt] {so-OPPRF + EqRand} (oton);
    \draw[arr] (oto1) to[overlay, bend right=10] node[overlay, lbl, below, xshift = -9pt, yshift = -3pt] {so-OPPRF} (mto1);
    \draw[arr] (oton) to[bend left=10]  node[lbl, above, xshift = 5pt, yshift=3pt] {+ sender-sided \\ assumption} (mton);
    \draw[arr] (mto1) to[bend right=10] node[overlay, lbl, below, xshift = 5pt, yshift=-3pt] {+ receiver-sided \\ assumption} (mton);

  \end{tikzpicture}
  \vspace{-2pt}
  \caption{Two approaches to fuzzy PSI, where $m$ is the sender input size and $n$ is the receiver input size. Both approaches only impose an assumption on one of the parties.}
  \label{Fig: approaches to fuzzy PSI}
\end{figure}

{
As shown in Figure~\ref{Fig: approaches to fuzzy PSI}, our starting point of the following technical overview is fuzzy matching (Section~\ref{overview: Revisiting OPPRF-Based Fuzzy Matching}), which can be viewed as a special case of fuzzy PSI in which each party holds a single element. Then, we show how to extend fuzzy matching to multi-point fuzzy matching (Section \ref{Sec: overview mq fuzzy matching}), where one party holds multiple points without any assumptions. 
We identify two possible methods of multi-point fuzzy matching to achieve fuzzy PSI: $1$-to-$n$ fuzzy matching (the upper part of Figure~\ref{Fig: approaches to fuzzy PSI}), where the sender holds a single point and the receiver holds $n$ arbitrary points, and $m$-to-$1$ fuzzy matching (the lower part of Figure~\ref{Fig: approaches to fuzzy PSI}), where the receiver holds a single point and the sender holds $m$ arbitrary points.
Finally, building on multi-point fuzzy matching, we realize fully fledged fuzzy PSI protocols under sender-sided or receiver-sided assumptions (Section \ref{Sec: overview fuzzy psi}). In addition, we further optimize our fuzzy PSI protocols using prefix-trie techniques (Section {\ref{Sec: overview prefix trie}}).
}

\subsection{Revisiting OPPRF-Based Fuzzy Matching}

\label{overview: Revisiting OPPRF-Based Fuzzy Matching}

{

In fuzzy matching, each party holds a single element, denoted by $\vecq$ and $\vecw$, respectively. The receiver learns $\vecq$ if $\mathsf{dist}(\vecq, \vecw) \le \delta$, and learns nothing otherwise.

Recently, van Baarsen and Pu~\cite{van2025} presented the state-of-the-art fuzzy matching protocol based on OPPRF (see Section~\ref{Sec: pre so-opprf}). At a high level, their key idea is to perform an interval membership test via OPPRF for each dimension and then aggregate the results across all dimensions. Concretely, for each dimension $k$, the sender samples a uniformly random value $s_k$ and programs all points in the interval $[q_k - \delta, q_k + \delta]$ to the same output $s_k$. Accordingly, the receiver then evaluates the programmed PRF at input $w_k$ to obtain $r_k$ for each $k \in [d]$.
Finally, the sender sends $\vecq \oplus H(\sum_{k \in [d]} s_k)$ to the receiver for an appropriate hash function $H$, and the receiver recovers $\vecq$ if and only if $\sum_{k \in [d]} r_k = \sum_{k \in [d]} s_k$. This follows because if $\mathsf{dist}(\vecq, \vecw) \le \delta$, then for every $k \in [d]$, we have $w_k \in [q_k - \delta, q_k + \delta]$, implying $r_k = s_k$ by correctness of OPPRF and thus $\sum_{k \in [d]} r_k = \sum_{k \in [d]} s_k$. Otherwise, if there exists some $k^* \in [d]$ such that $w_{k^*} \notin [q_{k^*} - \delta, q_{k^*} + \delta]$, then $r_{k^*} \neq s_{k^*}$ except with negligible probability and thus $\sum_{k \in [d]} r_k \neq \sum_{k \in [d]} s_k$.
Their scheme relies exclusively on symmetric-key operations and achieves substantial efficiency improvements over prior AHE-based approaches~\cite{van2024fuzzy,gao2025efficient}.}

{
However, the above construction cannot be extended to the setting where the receiver holds multiple points. 
For example, suppose the receiver holds two points $\vecw, \vecw'$ such that in some dimension~$k$, both $w_k, w_k' $ fall within the same interval $[q_k - \delta, q_k + \delta]$ programmed by the sender. 
By the correctness of OPPRF, the receiver will obtain the same $r_k$ on these two evaluations.
This leads to leakage when there is no match with $\vecq$:
the receiver can infer that the sender has encoded an interval covering both $w_k$ and $w_k'$ in dimension $k$.
Such a collision leaks the distribution of the sender's input.
To extend this OPPRF-based fuzzy matching to fuzzy PSI, where both the sender and receiver hold multiple points, \cite{van2025} prevents this leakage by imposing two-sided assumptions on both parties.
}

{
Alternative constructions~\cite{van2024fuzzy,dang2025ccs} address the above issue with expensive AHE, thereby keeping the OPPRF outputs hidden from the receiver. 
In the following, we show how to achieve multi-point fuzzy matching without relying on heavy public-key primitives.
}

\subsection{Multi-point Fuzzy Matching without Assumptions}
\label{Sec: overview mq fuzzy matching}

{We propose two variants of multi-point fuzzy matching: $1$-to-$n$ fuzzy matching, where the receiver holds $n$ arbitrarily distributed points, and $m$-to-$1$ fuzzy matching, where the sender holds $m$ arbitrarily distributed points.}

\subsubsection{$1$-to-$n$ Fuzzy Matching}\label{Sec: overview sender assumption}

{We first introduce our design for the $1$-to-$n$ fuzzy matching protocol, in which the receiver holds multiple arbitrary points without any additional assumptions. To address the partial-matching leakage arising from individual dimensions in~\cite{van2025}, our key insight is to secret-share the OPPRF output for each dimension so that neither party learns the OPPRF outputs in the clear. We instantiate this approach using shared-output OPPRF (so-OPPRF)~\cite{yang2026} (see Section~\ref{Sec: pre so-opprf}). Below, we first present a strawman construction that exhibits subtle privacy leakage, and then show how to eliminate it using additional techniques.}

{\textbf{Strawman Construction.} 
Let the sender hold a single point~$\vecq$ and the receiver hold multiple points $W = \{\vecw_i\}_{i \in [n]}$.
Similar to~\cite{van2025}, for each dimension $k$, the sender samples a random $s_k$ and programs the so-OPPRF with key-value pairs $\{(q_k + t, s_k)\}_{t \in [-\delta,\delta]}$.
Then, for each $\vecw_i \in W$ and $k \in [d]$, the receiver evaluates the so-OPPRF on input $w_{i,k}$, which outputs shares $r_{i,k}^\SSS$ and $r_{i,k}^\RRR$ to the two parties, respectively.
The secret-sharing form avoids the receiver inferring the sender's input from OPPRF outputs of some dimensions. 
After that, both parties locally aggregate their shares across all dimensions by computing $r_i^\SSS = \sum_{k \in [d]} r_{i,k}^\SSS$ and $r_i^\RRR = \sum_{k \in [d]} r_{i,k}^\RRR$, respectively.
Finally, the sender sends all $r_i^\SSS$ and $\vecq \oplus H(\sum_{k \in [d]} s_k)$ to the receiver for an appropriate hash function $H$.
The receiver subsequently reconstructs $r_i = r_i^\SSS + r_i^\RRR$ and recovers $\vecq$ if and only if there exists some $r_i = \sum_{k \in [d]} s_k$.}

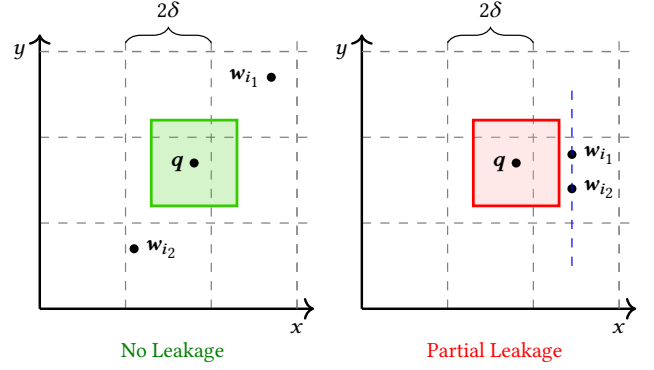
\begin{figure}[!t]
    \centering

    \begin{subfigure}[b]{0.235\textwidth}
        \centering
        \resizebox{0.95\textwidth}{!}{
        \begin{tikzpicture}[scale=0.98, font=\footnotesize]

            \draw[step=1cm, black!50, dashed] (0,0) grid (3.2,3.2);

            \draw[->, thick, black] (0,0) -- (3.2,0);
            \draw[->, thick, black] (0,0) -- (0,3.2);

            \draw[step=1cm, black!50, dashed] (0,3) -- (3,3);
            \draw[step=1cm, black!50, dashed] (3,0) -- (3,3);

            \node[overlay, below] at (3.0,-0.05) {$x$};
            \node[overlay, left]  at (-0.05,3.0) {$y$};

            \draw[decorate, decoration={brace, amplitude=5pt}, overlay]
                (1.0, 3.1) -- (2.0, 3.1)
                node[midway, above=5pt] {$2\delta$};

            \filldraw[fill={rgb,255:red,51;green,204;blue,0}, fill opacity=0.2,
                      draw={rgb,255:red,51;green,204;blue,0}, line width=0.9pt]
                (1.3, 1.2) rectangle (2.3, 2.2);

            \fill (1.8, 1.7) circle (1.5pt) node[left, opacity=1] {$\vecq$};
            \fill (2.7, 2.7) circle (1.5pt) node[left]            {$\vecw_{i_1}$};
            \fill (1.1, 0.7) circle (1.5pt) node[right]           {$\vecw_{i_2}$};

            \node[green!50!black] at (1.55,-0.5) {No Leakage};

        \end{tikzpicture}
        }
    \end{subfigure}
    \begin{subfigure}[b]{0.235\textwidth}
        \centering
        \resizebox{0.95\textwidth}{!}{
         \begin{tikzpicture}[scale=0.98, font=\footnotesize]

            \draw[step=1cm, black!50, dashed] (0,0) grid (3.2,3.2);

            \draw[->, thick, black] (0,0) -- (3.2,0);
            \draw[->, thick, black] (0,0) -- (0,3.2);

            \draw[step=1cm, black!50, dashed] (0,3) -- (3,3);
            \draw[step=1cm, black!50, dashed] (3,0) -- (3,3);

            \node[overlay, below] at (3.0,-0.05) {$x$};
            \node[overlay, left]  at (-0.05,3.0) {$y$};

            \draw[decorate, decoration={brace, amplitude=5pt}, overlay]
                (1.0, 3.1) -- (2.0, 3.1)
                node[midway, above=5pt] {$2\delta$};

            \filldraw[fill=red!30, fill opacity=0.3, draw=red, line width=0.9pt]
                (1.3, 1.2) rectangle (2.3, 2.2);

            \fill (1.8,   1.7 ) circle (1.5pt) node[left,  opacity=1] {$\vecq$};
            \fill (2.45, 1.4 ) circle (1.5pt) node[right]            {$\vecw_{i_2}$};
            \fill (2.45, 1.80) circle (1.5pt) node[right]            {$\vecw_{i_1}$};

            \draw[dashed, draw=blue] (2.45, 0.5) -- (2.45, 2.6);

            \node[red] at (1.55,-0.5) {Partial Leakage};

        \end{tikzpicture}
        }
    \end{subfigure}

    \vspace{-9pt}
    \caption{A counterexample for the strawman construction. In the right case, $\vecw_{i_1}$ and $\vecw_{i_2}$ share the same $x$-coordinate and both fall within the programmed interval of $\vecq$ along the $y$-dimension, and are therefore mapped to the same value.}
    \label{Fig: counter example}
\end{figure}

\textbf{Partial-matching Leakage.}\label{Sec: overview avoid leakage} {While the above construction appears correct, it introduces subtle privacy leakage. Specifically, although the so-OPPRF outputs remain secret-shared in each dimension, additional information is leaked once the aggregated values $r_i$ are revealed to the receiver. As illustrated in Figure~\ref{Fig: counter example} (right), consider a two-dimensional space in which two distinct receiver points $\vecw_{i_1}$ and $\vecw_{i_2}$ share the same $x$-coordinate and both lie within the programmed interval of $\vecq$ along the $y$-coordinate. During the so-OPPRF evaluation, the two points are mapped to the same random value for the $x$-coordinate. Since $\vecw_{i_1}$ and $\vecw_{i_2}$ fall within the same programmed interval along the $y$-dimension, they are mapped to the same programmed value as well. Consequently, after reconstructing from the shares $r^\SSS_{i_1}$ and $r^\SSS_{i_2}$, the receiver obtains two identical aggregate values, i.e., $r_{i_1} = r_{i_2}$. This collision reveals that there exists a sender point $\vecq$ that is close to both $\vecw_{i_1}$ and $\vecw_{i_2}$ in certain dimensions, even though neither receiver point actually matches $\vecq$.}

The root cause of this leakage is that the aggregated values $r_i$ for non-matching receiver points may coincide, revealing partial information about the sender's input. To eliminate this leakage, we observe that it suffices to re-randomize the secret shares of $r_i$ for non-matching inputs, while leaving those of matching inputs unchanged to preserve correctness.

\textbf{Enhance Security.} {To realize this, we introduce a new building block called equality-conditional randomization (\textsf{EqRand}), which randomizes values conditioned on an equality check. Concretely, it takes two pairs $(c_1, x_1)$ from the sender and $(c_2, x_2)$ from the receiver as input, and outputs random shares $(y_1, y_2)$ to the respective parties such that $y_1 + y_2 = x_1 + x_2$ if $c_1 = c_2$; otherwise, $y_1 + y_2$ is uniformly random.}

It remains to define an appropriate matching condition. To determine whether a receiver point $\vecw_i$ matches the sender point $\vecq$, we additionally program a zero indicator $\mathbf{0}$ together with the target value $s_{k}$, encoding each key-value pair as $\mathbf{0} \Vert s_{k}$. If $\vecw_i$ matches the sender point $\vecq$, it is mapped to $\mathbf{0}$ in every dimension, so the aggregated indicator remains $\mathbf{0}$. Otherwise, the aggregated indicator is uniformly random.

Specifically, the sender programs the so-OPPRF with key-value pairs\footnote{Here, we encode the dimension index $k$ into each key to allow a single so-OPPRF instance for all dimensions, though it can also be instantiated separately for each dimension as in~\cite{van2025}.} $\{(k \Vert q_{k}+t,\ \mathbf{0} \Vert s_{k})\}_{k\in[d],t\in[\text{-}\delta,\delta]}$, and the receiver evaluates it on each input $k \Vert w_{i,k}$ as before.
The two parties obtain packed shares $e_{i,k}^\RRR \Vert r_{i,k}^\RRR$ and $e_{i,k}^\SSS \Vert r_{i,k}^\SSS$, respectively, where $e_{i,k}^\RRR$ and $e_{i,k}^\SSS$ are shares of the indicator. They then aggregate these shares across all dimensions $k \in [d]$ to obtain $(e_{i}^\RRR, r_{i}^\RRR)$ and $(e_{i}^\SSS, r_{i}^\SSS)$. By the correctness of so-OPPRF, if $\vecw_i$ matches $\vecq$, then $e_{i}^\RRR$ and $e_{i}^\SSS$ are the shares of $\mathbf{0}$, i.e., $e_{i}^\RRR + e_{i}^\SSS = \mathbf{0}$. Otherwise, the shares $r_{i}^\RRR$ and $r_{i}^\SSS$ must be re-randomized to prevent collisions. The sender and receiver therefore invoke \textsf{EqRand} with inputs $(-e_{i}^\SSS, r_{i}^\SSS)$ and $(e_{i}^\RRR, r_{i}^\RRR)$, using the equality $e_{i}^\RRR = -e_{i}^\SSS$ as the matching condition, and obtain fresh shares $v_{i}^\RRR$ and $v_{i}^\SSS$, respectively.

Consequently, the sender can then send $v_{i}^\SSS$ without incurring leakage, and the receiver reconstructs $v_i = v_{i}^\RRR + v_{i}^\SSS$. Finally, as described above, the sender masks $\vecq$ under the key $s$ using a one-time pad, and the receiver attempts to unmask using each $v_i$, recovering $\vecq$ only when $\vecw_i$ matches $\vecq$.

\subsubsection{$m$-to-$1$ Fuzzy Matching}
\label{Sec: overview receiver assumption}

This mirrored path leads to a symmetric setting, where the sender holds multiple points $\vecq_j \in Q$ and the receiver holds a single point $\vecw$. Since our multi-point fuzzy matching primitive aims to be assumption-free, directly letting the sender program all its points into the so-OPPRF, as in the $1$-to-$n$ setting, may cause collisions among the key-value pairs. Our key insight is to switch the roles in the so-OPPRF routine: the sender evaluates the so-OPPRF instead of programming it.

This role switching yields a more natural construction. A key difference from the $1$-to-$n$ setting is that, for each sender point $\vecq_j$, the two parties obtain secret shares $e_{j}^\RRR$ and $e_{j}^\SSS$, which are sufficient to jointly determine whether $\vecq_j$ matches the receiver point $\vecw$. As a result, the random target value $s_k$ is no longer needed, and the receiver only programs the indicator $\mathbf{0}$.

Concretely, the receiver programs the so-OPPRF with key-value pairs $\{( k \Vert w_{k} + t,\ \mathbf{0})\}_{k \in [d],t \in [-\delta, \delta]}$, and the sender evaluates it on inputs $\{k \Vert q_{j,k}\}_{j \in [m], k \in [d]}$. The two parties obtain shares $e_{j,k}^\RRR$ and $e_{j,k}^\SSS$, and aggregate them across dimensions to obtain $e_{j}^\RRR = \sum_{k\in[d]} e_{j,k}^\RRR$ and $e_{j}^\SSS = \sum_{k\in[d]} e_{j,k}^\SSS$. By the correctness of so-OPPRF, the sender point $\vecq_j$ matches the receiver point if and only if $e_{j}^\RRR + e_{j}^\SSS = \mathbf{0}$, equivalently, $e_{j}^\RRR = -e_{j}^\SSS$. Then, transferring matching elements becomes considerably simpler. The two parties invoke the private equality test $\Func[Eq]$ to compare $e_{j}^\RRR$ and $-e_{j}^\SSS$, from which the receiver obtains a bit $b_{j}$. If $b_{j} = 0$, then $\vecq_j$ has no match; otherwise, $\vecq_j$ matches the receiver point $\vecw$. They then invoke $\Func[OT]$ with $b_j$ as the selection bit to transfer $\vecq_j$ to the receiver.

A notable advantage of this construction is that, since the target value $s_k$ is eliminated and the shares $e_{j}^\RRR$ and $e_{j}^\SSS$ remain private throughout, no secret value is ever revealed to either party. Consequently, invoking \textsf{EqRand} is unnecessary. This stands in contrast to the $1$-to-$n$ setting in Section~\ref{Sec: overview sender assumption}, where the aggregated shares are revealed to the receiver and \textsf{EqRand} is required to prevent the receiver from observing identical outputs across multiple inputs.

\subsection{Fuzzy PSI under One-sided Assumptions}
\label{Sec: overview fuzzy psi}

We now turn to the full $m$-to-$n$ fuzzy PSI construction. Having established two assumption-free multi-point fuzzy matching protocols in the preceding section, extending each to the $m$-to-$n$ setting follows the standard paradigm: imposing a separation assumption on the single-input party to allow it to program additional points without collision. To satisfy this requirement, we leverage spatial hashing techniques~\cite{dang2025ccs,van2024fuzzy,van2025}, partitioning the space into cells of side length $2\delta$ and assuming that each cell contains at most one point from the single-input party, while the other party's inputs remain unconstrained. We refer to this as the \textbf{unique cell assumption} (see Section~\ref{sec:pre:assumptions}). For each point, we define $\CCC^\SSS_j$ (resp.\ $\CCC^\RRR_i$) as the cell identifier of the cell containing $\vecq_j$ (resp.\ $\vecw_i$). 

Depending on which party satisfies the unique cell assumption, we distinguish two cases: the \textit{sender-sided assumption} and the \textit{receiver-sided assumption}, each giving rise to a corresponding $m$-to-$n$ fuzzy PSI construction.

\textbf{Sender-sided unique cell assumption.} In this case, each cell contains at most one sender point. Each sender point $\vecq_j$ is therefore associated with a unique cell identifier $\CCC^\SSS_j=\mathsf{cell}_{2\delta}(\vecq_j)$, which the sender concatenates with its programmed keys to build the key-value pairs as $\{ (\CCC^\SSS_j \Vert k \Vert q_{j,k}+t,\ \mathbf{0} \Vert s_{j,k}) \}_{j\in [m],k\in[d],t\in [\text{-}\delta,\delta]}$, ensuring that no two key-value pairs collide. Accordingly, as mentioned in~\cite{van2024fuzzy,van2025}, the receiver is supposed to evaluate the so-OPPRF at inputs $\{ \CCC^\RRR_{i,z} \Vert k \Vert w_{i,k} \}_{i\in [n],k\in[d],z\in [2^d]}$ to maintain correctness, where $\{\CCC^\RRR_{i,z}\}_{z\in [2^d]} = \mathsf{neigh}_{2\delta}(\vecw_i)$ (see Section~\ref{sec:pre:assumptions}) denotes the identifiers of all neighboring cells for $\vecw_i$. This is because two mutually matching points do not necessarily reside in the same cell; they may also fall into neighboring cells. Then, after evaluation, the two parties obtain packed shares $e_{i,k,z}^\RRR \Vert r_{i,k,z}^\RRR$ and $e_{i,k,z}^\SSS \Vert r_{i,k,z}^\SSS$, respectively. The remaining steps follow the same process in multi-point fuzzy matching.
For the detailed construction, we refer the reader to Section~\ref{Sec: sender assumption}.

\textbf{Receiver-sided unique cell assumption.} In this case, each cell contains at most one receiver point. Symmetrically, the receiver concatenates the cell identifier $\CCC^\RRR_i$ with its programmed keys to form the key-value pairs $\{ (\CCC^\RRR_i \Vert k \Vert w_{i,k}+t,\ \mathbf{0}) \}_{i\in [n],k\in[d],t\in [\text{-}\delta,\delta]}$, ensuring collision-freedom. The sender then evaluates the so-OPPRF at inputs $\{ \CCC^\SSS_{j,z} \Vert k \Vert q_{j,k} \}_{j\in [m],k\in[d],z\in [2^d]}$, where $\{\CCC^\SSS_{j,z}\}_{z\in [2^d]}$ denotes all candidate cell identifiers for $\vecq_j$. Then, after evaluation, the two parties obtain shares $e_{j,k,z}^\RRR$ and $e_{j,k,z}^\SSS$, respectively. The remaining steps also follow the same process in multi-point fuzzy matching.
For the detailed construction, we refer the reader to Section~\ref{Sec: receiver assumption}.

\subsection{Optimizations from Prefix Trie Techniques without Superlinear Overhead}
\label{Sec: overview prefix trie}

In this section, we describe how to non-trivially incorporate prefix-trie techniques~\cite{dang2025ccs,van2025} into our protocol to achieve $O(\log\delta)$ complexity in $\delta$, in contrast to the superlinear $O((\log\delta)^d)$ complexity incurred by prior works~\cite{dang2025ccs,van2025}. For clarity, we focus on the construction under the receiver-sided assumption; the sender-sided case follows analogously and is deferred to Appendix~\ref{appendix: other protocols for sender}.

Dang et al.~\cite{dang2025ccs} and van Baarsen and Pu~\cite{van2025} introduced an optimization that represents each interval $[w_{i,k}-\delta,\, w_{i,k}+\delta]$ via $\mu = O(\log \delta)$ prefixes, reducing the number of programmed points from $O(\delta)$ to $O(\log \delta)$. As a trade-off, however, the sender must evaluate the OPPRF along all $\mu$ prefix paths of $q_{j,k}$ to maintain correctness, obtaining outputs $\{e_{j,k,h}\}_{h\in[\mu]}$. For each dimension $k$, since the sender has no knowledge of which prefix is valid, it must exhaustively search all combinations $\sum_{k\in[d]} e_{j,k,h_k}$ for $h_k \in [\mu]$ and compare each against the target value, yielding a complexity\footnote{This can be reduced to $O((\log\delta)^{d/2})$ via the meet-in-the-middle technique, but the complexity still grows rapidly for large $\delta$.} of $O((\log \delta)^d)$ (or $O((2\log \delta)^d)$ for $L_p$ distance~\cite{dang2025ccs}).

\textbf{Dimension-by-dimension Filtering.}
Our key observation is that, for each dimension, at most one prefix can match the target interval, rendering all other prefixes unnecessary. This allows both parties to prune the search space on a per-dimension basis by discarding non-matching prefixes, rather than enumerating all prefix combinations across dimensions as in~\cite{van2025,dang2025ccs}. Since the exponential blowup in prior works arises precisely from taking the Cartesian product of per-dimension prefix sets, retaining only one candidate prefix per dimension collapses the search space from $O((\log\delta)^d)$ combinations to $O(d \cdot \log\delta)$.

Specifically, we denote the prefixes of the interval centered at $w_{i,k}$ as $\{w_{i,k,h}\}_{h\in[\mu]} = \mathsf{PxTrie}(w_{i,k}-\delta,\, w_{i,k}+\delta)$, and the prefix paths of $q_{j,k}$ as $\{q_{j,k,h}\}_{h\in[\mu]} = \mathsf{PxPath}(q_{j,k},\delta)$, where $\mu = O(\log \delta)$ denotes the number of prefixes. The sender evaluates the so-OPPRF at inputs $\CCC^\SSS_{j,z}\Vert k \Vert q_{j,k,h}$, and the parties obtain shares $e_{j,k,z,h}^\RRR$ and $e_{j,k,z,h}^\SSS$. According to Lemma~\ref{lemma: prefix}, if $q_{j,k}$ falls within $[w_{i,k}-\delta,\, w_{i,k}+\delta]$, there exists exactly one $h^*\in[\mu]$ such that $e_{j,k,z,h^*}^\RRR + e_{j,k,z,h^*}^\SSS = \mathbf{0}$.

To extract this unique index $h^*$ and drop out all other non-matching indices, we introduce a new primitive termed conditional selection $\FFF^{f,\mu}_{\mathsf{ConSel}}$ in {Figure~\ref{Func: eq sel}}. Concretely, it takes two tuples of $\mu$ secret-shared values $(e^\RRR_1,\ldots,e^\RRR_{\mu})$ and $(e^\SSS_1,\ldots,e^\SSS_{\mu})$ as input, and outputs two random shares $(y_1, y_2)$ such that $y_1 + y_2 = e^\RRR_{h^*} + e^\SSS_{h^*}$ if there exists a unique index $h^*\in[\mu]$ with $f(e^\RRR_{h^*} , e^\SSS_{h^*}) = 1$, and $y_1 + y_2$ is uniformly random otherwise. Here, we instantiate the condition function $f$ as an equality check $\mathsf{Eq}$ and use the equality $e_{j,k,z,h}^\RRR = -e_{j,k,z,h}^\SSS$ as the condition.

Consequently, for each dimension $k$, both parties input their per-prefix shares $\{e_{j,k,z,h}^\RRR\}_{h\in[\mu]}$ and $\{-e_{j,k,z,h}^\SSS\}_{h\in[\mu]}$ to $\FFF^{\mathsf{Eq},\mu}_{\mathsf{ConSel}}$ and receive a single pair of fresh shares $(v_{j,k,z}^\RRR, v_{j,k,z}^\SSS)$. Lemma~\ref{lemma: prefix} ensures that at most one prefix matches per dimension, so this functionality is well-defined. After this process, each dimension retains exactly one candidate prefix, and therefore yields just one combination across all dimensions, in contrast to the $O((\log \delta)^d)$ combinations in prior works.
The complexity of \Func[ConSel] is $O(\log\delta)$ for each dimension, and $O(d \cdot \log\delta)$ across all dimensions.

Given the shares $(v_{j,k,z}^\RRR, v_{j,k,z}^\SSS)$, with the prefix index eliminated, the remaining steps follow the same procedure as in the non-optimized construction of Section~\ref{Sec: overview fuzzy psi}.

%% file: method.tex
\section{Fuzzy PSI under Unique Cell Assumptions}
\label{Sec: fuzzy psi 2delta}

In this section, we present our protocols for general $L_{p\in[1,\infty]}$ distances under one-sided \textit{unique cell} assumptions.
That is, only one party (either the sender \textit{or} the receiver) is required to have at most one point in each cell, while the other party's set may follow an arbitrary distribution.

\begin{figure}[!h]
\begin{nffunc}{\ensuremath{\FFF^f_{\mathsf{ConRand}}}}

\noindent \textbf{Parameters:} Sender \SSS and Receiver \RRR. Condition predicate function $f(\cdot,\cdot)$ where $f$ can be an equality or a comparison function.

\noindent \textbf{Functionality:}

\begin{enumerate}

    \item Wait for input $(e^\SSS, v^\SSS)$ from \SSS and $(e^\RRR, v^\RRR)$ from \RRR.

    \item Compute $b:= f(e^\SSS,e^\RRR)$

    \item Generate random shares $z^\SSS, z^\RRR$ such that $z^\SSS + z^\RRR = v^\SSS + v^\RRR$ if $b=1$, and sample $z^\SSS, z^\RRR$ uniformly at random otherwise.

    \item Output $z^\SSS$ to \SSS and $z^\RRR$ to \RRR
    
\end{enumerate}

\end{nffunc}
\caption{Functionality of conditional randomization.}
\label{Func: cond rand}
\end{figure}

\subsection{Building Block}
\label{Sec: cond rand}

Before presenting our fuzzy PSI construction, we introduce a building block for preventing information leakage. Specifically, it takes two tuples $(e^\SSS, v^\SSS)$ and $(e^\RRR, v^\RRR)$ as input from the respective parties, and randomizes the values $v^\SSS$ and $v^\RRR$ depending on whether the conditions $e^\SSS$ and $e^\RRR$ satisfy a predefined criterion $f$. The ideal functionality of Conditional Randomization is given in Figure~\ref{Func: cond rand}.

\begin{figure}[!h]
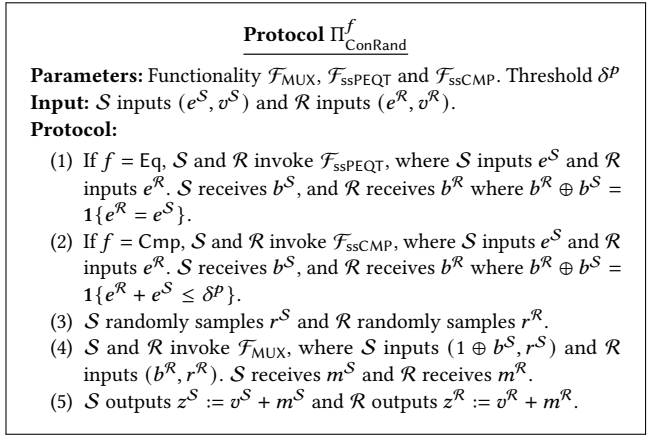

    \begin{nfprot}{\ensuremath{\Pi^f_{\mathsf{ConRand}}}}
    \noindent \textbf{Parameters:} Functionality \Func[MUX], \Func[ssPEQT] and \Func[ssCMP]. Threshold $\delta^p$
    
    \noindent \textbf{Input:} \SSS inputs $(e^\SSS, v^\SSS)$ and \RRR inputs $(e^\RRR, v^\RRR)$.
    
    \noindent \textbf{Protocol:}
    \begin{enumerate}
        \item If $f=\mathsf{Eq}$, \SSS and \RRR invoke \Func[ssPEQT], where \SSS inputs $e^\SSS$ and \RRR inputs $e^\RRR$. \SSS receives $b^\SSS$, and \RRR receives $b^\RRR$ where $b^\RRR \xor b^\SSS  = \mathbf{1}\{e^\RRR =  e^\SSS\}$.

        \item If $f=\mathsf{Cmp}$, \SSS and \RRR invoke \Func[ssCMP], where \SSS inputs $e^\SSS$ and \RRR inputs $e^\RRR$. \SSS receives $b^\SSS$, and \RRR receives $b^\RRR$ where $b^\RRR \xor b^\SSS  = \mathbf{1}\{e^\RRR+ e^\SSS \le \delta^p\}$. 

        \item \SSS randomly samples $r^\SSS$ and \RRR randomly samples $r^\RRR$.
        
        \item \SSS and \RRR invoke \Func[MUX], where \SSS inputs $(1 \xor b^\SSS,  r^\SSS)$ and \RRR inputs $(b^\RRR, r^\RRR)$. \SSS receives $m^\SSS$ and \RRR receives $m^\RRR$.

        \item \SSS outputs $z^\SSS := v^\SSS + m^\SSS$ and \RRR outputs $z^\RRR := v^\RRR + m^\RRR$.

    \end{enumerate}
    
    \end{nfprot}
    \vspace{0.5em}
    \caption{Protocol of conditional randomization.}
    \label{Prot: con rand}
\end{figure}

We realize this functionality efficiently via a customized secret-shared multiplexer protocol~\cite{rathee2020cryptflow2}. The conditional function $f(\cdot,\cdot)$ returns a boolean result based on a specified condition, and can be instantiated with either an equality or a comparison function. 
Depending on the choice of $f$, we refer to this functionality as Equality-conditional Randomization ($\mathsf{EqRand}$) or Comparison-conditional Randomization ($\mathsf{CmpRand}$) in the remainder of this paper.
The protocol is given in Figure~\ref{Prot: con rand} and we defer the proof in Appendix~\ref{appendix: proof con rand}.

\subsection{Protocols for Sender-sided Setting}
\label{Sec: sender assumption}

Since we have already presented the main idea for the $L_\infty$ distance in Section~\ref{Sec: overview avoid leakage} and Section~\ref{Sec: overview fuzzy psi}, we directly provide the detailed protocol in Figure~\ref{Prot: sender 2delta Linf}. Its security proof is deferred to Appendix~\ref{appendix: proof sender 2delta Linf}.

\begin{theorem}
\label{thm: sender 2delta Linf}
    The protocol $\Pi_\mathsf{FPSI}^{L_\infty}$ in Figure~\ref{Prot: sender 2delta Linf} securely realizes the functionality \Func[FPSI] for $L_\infty$ distance against semi-honest adversaries in the (\Func[so\text{-}OPPRF],\ensuremath{\FFF^{\mathsf{Eq}}_{\mathsf{ConRand}}})-hybrid model, assuming $H$ is a random oracle.
\end{theorem}

\textbf{$L_p$ distance.} We here focus on how to support arbitrary $L_p$ distances. Our key observation is the fact that $\mathsf{dist}_p(\vecw_i, \vecq_j) \le \delta$ implies $\mathsf{dist}_\infty(\vecw_i, \vecq_j) \le \delta$, and we therefore could reuse most parts of the $L_\infty$ construction with only minor modifications.
In the $L_\infty$ setting, the sender programs the value $0 \Vert s_{j,k}$, where $0$ serves as an indicator of whether $\vecw_i$ is close to some $\vecq_j$, as described in Section~\ref{Sec: overview avoid leakage}. For $L_p$ distance, the sender instead encodes $|t|^p \Vert s_{j,k}$, where $|t|^p$ captures the per-dimension distance contribution when $w_{i,k}- q_{j,k} = t$ for some $t \in [-\delta, \delta]$. The receiver proceeds as before, evaluating the so-OPPRF at $\CCC^{\RRR}_{i,z} \Vert k \Vert w_{i,k}$. After aggregating across all dimensions, $d^\RRR_{i,z}$ and $d^\SSS_{i,z}$ are secret shares of the $p$-th power of the distance between $\vecw_i$ and the candidate $\vecq_j$, provided they are potentially matched. 
To avoid the similar leakage mentioned in Section~\ref{Sec: overview avoid leakage}, we invoke the $\mathsf{CmpRand}$ to randomize shares $r^\RRR_{i,z}$ and $r^\SSS_{i,z}$ depending on the comparison condition $d^\RRR_{i,z} + d^\SSS_{i,z} \le \delta^p$. The remaining steps are identical to those in the $L_\infty$ setting, allowing the receiver to obtain all valid matches. 
We defer the detailed protocol in Figure~\ref{Prot: sender 2delta Lp} of Appendix~\ref{appendix: other sender 2delta Lp}.

\begin{figure}[t]
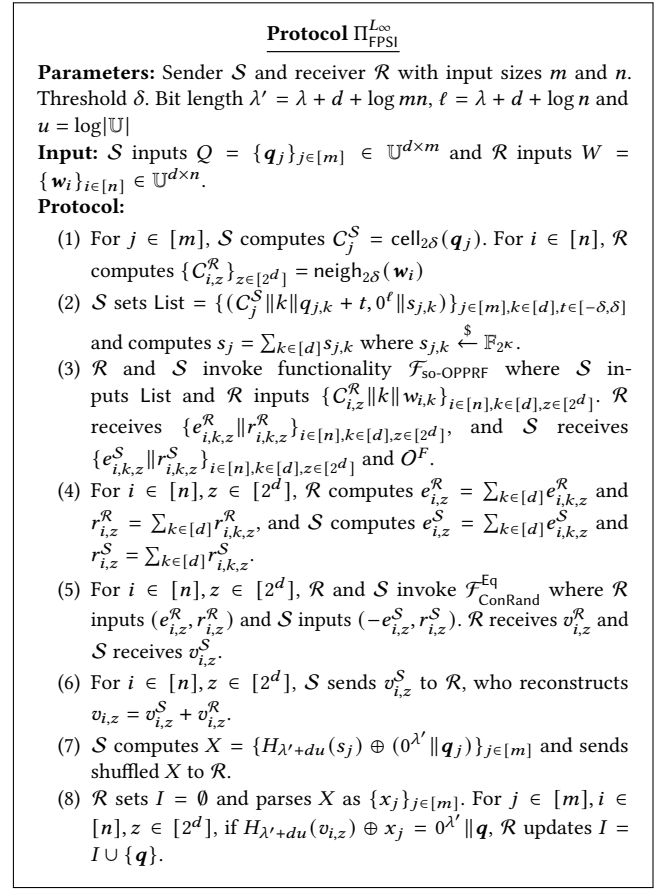

    \begin{nfprot}{\ensuremath{\Pi_\mathsf{FPSI}^{L_\infty}}\xspace}
        \noindent \textbf{Parameters:}  Sender \SSS and receiver \RRR with input sizes $m$ and $n$. Threshold \( \delta \). Bit length $\lambda^\prime = \lambda + d + \log mn$, $\ell = \lambda +d+ \log n$ and $u = \log\abs{\UU}$

        \noindent \textbf{Input:} \SSS inputs \( Q=\{\vecq_j\}_{j\in [m]} \in \mathbb{U}^{d\times m} \) and \RRR inputs \( W=\{\vecw_i\}_{i\in [n]} \in \mathbb{U}^{d\times n} \).

        \noindent \textbf{Protocol:}
        \begin{enumerate}
            \item For $j\in[m]$, \SSS computes $\CCC_j^\SSS = \mathsf{cell}_{2\delta}(\vecq_j)$. For $i\in [n]$, \RRR computes $\{ \CCC^{\RRR}_{i,z}\}_{z\in [2^d]} = \mathsf{neigh}_{2\delta}(\vecw_i)$
        
            \item \SSS sets $\mathsf{List} = \{ (\CCC^{\SSS}_{j}\Vert k \Vert q_{j,k} + t, {0^\ell} \Vert s_{j,k})\}_{j\in[m],k\in[d], t\in [-\delta,\delta]}$ and computes $s_{j} = {\sum}_{k\in[d]}s_{j,k}$ where $s_{j,k} \overset{\$}{\leftarrow} \FF_{2^\kappa}$.

            \item \RRR and \SSS invoke functionality \Func[so\text{-}OPPRF] where \SSS inputs $\mathsf{List}$ and \RRR inputs $\{\CCC^{\RRR}_{i,z} \Vert k \Vert w_{i,k}\}_{i\in[n],k\in[d],z\in[2^d]}$. \RRR receives $\{e^\RRR_{i,k,z}\Vert r^\RRR_{i,k,z}\}_{i\in[n],k\in[d],z\in[2^d]}$, and \SSS receives $\{e^\SSS_{i,k,z}\Vert r^\SSS_{i,k,z} \}_{i\in[n],k\in[d],z\in[2^d]}$ and $\OOO^F$.
            
            \item For $i\in[n],z\in[2^d]$, \RRR computes $e^\RRR_{i,z}={\sum}_{k\in[d]}e^\RRR_{i,k,z}$ and $r^\RRR_{i,z}={\sum}_{k\in[d]}r^\RRR_{i,k,z}$, and \SSS computes $e^\SSS_{i,z}={\sum}_{k\in[d]}e^\SSS_{i,k,z}$ and $r^\SSS_{i,z}={\sum}_{k\in[d]}r^\SSS_{i,k,z}$.

            \item For $i\in [n],z\in [2^d]$, \RRR and \SSS invoke \ensuremath{\FFF^{\mathsf{Eq}}_{\mathsf{ConRand}}} where \RRR inputs $(e^\RRR_{i,z}, r^\RRR_{i,z})$ and \SSS inputs $(-e^\SSS_{i,z}, r^\SSS_{i,z})$. \RRR receives $v^\RRR_{i,z}$ and \SSS receives $v^\SSS_{i,z}$.

            \item For $i\in[n],z\in[2^d]$, \SSS sends $v^\SSS_{i,z}$ to \RRR, who reconstructs $v_{i,z} = v^\SSS_{i,z} + v^\RRR_{i,z}$.

            \item \SSS computes $X = \{ H_{\lambda^\prime+du}(s_j)\oplus(0^{\lambda^\prime}\Vert \vecq_j)\}_{j\in[m]}$ and sends {shuffled} $X$ to \RRR.

            \item \RRR sets $I = \emptyset$ and parses $X$ as $\{ x_{j}\}_{j\in [m]}$.  For $j\in [m],i\in[n],z\in[2^d]$, if $H_{\lambda^\prime+du}(v_{i,z}) \oplus x_{j} = 0^{\lambda^\prime} \Vert \vecq$, \RRR updates $I = I \cup \{\vecq\}$.
            
        \end{enumerate}

    \end{nfprot}
    \vspace{0.5em}
    \caption{Protocol of fuzzy PSI for $L_\infty$ distance under sender-sided unique cell assumption.}
    \label{Prot: sender 2delta Linf}
\end{figure}

\subsection{Protocols for Receiver-sided Setting}
\label{Sec: receiver assumption}

Since we have presented our idea in Section~\ref{Sec: overview receiver assumption} and Section~\ref{Sec: overview fuzzy psi}, we directly give the protocol under the receiver-sided assumption for $L_\infty$ distance in Figure~\ref{Prot: receiver 2delta Linf}. The security proof is deferred to Appendix~\ref{appendix: proof receiver 2delta Linf}.

\begin{theorem}
\label{thm: receiver 2delta Linf}
    The protocol $\Pi_\mathsf{FPSI}^{L_\infty}$ in Figure~\ref{Prot: receiver 2delta Linf} securely realizes the functionality \Func[FPSI] for $L_\infty$ distance against semi-honest adversaries in the (\Func[so\text{-}OPPRF],\Func[Eq],\Func[OT])-hybrid model.
\end{theorem}

\textbf{$L_p$ distance.} We next show how to extend the protocol to $L_p$ distance.
Similarly, the receiver now programs $|t|^p$ in place of $0$ as in the $L_\infty$ setting. As noted in Section~\ref{Sec: overview receiver assumption}, invoking conditional randomization is unnecessary under the receiver-sided assumption, since the receiver does not reconstruct $r_{j,z}$. By the correctness of \Func[so\text{-}OPPRF], $r^\RRR_{j,z}$ and $r^\SSS_{j,z}$ are secret shares of the $p$-th power of the distance between $\vecw_i$ and the candidate $\vecq_j$. It therefore suffices to test whether $r^\RRR_{j,z} + r^\SSS_{j,z}$ falls within interval $[0, \delta^p]$. A natural approach would be to invoke the secure comparison \Func[ssCMP]; however, since the output need not be secret-shared, we instead realize this check via an interval test (see Appendix~\ref{appendix: func other}), following~\cite{chakraborti2023distance,van2025}.

Specifically, the two parties invoke functionality $\FFF^{\delta^p}_{\mathsf{Interval}}$, and the receiver obtains outputs $b_{j,z} = \mathbf{1}\{r^\RRR_{j,z} + r^\SSS_{j,z} \le \delta^p\}$. Since $q_j$ may match multiple receiver elements, the receiver aggregates the per-cell results as $b_j = \bigvee_{z\in[2^d]} b_{j,z}$. Finally, $b_j$ serves as the choice bit, and the sender transfers the matching elements to the receiver via an OT protocol. 
We defer the detailed protocol in Figure~\ref{Prot: receiver 2delta Lp} of Appendix~\ref{appendix: other receiver 2delta Lp}.

\begin{figure}[t]
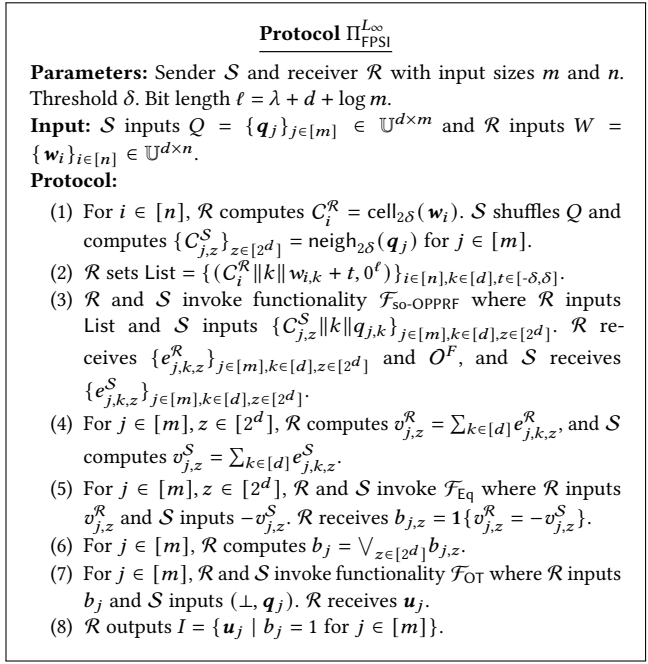

    \begin{nfprot}{\ensuremath{\Pi_\mathsf{FPSI}^{L_\infty}}\xspace}
        \noindent \textbf{Parameters:}  Sender \SSS and receiver \RRR with input sizes $m$ and $n$. Threshold \( \delta \). Bit length $\ell = \lambda + d+\log m$.

        \noindent \textbf{Input:} \SSS inputs \( Q=\{\vecq_j\}_{j\in [m]} \in \mathbb{U}^{d\times m} \) and \RRR inputs \( W=\{\vecw_i\}_{i\in [n]} \in \mathbb{U}^{d\times n} \).

        \noindent \textbf{Protocol:}
        \begin{enumerate}
            
            \item For $i\in[n]$, \RRR computes $\CCC_i^\RRR = \mathsf{cell}_{2\delta}(\vecw_i)$. \SSS shuffles $Q$ and computes $\{ \CCC^{\SSS}_{j,z}\}_{z\in [2^d]} = \mathsf{neigh}_{2\delta}(\vecq_j)$ for $j\in [m]$.
            
            \item \RRR sets $\mathsf{List} = \{(\CCC^{\RRR}_{i}\Vert k \Vert w_{i,k} + t,0^\ell)\}_{i\in[n],k\in[d],t\in[\text{-}\delta,\delta]}$.
            
            \item \RRR and \SSS invoke functionality \Func[so\text{-}OPPRF] where \RRR inputs $\mathsf{List}$ and \SSS inputs $\{\CCC^{\SSS}_{j,z} \Vert k \Vert q_{j,k}\}_{j\in[m],k\in[d],z\in [2^d]}$. \RRR receives $\{e^\RRR_{j,k,z}\}_{j\in[m],k\in[d],z\in [2^d]}$ and $\OOO^F$, and \SSS receives $\{e^\SSS_{j,k,z}\}_{j\in[m],k\in[d],z\in [2^d]}$.

            \item For $j\in [m],z\in [2^d]$, \RRR computes $v^\RRR_{j,z} = {\sum}_{k\in[d]} e^\RRR_{j,k,z}$, and \SSS computes $v^\SSS_{j,z} = {\sum}_{k\in[d]} e^\SSS_{j,k,z}$.

            \item For $j\in [m],z\in[2^d]$, \RRR and \SSS invoke \ensuremath{\FFF_{\mathsf{Eq}}} where \RRR inputs $v^{\RRR}_{j,z}$ and \SSS inputs $-v^{\SSS}_{j,z}$. \RRR receives $b_{j,z} = \mathbf{1}\{ v^{\RRR}_{j,z} = -v^{\SSS}_{j,z} \}$.

            \item For $j\in [m]$, \RRR computes $b_j = {\bigvee}_{z\in[2^d]}b_{j,z}$.

            \item For $j\in [m]$, \RRR and \SSS invoke functionality \Func[OT] where \RRR inputs $b_j$ and \SSS inputs $(\bot,\vecq_j)$. \RRR receives $\vecu_j$.

            \item \RRR outputs $I = \{ \vecu_j \; \vert \; b_j=1 \;\text{for} \;  j\in[m] \}$.

        \end{enumerate}

    \end{nfprot}
    \vspace{0.5em}
    \caption{Protocol of fuzzy PSI for $L_\infty$ distance under receiver-sided unique cell assumption.}
    \label{Prot: receiver 2delta Linf}
\end{figure}

\section{Optimized Fuzzy PSI from Prefix Trie}

In this section, we optimize our protocols to achieve $O(\log\delta)$ complexity in the threshold $\delta$ via prefix trie techniques.

We introduce a core building block that resolves the exponential complexity arising from the incorporation of prefix trie techniques. The ideal functionality is given in Figure~\ref{Func: eq sel}. Functionality $\FFF^{f,\mu}_{\mathsf{ConSel}}$ takes as input two sets of $\mu$ values and selects a unique pair based on a condition $f$. Specifically, if there exists a unique index $i$ such that $f(e^\SSS_i, e^\RRR_i) = 1$, it outputs fresh random shares summing to $e^\SSS_i + e^\RRR_i$; otherwise, it outputs uniformly random shares. For each dimension, this functionality acts as a filter, extracting only one prefix and discarding all non-matching prefixes, thereby avoiding exponential complexity. It is worth noting that this functionality is predicated on there being at most one pair of inputs satisfying $f$, as it produces only a single pair of outputs. The uniqueness of the matching prefix is established in Lemma~\ref{lemma: prefix}, ensuring that this functionality can be safely adopted without compromising correctness when we instantiate the $f$ as an equality test function $\mathsf{Eq}$ in our protocol. The concrete construction is presented in Figure~\ref{Prot: eq sel}, and we refer to it as Equality-conditional Selection ($\mathsf{EqSel}$). Over Boolean shares, $\mathsf{EqSel}$ selects precisely the shares of zero. 
When $\mu=1$, this functionality reduces to a special case of $\mathsf{EqRand}$ in which the value is the condition itself. Its correctness follows that of $\FFF^{f}_{\mathsf{ConRand}}$, and we defer the security proof to Appendix~\ref{appendix: proof con sel}.

\begin{figure}[t]
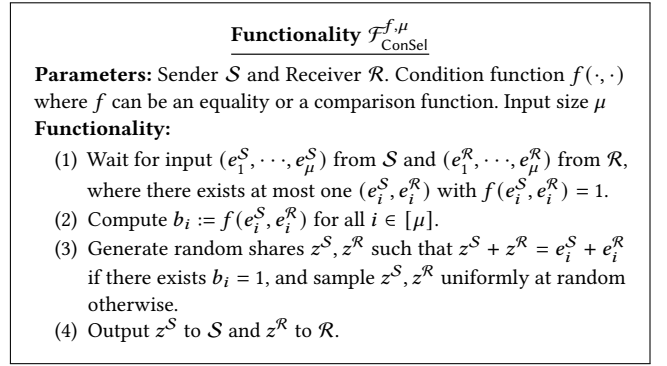

\begin{nffunc}{\ensuremath{\FFF^{f,\mu}_{\mathsf{ConSel}}}}

\noindent \textbf{Parameters:} Sender \SSS and Receiver \RRR. Condition function $f(\cdot,\cdot)$ where $f$ can be an equality or a comparison function. Input size $\mu$

\noindent \textbf{Functionality:}

\begin{enumerate}

    \item Wait for input $(e^\SSS_1, \cdots\!,e^\SSS_\mu)$ from \SSS and $(e^\RRR_1, \cdots\!,e^\RRR_\mu)$ from \RRR, where there exists at most one $(e^\SSS_i,e^\RRR_i)$ with $f(e^\SSS_i,e^\RRR_i) = 1$.

    \item Compute $b_i:= f(e^\SSS_i,e^\RRR_i)$ for all $i\in[\mu]$.

    \item Generate random shares $z^\SSS, z^\RRR$ such that $z^\SSS + z^\RRR = e_i^\SSS + e_i^\RRR$ if there exists $b_i=1$, and sample $z^\SSS, z^\RRR$ uniformly at random otherwise.

    \item Output $z^\SSS$ to \SSS and $z^\RRR$ to \RRR.
    
\end{enumerate}

\end{nffunc}
\caption{Functionality of conditional selection.}
\label{Func: eq sel}
\end{figure}

\begin{figure}[t]
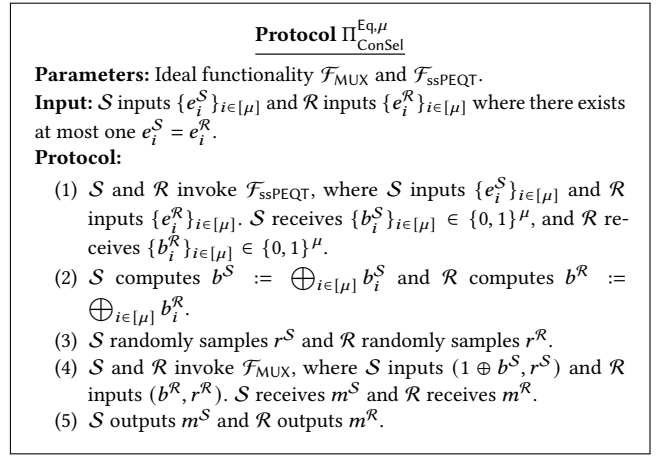

    \begin{nfprot}{$\Pi^{\mathsf{Eq},\mu}_{\mathsf{ConSel}}$}
    \noindent \textbf{Parameters:} Ideal functionality \Func[MUX] and \Func[ssPEQT].
    
    \noindent \textbf{Input:} \SSS inputs $\{e_i^\SSS\}_{i \in [\mu]}$ and \RRR inputs $\{e_i^\RRR\}_{i \in [\mu]}$ where there exists at most one $e_i^\SSS=e_i^\RRR$.
    
    \noindent \textbf{Protocol:}
    \begin{enumerate}
        \item \SSS and \RRR invoke \Func[ssPEQT], where \SSS inputs $\{e_i^\SSS\}_{i \in [\mu]}$ and \RRR inputs $\{e_i^\RRR\}_{i \in [\mu]}$. \SSS receives $\{b_i^\SSS\}_{i \in [\mu]} \in \bool^\mu$, and \RRR receives $\{b_i^\RRR\}_{i \in [\mu]} \in \bool^\mu$.

        \item \SSS computes $b^\SSS := \Xor_{i \in [\mu]} b_i^\SSS$ and \RRR computes $b^\RRR := \Xor_{i \in [\mu]} b_i^\RRR$.
        
        \item \SSS randomly samples $r^\SSS$ and \RRR randomly samples $r^\RRR$.
        
        \item \SSS and \RRR invoke \Func[MUX], where \SSS inputs $(1 \xor b^\SSS, r^\SSS)$ and \RRR inputs $(b^\RRR, r^\RRR)$. \SSS receives $m^\SSS$ and \RRR receives $m^\RRR$.

        \item \SSS outputs $m^\SSS$ and \RRR outputs $m^\RRR$.

    \end{enumerate}
    
    \end{nfprot}
    \vspace{0.5em}
    \caption{Protocol of equality-conditional selection.}
    \label{Prot: eq sel}
\end{figure}

We have presented our idea to incorporate prefix trie in Section~\ref{Sec: overview prefix trie}, and the prefix-optimized protocol under the receiver-sided assumption for $L_\infty$ distance is illustrated in Figure~\ref{Prot: receiver 2delta Linf px}. We defer
the security proof in Appendix~\ref{appendix: proof receiver 4delta Linf prefix}. The protocols for $L_p$ distances can be optimized analogously, and we defer the details to Appendix~\ref{appendix: other protocols}.

\begin{theorem}
\label{thm: receiver 2delta px Linf}
    The protocol $\Pi_\mathsf{FPSI\text{-}Px}^{L_\infty}$ in Figure~\ref{Prot: receiver 2delta Linf px} securely realizes the functionality \Func[FPSI] for $L_\infty$ distance against semi-honest adversaries in the (\Func[so\text{-}OPPRF],$\FFF^{\mathsf{Eq},\mu}_{\mathsf{ConSel}}$,\Func[Eq],\Func[OT])-hybrid model.
\end{theorem}

\begin{figure}[t]
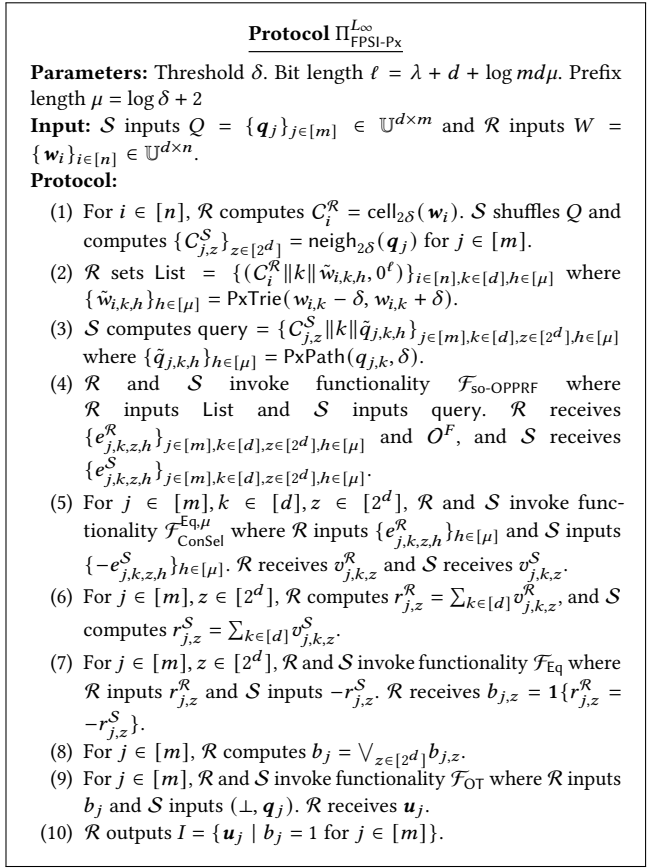

    \begin{nfprot}{\ensuremath{\Pi_\mathsf{FPSI\text{-}Px}^{L_\infty}}\xspace}
        \noindent \textbf{Parameters:}  Threshold \( \delta \). Bit length $\ell = \lambda + d + \log md\mu$. Prefix length $\mu = \log\delta+2$

        \noindent \textbf{Input:} \SSS inputs \( Q=\{\vecq_j\}_{j\in [m]} \in \mathbb{U}^{d\times m} \) and \RRR inputs \( W=\{\vecw_i\}_{i\in [n]} \in \mathbb{U}^{d\times n} \).

        \noindent \textbf{Protocol:}
        \begin{enumerate}
            \item For $i\in[n]$, \RRR computes $\CCC_i^\RRR = \mathsf{cell}_{2\delta}(\vecw_i)$. \SSS shuffles $Q$ and computes $\{ \CCC^{\SSS}_{j,z}\}_{z\in [2^d]} = \mathsf{neigh}_{2\delta}(\vecq_j)$ for $j\in [m]$.
            
            \item \RRR sets $\mathsf{List} = \{(\CCC^{\RRR}_{i}\Vert k \Vert \tilde{w}_{i,k,h},0^\ell)\}_{i\in[n],k\in[d],h\in [\mu]}$ where $\{\tilde{w}_{i,k,h}\}_{h\in[\mu]} = \textsf{PxTrie}(w_{i,k}-\delta,w_{i,k}+\delta)$.

            \item \SSS computes $\mathsf{query} = \{\CCC^{\SSS}_{j,z} \Vert k \Vert \tilde{q}_{j,k,h}\}_{j\in[m],k\in[d],z\in [2^d],h\in[\mu]}$ where $\{\tilde{q}_{j,k,h} \}_{h\in[\mu]} = \mathsf{PxPath}(q_{j,k},\delta)$.
            
            \item \RRR and \SSS invoke functionality \Func[so\text{-}OPPRF] where \RRR inputs $\mathsf{List}$ and \SSS inputs $\mathsf{query}$. \RRR receives $\{e^\RRR_{j,k,z,h}\}_{j\in[m],k\in[d],z\in [2^d],h\in[\mu]}$ and $\OOO^F$, and \SSS receives $\{e^\SSS_{j,k,z,h}\}_{j\in[m],k\in[d],z\in [2^d],h\in[\mu]}$.

            \item For $j\in [m],k\in [d],z\in[2^d]$, \RRR and \SSS invoke functionality $\FFF^{\mathsf{Eq},\mu}_{\mathsf{ConSel}}$ where \RRR inputs $\{e^\RRR_{j,k,z,h}\}_{h\in[\mu]}$ and \SSS inputs $\{-e^\SSS_{j,k,z,h}\}_{h\in[\mu]}$. \RRR receives $v^\RRR_{j,k,z}$ and \SSS receives $v^\SSS_{j,k,z}$.
            
            \item For $j\in[m],z\in[2^d]$, \RRR computes $r^\RRR_{j,z}={\sum}_{k\in[d]}v^\RRR_{j,k,z}$, and  \SSS computes $r^\SSS_{j,z}={\sum}_{k\in[d]}v^\SSS_{j,k,z}$.

            \item For $j\in [m],z\in[2^d]$, \RRR and \SSS invoke functionality \ensuremath{\FFF_{\mathsf{Eq}}} where \RRR inputs $r^{\RRR}_{j,z}$ and \SSS inputs $-r^{\SSS}_{j,z}$. \RRR receives $b_{j,z} = \mathbf{1}\{ r^{\RRR}_{j,z} = -r^{\SSS}_{j,z} \}$.

            \item For $j\in [m]$, \RRR computes $b_j = {\bigvee}_{z\in[2^d]}b_{j,z}$.

            \item For $j\in [m]$, \RRR and \SSS invoke functionality \Func[OT] where \RRR inputs $b_j$ and \SSS inputs $(\bot,\vecq_j)$. \RRR receives $\vecu_j$.

            \item \RRR outputs $I = \{ \vecu_j \; \vert \; b_j=1 \;\text{for} \;  j\in[m] \}$.

        \end{enumerate}

    \end{nfprot}
    \vspace{0.5em}
    \caption{Protocol of prefix-optimized fuzzy PSI for $L_\infty$ distance under receiver-sided unique cell assumption.}
    \label{Prot: receiver 2delta Linf px}
\end{figure}

\section{Fuzzy PSI for Unique Block Assumptions}
\label{Sec: 4delta}

We further consider a sparser distribution, termed \textit{unique block} (refer to Appendix~\ref{appendix: analysis of assumptions}), in which each cell intersects with at most one $L_\infty$ ball centered at a receiver point. This distribution guarantees that each receiver point's $L_\infty$ ball occupies a unique block and is therefore associated with a collision-free set of $2^d$ neighboring cells.

Unlike the \textit{unique cell} setting, our insight is to reverse the roles of the two parties in the spatial hashing procedure: the receiver programs all $2^d$ candidate cells without collision, while the sender evaluates only the single cell containing its point. This shifts the $2^d$ factor from the evaluation side to the programming side. As noted in~\cite{yang2026}, programming the so-OPPRF primarily involves OKVS~\cite{raghuraman2022blazing} encoding operations, which is substantially cheaper than evaluating the so-OPPRF; exploiting this asymmetry therefore yields a more efficient construction.

The detailed protocol for $L_\infty$ distance is presented in Figure~\ref{Prot: receiver 4delta Linf} and the security proof is deferred to Appendix~\ref{appendix: proof receiver 4delta Linf}. Constructions for general $L_p$ distances and their prefix-optimized variants are provided in Appendix~\ref{appendix: other protocols 4delta}.

\begin{theorem}
\label{thm: receiver 4delta Linf}
    The protocol $\Pi_\mathsf{FPSI}^{L_\infty}$ in Figure~\ref{Prot: receiver 4delta Linf} securely realizes the functionality \Func[FPSI] for $L_\infty$ distance against semi-honest adversaries in the (\Func[so\text{-}OPPRF],\Func[Eq],\Func[OT])-hybrid model.
\end{theorem}

\begin{figure}[t]
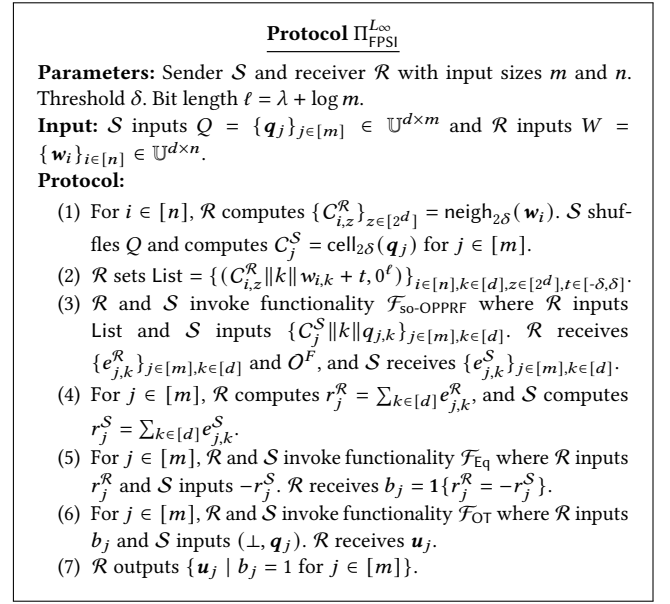

    \begin{nfprot}{\ensuremath{\Pi_\mathsf{FPSI}^{L_\infty}}\xspace}
        \noindent \textbf{Parameters:}  Sender \SSS and receiver \RRR with input sizes $m$ and $n$. Threshold \( \delta \). Bit length $\ell = \lambda + \log m$.

        \noindent \textbf{Input:} \SSS inputs \( Q=\{\vecq_j\}_{j\in [m]} \in \mathbb{U}^{d\times m} \) and \RRR inputs \( W=\{\vecw_i\}_{i\in [n]} \in \mathbb{U}^{d\times n} \).

        \noindent \textbf{Protocol:}
        \begin{enumerate}

            \item For $i\in[n]$, \RRR computes $\{\CCC_{i,z}^\RRR\}_{z\in[2^d]} = \mathsf{neigh}_{2\delta}(\vecw_i)$. \SSS shuffles $Q$ and computes $ \CCC^{\SSS}_{j} = \mathsf{cell}_{2\delta}(\vecq_j)$ for $j\in [m]$.
            
            \item \RRR sets $\mathsf{List} = \{(\CCC^{\RRR}_{i,z}\Vert k \Vert w_{i,k} + t,0^\ell)\}_{i\in[n],k\in[d],z\in[2^d],t\in[\text{-}\delta,\delta]}$.
            
            \item \RRR and \SSS invoke functionality \Func[so\text{-}OPPRF] where \RRR inputs $\mathsf{List}$ and \SSS inputs $\{\CCC^{\SSS}_j \Vert k \Vert q_{j,k}\}_{j\in[m],k\in[d]}$. \RRR receives $\{e^\RRR_{j,k}\}_{j\in[m],k\in[d]}$ and $\OOO^F$, and \SSS receives $\{e^\SSS_{j,k}\}_{j\in[m],k\in[d]}$.

            \item For $j\in[m]$, \RRR computes $r^\RRR_j={\sum}_{k\in[d]}e^\RRR_{j,k}$, and  \SSS computes $r^\SSS_j={\sum}_{k\in[d]}e^\SSS_{j,k}$.
            
            \item For $j\in[m]$, \RRR and \SSS invoke functionality \Func[Eq] where \RRR inputs $r^\RRR_j$ and \SSS inputs $-r^\SSS_j$. \RRR receives $b_j=\mathbf{1}\{r^\RRR_j=-r^\SSS_j\}$.
            \item For $j\in[m]$, \RRR and \SSS invoke functionality \Func[OT] where \RRR inputs $b_j$ and \SSS inputs $(\perp,\vecq_j)$. \RRR receives ${\vecu_j}$.
            
            \item \RRR outputs $\{\vecu_j\; \vert \; b_j = 1\;\text{for}\; j\in [m] \}$.

        \end{enumerate}

    \end{nfprot}
    \vspace{0.5em}
    \caption{Protocol of fuzzy PSI for $L_\infty$ distance under receiver-sided unique block assumption.}

    \label{Prot: receiver 4delta Linf}
\end{figure}

%% file: experiment.tex
\begin{table*}[t]
\caption{The communication (MB) and running time (s) of fuzzy PSI for $L_\infty$ distance under the unique cell assumption. Results of our protocols for $L_p$ distance are reported in Appendix~\ref{appendix: results Lp 2delta}, as existing works do not support this setting. ``---'' indicates out-of-memory or execution timeout (> 2.5h). The best result is marked in {\textcolor{green!95!black}{green}}.
}
\label{Tab: performance 2delta}
\resizebox{\textwidth}{!}{
\begin{tabular}{ccccccccccccc}

\hline
\multirow{2}{*}{\begin{tabular}[c]{@{}c@{}}Size\\ $m\!=\!n$\end{tabular}}     & \multirow{2}{*}{\begin{tabular}[c]{@{}c@{}}Dim.\\ $d$\end{tabular}} & \multirow{2}{*}{Protocol} & \multicolumn{2}{c}{$\delta=32$}          & \multicolumn{2}{c}{$\delta=64$}          & \multicolumn{2}{c}{$\delta=128$}         & \multicolumn{2}{c}{$\delta=256$}         & \multicolumn{2}{c}{$\delta=512$}         \\ \cline{4-13} 
                         &                    &                           & \multicolumn{1}{c}{Comm} & Time & \multicolumn{1}{c}{Comm} & Time & \multicolumn{1}{c}{Comm} & Time & \multicolumn{1}{c}{Comm} & Time & \multicolumn{1}{c}{Comm} & Time \\ \hline
\multirow{12}{*}{$2^8$}   & \multicolumn{1}{c}{\multirow{4}{*}{2}} & van Baarsen et al.~\cite{van2024fuzzy}  & \multicolumn{1}{c}{3.11}    & 1.77    & \multicolumn{1}{c}{6.11}     & 3.47    & \multicolumn{1}{c}{12.11}   & 6.95    & \multicolumn{1}{c}{24.11}    & 14.05   & \multicolumn{1}{c}{48.11}        & 27.98         \\ 
                          & \multicolumn{1}{c}{}                   & Bui et al.~\cite{fss25}             & \multicolumn{1}{c}{3.56}    & 1.48    & \multicolumn{1}{c}{4.15}     & 1.53    & \multicolumn{1}{c}{4.32}    & 1.71    & \multicolumn{1}{c}{ 4.98}         &    1.73      & \multicolumn{1}{c}{5.15}        &  2.07       \\ 
                          & \multicolumn{1}{c}{}                   & Ours                               & \multicolumn{1}{c}{\cellcolor{green!30}2.64} & \cellcolor{green!30}0.23 & \multicolumn{1}{c}{\cellcolor{green!30}3.32} & \cellcolor{green!30}0.24 & \multicolumn{1}{c}{4.68}    & \cellcolor{green!30}0.25 & \multicolumn{1}{c}{7.39}     & \cellcolor{green!30}0.29 & \multicolumn{1}{c}{12.83}        &  \cellcolor{green!30}0.41       \\
                          & \multicolumn{1}{c}{}                   & Ours-Px                            & \multicolumn{1}{c}{3.26}    & 0.45    & \multicolumn{1}{c}{3.34}     & 0.43    & \multicolumn{1}{c}{\cellcolor{green!30}3.41} & 0.45 & \multicolumn{1}{c}{\cellcolor{green!30}3.76} & 0.45 & \multicolumn{1}{c}{\cellcolor{green!30}3.76} & 0.48 \\ \cline{2-13}
                          & \multicolumn{1}{c}{\multirow{4}{*}{4}} & van Baarsen et al.~\cite{van2024fuzzy}  & \multicolumn{1}{c}{\cellcolor{green!30}6.34} & 3.72   & \multicolumn{1}{c}{12.34}    & 7.12    & \multicolumn{1}{c}{24.34}   & 14.31   & \multicolumn{1}{c}{48.34}    & 27.77   & \multicolumn{1}{c}{96.34}        & 55.63        \\ 
                          & \multicolumn{1}{c}{}                   & Bui et al.~\cite{fss25}             & \multicolumn{1}{c}{52.45}   & 10.53   & \multicolumn{1}{c}{121.79}   & 14.13   & \multicolumn{1}{c}{122.19}  & 23.61   & \multicolumn{1}{c}{248.00}         &   28.92      & \multicolumn{1}{c}{248.02}        &  42.94       \\ 
                          & \multicolumn{1}{c}{}                   & Ours                               & \multicolumn{1}{c}{8.66}    & \cellcolor{green!30}0.41 & \multicolumn{1}{c}{\cellcolor{green!30}10.02} & \cellcolor{green!30}0.41 & \multicolumn{1}{c}{\cellcolor{green!30}12.73} & \cellcolor{green!30}0.51 & \multicolumn{1}{c}{18.17}   & \cellcolor{green!30}0.67 & \multicolumn{1}{c}{29.05}        &   \cellcolor{green!30}0.73      \\
                          & \multicolumn{1}{c}{}                   & Ours-Px                            & \multicolumn{1}{c}{15.22}   & 1.22    & \multicolumn{1}{c}{15.57}    & 1.12    & \multicolumn{1}{c}{15.57}   & 1.14    & \multicolumn{1}{c}{\cellcolor{green!30}15.91} & 1.12 & \multicolumn{1}{c}{\cellcolor{green!30}16.25} & 1.16 \\ \cline{2-13}
                          & \multicolumn{1}{c}{\multirow{4}{*}{6}} & van Baarsen et al.~\cite{van2024fuzzy}  & \multicolumn{1}{c}{\cellcolor{green!30}10.14} & 6.43  & \multicolumn{1}{c}{\cellcolor{green!30}19.14} & 11.64  & \multicolumn{1}{c}{\cellcolor{green!30}37.14} & 22.08  & \multicolumn{1}{c}{73.14}    & 43.13   & \multicolumn{1}{c}{145.14}        &  83.55        \\ 
                          & \multicolumn{1}{c}{}                   & Bui et al.~\cite{fss25}             & \multicolumn{1}{c}{3070.90} & 2048.17 & \multicolumn{1}{c}{11695.09} & 2755.55 & \multicolumn{1}{c}{11695.82} & 7516.48 & \multicolumn{1}{c}{---}        & ---        & \multicolumn{1}{c}{---}        &    ---     \\ 
                          & \multicolumn{1}{c}{}                   & Ours                               & \multicolumn{1}{c}{39.43}   & \cellcolor{green!30}1.56 & \multicolumn{1}{c}{41.47}    & \cellcolor{green!30}1.43 & \multicolumn{1}{c}{45.55}   & \cellcolor{green!30}1.51 & \multicolumn{1}{c}{\cellcolor{green!30}53.70} & \cellcolor{green!30}1.71 & \multicolumn{1}{c}{\cellcolor{green!30}70.07}        &  \cellcolor{green!30}2.08       \\
                          & \multicolumn{1}{c}{}                   & Ours-Px                            & \multicolumn{1}{c}{82.96}   & 5.31    & \multicolumn{1}{c}{83.06}    & 5.67    & \multicolumn{1}{c}{83.30}   & 5.20    & \multicolumn{1}{c}{83.82}    & 5.41    & \multicolumn{1}{c}{84.32} & 5.38 \\ \hline
\multirow{12}{*}{$2^{12}$} & \multicolumn{1}{c}{\multirow{4}{*}{2}} & van Baarsen et al.~\cite{van2024fuzzy} & \multicolumn{1}{c}{49.75}   & 29.07   & \multicolumn{1}{c}{97.75}    & 56.37   & \multicolumn{1}{c}{193.75}  & 112.96  & \multicolumn{1}{c}{385.75}   & 222.60  & \multicolumn{1}{c}{769.75}        &  442.27       \\ 
                          & \multicolumn{1}{c}{}                   & Bui et al.~\cite{fss25}             & \multicolumn{1}{c}{38.18}   & 3.57    & \multicolumn{1}{c}{47.55}    & 4.64    & \multicolumn{1}{c}{50.18}   & 6.55    & \multicolumn{1}{c}{61.06}         &  9.62       & \multicolumn{1}{c}{63.69}        &  14.78       \\ 
                          & \multicolumn{1}{c}{}                   & Ours                               & \multicolumn{1}{c}{\cellcolor{green!30}24.72} & \cellcolor{green!30}0.96 & \multicolumn{1}{c}{35.60}    & \cellcolor{green!30}0.94 & \multicolumn{1}{c}{57.37}   & \cellcolor{green!30}1.43 & \multicolumn{1}{c}{100.93}   & \cellcolor{green!30}2.05    & \multicolumn{1}{c}{188.10}        &  3.74       \\
                          & \multicolumn{1}{c}{}                   & Ours-Px                            & \multicolumn{1}{c}{31.39}   & 2.11    & \multicolumn{1}{c}{\cellcolor{green!30}32.75} & 2.16 & \multicolumn{1}{c}{\cellcolor{green!30}34.11} & 2.25 & \multicolumn{1}{c}{\cellcolor{green!30}36.84} & 2.12 & \multicolumn{1}{c}{\cellcolor{green!30}39.54} & \cellcolor{green!30}2.26 \\ \cline{2-13}
                          & \multicolumn{1}{c}{\multirow{4}{*}{4}} & van Baarsen et al.~\cite{van2024fuzzy}  & \multicolumn{1}{c}{\cellcolor{green!30}101.50} & 59.81  & \multicolumn{1}{c}{197.50}   & 115.70  & \multicolumn{1}{c}{389.50}  & 228.65  & \multicolumn{1}{c}{773.50}   & 448.70  & \multicolumn{1}{c}{1541.50}        & 887.83         \\ 
                          & \multicolumn{1}{c}{}                   & Bui et al.~\cite{fss25}             & \multicolumn{1}{c}{819.28}  & 155.32  & \multicolumn{1}{c}{1931.30}  & 206.54  & \multicolumn{1}{c}{1936.54} & 341.16  & \multicolumn{1}{c}{3954.58}         &  453.26       & \multicolumn{1}{c}{3959.78}        &    704.95     \\ 
                          & \multicolumn{1}{c}{}                   & Ours                               & \multicolumn{1}{c}{120.46}  & \cellcolor{green!30}3.92 & \multicolumn{1}{c}{\cellcolor{green!30}142.23} & \cellcolor{green!30}4.34 & \multicolumn{1}{c}{\cellcolor{green!30}185.79} & \cellcolor{green!30}4.99 & \multicolumn{1}{c}{272.97}   & \cellcolor{green!30}6.48 & \multicolumn{1}{c}{447.41}        &   \cellcolor{green!30}9.61      \\
                          & \multicolumn{1}{c}{}                   & Ours-Px                            & \multicolumn{1}{c}{223.36}  & 14.45   & \multicolumn{1}{c}{226.09}   & 13.69   & \multicolumn{1}{c}{228.79}  & 14.60   & \multicolumn{1}{c}{\cellcolor{green!30}234.23} & 13.85 & \multicolumn{1}{c}{\cellcolor{green!30}239.67} & 14.81 \\ \cline{2-13}
                          & \multicolumn{1}{c}{\multirow{4}{*}{6}} & van Baarsen et al.~\cite{van2024fuzzy}  & \multicolumn{1}{c}{\cellcolor{green!30}162.25} & 106.67 & \multicolumn{1}{c}{\cellcolor{green!30}306.25} & 192.80 & \multicolumn{1}{c}{\cellcolor{green!30}594.25} & 354.21 & \multicolumn{1}{c}{1170.25}  & 690.49  & \multicolumn{1}{c}{2322.25}        &   1357.23      \\ 
                          & \multicolumn{1}{c}{}                   & Bui et al.~\cite{fss25}             & \multicolumn{1}{c}{---}     & ---     & \multicolumn{1}{c}{---}      & ---     & \multicolumn{1}{c}{---}     & ---     & \multicolumn{1}{c}{---}      & ---     & \multicolumn{1}{c}{---}        &   ---      \\ 
                          & \multicolumn{1}{c}{}                   & Ours                               & \multicolumn{1}{c}{612.81}  & \cellcolor{green!30}21.77 & \multicolumn{1}{c}{645.52}   & \cellcolor{green!30}22.00 & \multicolumn{1}{c}{710.97}  & \cellcolor{green!30}22.25 & \multicolumn{1}{c}{\cellcolor{green!30}841.95} & \cellcolor{green!30}24.60 & \multicolumn{1}{c}{\cellcolor{green!30}1104.05}        &   \cellcolor{green!30}29.58      \\
                          & \multicolumn{1}{c}{}                   & Ours-Px                            & \multicolumn{1}{c}{1311.18} & 80.01   & \multicolumn{1}{c}{1315.26}  & 80.14   & \multicolumn{1}{c}{1319.33} & 82.16   & \multicolumn{1}{c}{1327.52}  & 80.28   & \multicolumn{1}{c}{1335.70} & 81.44 \\ \hline
\multirow{12}{*}{$2^{16}$} & \multicolumn{1}{c}{\multirow{4}{*}{2}} & van Baarsen et al.~\cite{van2024fuzzy}  & \multicolumn{1}{c}{796.00}  & 478.07  & \multicolumn{1}{c}{1564.00}  & 904.65  & \multicolumn{1}{c}{3100.00} & 1836.01 & \multicolumn{1}{c}{6172.00}  & 3693.64 & \multicolumn{1}{c}{12316.00}        &  7499.23       \\ 
                          & \multicolumn{1}{c}{}                   & Bui et al.~\cite{fss25}             & \multicolumn{1}{c}{583.20}  & 39.28   & \multicolumn{1}{c}{733.25}   & 55.02   & \multicolumn{1}{c}{775.30}  & 80.48   & \multicolumn{1}{c}{949.34}         &    122.20     & \multicolumn{1}{c}{991.39}        &   240.33      \\ 
                          & \multicolumn{1}{c}{}                   & Ours                               & \multicolumn{1}{c}{\cellcolor{green!30}377.35} & \cellcolor{green!30}10.79 & \multicolumn{1}{c}{551.79}   & \cellcolor{green!30}13.15 & \multicolumn{1}{c}{900.81}  & \cellcolor{green!30}19.00   & \multicolumn{1}{c}{1599.34}  & 37.03   & \multicolumn{1}{c}{2996.97}        &   59.56      \\
                          & \multicolumn{1}{c}{}                   & Ours-Px                            & \multicolumn{1}{c}{483.52}  & 28.61   & \multicolumn{1}{c}{\cellcolor{green!30}505.41} & 29.10 & \multicolumn{1}{c}{\cellcolor{green!30}527.06} & 29.72 & \multicolumn{1}{c}{\cellcolor{green!30}570.86} & \cellcolor{green!30}30.17 & \multicolumn{1}{c}{\cellcolor{green!30}614.18} & \cellcolor{green!30}30.28 \\ \cline{2-13}
                          & \multicolumn{1}{c}{\multirow{4}{*}{4}} & van Baarsen et al.~\cite{van2024fuzzy}  & \multicolumn{1}{c}{\cellcolor{green!30}1624.00} & 970.83 & \multicolumn{1}{c}{3160.00}  & 1877.65 & \multicolumn{1}{c}{6232.00} & 3778.05 & \multicolumn{1}{c}{12376.00}      &   7453.94   & \multicolumn{1}{c}{---}        &   ---      \\ 
                          & \multicolumn{1}{c}{}                   & Bui et al.~\cite{fss25}             & \multicolumn{1}{c}{13078.68} & 2522.63 & \multicolumn{1}{c}{30874.60} & 3446.77 & \multicolumn{1}{c}{30958.79}     &  5633.09   & \multicolumn{1}{c}{---}         &  ---       & \multicolumn{1}{c}{---}        &   ---      \\ 
                          & \multicolumn{1}{c}{}                   & Ours                               & \multicolumn{1}{c}{1909.78} & \cellcolor{green!30}59.38 & \multicolumn{1}{c}{\cellcolor{green!30}2258.83} & \cellcolor{green!30}64.31 & \multicolumn{1}{c}{\cellcolor{green!30}2957.36} & \cellcolor{green!30}78.68 & \multicolumn{1}{c}{4354.93}  & \cellcolor{green!30}103.90 & \multicolumn{1}{c}{7151.31}        &  \cellcolor{green!30}157.46       \\
                          & \multicolumn{1}{c}{}                   & Ours-Px                            & \multicolumn{1}{c}{3577.83} & 211.56  & \multicolumn{1}{c}{3621.63}  & 219.71  & \multicolumn{1}{c}{3664.95} & 218.80  & \multicolumn{1}{c}{\cellcolor{green!30}3752.61} & 219.09 & \multicolumn{1}{c}{\cellcolor{green!30}3839.31} & 229.93 \\ \cline{2-13}
                          & \multicolumn{1}{c}{\multirow{4}{*}{6}} & van Baarsen et al.~\cite{van2024fuzzy}  & \multicolumn{1}{c}{\cellcolor{green!30}2596.00} & 1818.51 & \multicolumn{1}{c}{\cellcolor{green!30}4900.00} & 3066.45 & \multicolumn{1}{c}{\cellcolor{green!30}9508.00}     &   5825.47   & \multicolumn{1}{c}{18724.00}      & ---     & \multicolumn{1}{c}{37156.00}        &   ---      \\ 
                          & \multicolumn{1}{c}{}                   & Bui et al.~\cite{fss25}             & \multicolumn{1}{c}{---}     & ---     & \multicolumn{1}{c}{---}      & ---     & \multicolumn{1}{c}{---}     & ---     & \multicolumn{1}{c}{---}      & ---     & \multicolumn{1}{c}{---}        &   ---      \\ 
                          & \multicolumn{1}{c}{}                   & Ours                               & \multicolumn{1}{c}{9790.15} & \cellcolor{green!30}331.51 & \multicolumn{1}{c}{10314.64} & \cellcolor{green!30}339.08 & \multicolumn{1}{c}{11364.13} & \cellcolor{green!30}362.54 & \multicolumn{1}{c}{\cellcolor{green!30}13464.23}         & \cellcolor{green!30}400.39        & \multicolumn{1}{c}{\cellcolor{green!30}17666.67}        &  \cellcolor{green!30}477.81       \\
                          & \multicolumn{1}{c}{}                   & Ours-Px                            & \multicolumn{1}{c}{21118.37} & 1351.75 & \multicolumn{1}{c}{21183.78} & 1372.92 & \multicolumn{1}{c}{21249.35} & 1367.30 & \multicolumn{1}{c}{21380.25} & 1394.17 & \multicolumn{1}{c}{21511.43} & 1349.09 \\ \hline
\end{tabular}}
\end{table*}

\section{Evaluation}

We provide a full implementation of our protocols, available at\linebreak {{\color{blue}\url{https://github.com/Th0masAndy/FPSI}}}. All experiments are conducted on a single server equipped with AMD EPYC 9554 processor and 512 GB of RAM, with the sender and receiver emulated as two separate threads. Each experiment is repeated ten times, and we report the average total running time and communication cost. 
We set the computational and statistical security parameters to $\kappa = 128$ and $\lambda = 40$, respectively. 
For implementation, we adopt the so-OPPRF from~\cite{alamati2024improved,yang2026}, silent OT from libOTe\footnote{\url{https://github.com/osu-crypto/libOTe}}, secret-shared comparison from~\cite{huang2022cheetah}, and instantiate hash functions using BLAKE3\footnote{\url{https://github.com/BLAKE3-team/BLAKE3}}.

{
\noindent \textbf{Comparison Baselines.}
We compare our protocols with state-of-the-art fuzzy PSI constructions under one-sided assumptions~\cite{van2024fuzzy,dang2025ccs,fss25}.
The evaluation is divided according to two assumptions: unique cell and unique block.
Under the unique cell assumption, we compare against~\cite{van2024fuzzy,fss25}, as our unique cell assumption is strictly weaker than the $2\delta$-apart assumption of~\cite{van2024fuzzy} and equivalent to the mini-universe assumption of~\cite{fss25}.
Under the unique block assumption, we compare against~\cite{van2024fuzzy,dang2025ccs}, since our unique block assumption is strictly weaker than the $4\delta$-apart assumption used in~\cite{van2024fuzzy,dang2025ccs}.
Please refer to Appendix~\ref{appendix: analysis of assumptions} for detailed comparisons of different assumptions.
For fair comparison with~\cite{fss25}, we double its radius parameter, as~\cite{fss25} defines $L_\infty$ balls with side length $\delta$ rather than $2\delta$.

Since these baselines~\cite{van2024fuzzy,dang2025ccs,fss25} consider only receiver-sided assumptions, we evaluate our protocols under the same receiver-sided setting. Our sender-sided results are deferred to Appendix~\ref{appendix: results sender 2delta}.
}

\begin{table*}[t]
\caption{The communication (MB) and running time (s) of fuzzy PSI under the unique block assumption. The set size is $m=n=2^{12}$. ``---'' indicates out-of-memory or execution timeout (> 2.5h). The best result is marked in {\textcolor{green!95!black}{green}}.
}
\label{Tab: performance 4delta}
\resizebox{\textwidth}{!}{
\begin{tabular}{ccccccccccccc}
\hline
\multirow{2}{*}{Metric} & \multirow{2}{*}{\begin{tabular}[c]{@{}c@{}}Dim.\\ $d$\end{tabular}} & \multirow{2}{*}{Protocol} & \multicolumn{2}{c}{$\delta=32$} & \multicolumn{2}{c}{$\delta=64$} & \multicolumn{2}{c}{$\delta=128$} & \multicolumn{2}{c}{$\delta=256$} & \multicolumn{2}{c}{$\delta=512$} \\ \cline{4-13}
 & & & \multicolumn{1}{c}{Comm} & Time & \multicolumn{1}{c}{Comm} & Time & \multicolumn{1}{c}{Comm} & Time & \multicolumn{1}{c}{Comm} & Time & \multicolumn{1}{c}{Comm} & Time \\ \hline
\multirow{12}{*}{$L_\infty$} & \multirow{4}{*}{2} & van Baarsen et al.~\cite{van2024fuzzy} & \multicolumn{1}{c}{195.25}   & 111.97  & \multicolumn{1}{c}{387.25}   & 223.08  & \multicolumn{1}{c}{771.25}   & 446.28   & \multicolumn{1}{c}{1539.25}  & 893.64   & \multicolumn{1}{c}{3075.25}  & 1806.35  \\
 & & Dang et al.~\cite{dang2025ccs}   & \multicolumn{1}{c}{127.56}   & 14.42   & \multicolumn{1}{c}{145.48}   & 16.16   & \multicolumn{1}{c}{164.38}   & 19.12    & \multicolumn{1}{c}{183.55}   & 18.61    & \multicolumn{1}{c}{201.86 }         &   25.58       \\
 & & Ours                            & \multicolumn{1}{c}{48.58}    & \cellcolor{green!30}1.45    & \multicolumn{1}{c}{92.15}    & 2.54    & \multicolumn{1}{c}{179.32}   & 4.54     & \multicolumn{1}{c}{353.75}   & 8.67    & \multicolumn{1}{c}{702.75}         &   17.03       \\
 & & Ours-Px                         & \multicolumn{1}{c}{\cellcolor{green!30}29.76}  & 1.96  & \multicolumn{1}{c}{\cellcolor{green!30}33.70}  & \cellcolor{green!30}2.12  & \multicolumn{1}{c}{\cellcolor{green!30}37.79}  & \cellcolor{green!30}2.32  & \multicolumn{1}{c}{\cellcolor{green!30}41.83}  & \cellcolor{green!30}2.53  & \multicolumn{1}{c}{\cellcolor{green!30}45.75}         &   \cellcolor{green!30}2.71       \\ \cline{2-13}
 & \multirow{4}{*}{4} & van Baarsen et al.~\cite{van2024fuzzy} & \multicolumn{1}{c}{1560.25}  & 932.15  & \multicolumn{1}{c}{3096.25}  & 1836.02 & \multicolumn{1}{c}{6168.25}  & 3732.99  & \multicolumn{1}{c}{12312.25} & 7697.61  & \multicolumn{1}{c}{---}         &  ---        \\
 & & Dang et al.~\cite{dang2025ccs}   & \multicolumn{1}{c}{2289.54}  & 196.76  & \multicolumn{1}{c}{1576.67}  & 142.00  & \multicolumn{1}{c}{2006.78}  & 161.69   & \multicolumn{1}{c}{2578.55}  & 208.96   & \multicolumn{1}{c}{2866.33}         &  280.98        \\
 & & Ours                            & \multicolumn{1}{c}{361.61}   & 8.91   & \multicolumn{1}{c}{710.66}   & 17.14   & \multicolumn{1}{c}{1409.19}  & 34.32    & \multicolumn{1}{c}{2806.76}  & 68.73    & \multicolumn{1}{c}{5603.14}         &  139.48        \\
 & & Ours-Px                         & \multicolumn{1}{c}{\cellcolor{green!30}86.46}  & \cellcolor{green!30}3.93  & \multicolumn{1}{c}{\cellcolor{green!30}98.56}  & \cellcolor{green!30}4.83  & \multicolumn{1}{c}{\cellcolor{green!30}110.85} & \cellcolor{green!30}5.07  & \multicolumn{1}{c}{\cellcolor{green!30}123.06} & \cellcolor{green!30}5.59  & \multicolumn{1}{c}{\cellcolor{green!30}135.24}         & \cellcolor{green!30}5.90        \\ \cline{2-13}
 & \multirow{4}{*}{6} & van Baarsen et al.~\cite{van2024fuzzy} & \multicolumn{1}{c}{9360.25}  & 5507.11 & \multicolumn{1}{c}{---}      & ---     & \multicolumn{1}{c}{---}      & ---      & \multicolumn{1}{c}{---}      & ---      & \multicolumn{1}{c}{---}      & ---      \\
 & & Dang et al.~\cite{dang2025ccs}   & \multicolumn{1}{c}{13731.49} & 1192.01 & \multicolumn{1}{c}{9443.74}  & 821.75  & \multicolumn{1}{c}{12016.79} & 1013.32  & \multicolumn{1}{c}{15450.58} & 1275.52  & \multicolumn{1}{c}{17170.10}         &   1617.10       \\
 & & Ours                            & \multicolumn{1}{c}{2141.00}  & 55.85   & \multicolumn{1}{c}{4241.01}  & 103.50  & \multicolumn{1}{c}{8443.66}  & 219.94   & \multicolumn{1}{c}{16852.30} & 474.79   & \multicolumn{1}{c}{33544.18}         &  996.22        \\
 & & Ours-Px                         & \multicolumn{1}{c}{\cellcolor{green!30}301.11} & \cellcolor{green!30}9.12  & \multicolumn{1}{c}{\cellcolor{green!30}343.89} & \cellcolor{green!30}10.47  & \multicolumn{1}{c}{\cellcolor{green!30}386.23} & \cellcolor{green!30}12.30 & \multicolumn{1}{c}{\cellcolor{green!30}428.84} & \cellcolor{green!30}13.69 & \multicolumn{1}{c}{\cellcolor{green!30}471.97}         &   \cellcolor{green!30}14.29       \\ \hline
\multirow{12}{*}{$L_1$}      & \multirow{4}{*}{2} & van Baarsen et al.~\cite{van2024fuzzy} & \multicolumn{1}{c}{195.77}   & 112.01  & \multicolumn{1}{c}{388.27}   & 224.03  & \multicolumn{1}{c}{773.27}   & 454.50   & \multicolumn{1}{c}{1543.27}  & 893.72   & \multicolumn{1}{c}{3083.27}  & 1787.95  \\
 & & Dang et al.~\cite{dang2025ccs}   & \multicolumn{1}{c}{312.84}   & 155.45  & \multicolumn{1}{c}{381.93}   & 207.62  & \multicolumn{1}{c}{456.06}   & 279.29   & \multicolumn{1}{c}{534.66}   & 359.87   & \multicolumn{1}{c}{618.78}         & 461.89         \\
 & & Ours                            & \multicolumn{1}{c}{\cellcolor{green!30}51.83}    & \cellcolor{green!30}1.57    & \multicolumn{1}{c}{95.43}    & \cellcolor{green!30}2.67    & \multicolumn{1}{c}{182.64}   & \cellcolor{green!30}4.74     & \multicolumn{1}{c}{357.21}   & 9.09    & \multicolumn{1}{c}{706.33}         &  17.25        \\
 & & Ours-Px                         & \multicolumn{1}{c}{80.17}    & 3.97    & \multicolumn{1}{c}{\cellcolor{green!30}93.12}    & 4.32    & \multicolumn{1}{c}{\cellcolor{green!30}106.11}   & 4.82     & \multicolumn{1}{c}{\cellcolor{green!30}119.26}   & \cellcolor{green!30}5.56     & \multicolumn{1}{c}{\cellcolor{green!30}132.36}         & \cellcolor{green!30}6.02         \\ \cline{2-13}
 & \multirow{4}{*}{4} & van Baarsen et al.~\cite{van2024fuzzy} & \multicolumn{1}{c}{1560.77}  & 901.23  & \multicolumn{1}{c}{3097.27}  & 1823.36 & \multicolumn{1}{c}{6170.27}  & 3660.91  & \multicolumn{1}{c}{12316.27} & 7279.21  & \multicolumn{1}{c}{---}      & ---      \\
 & & Dang et al.~\cite{dang2025ccs}   & \multicolumn{1}{c}{3609.97}  & 603.43  & \multicolumn{1}{c}{4944.93}  & 5443.02 & \multicolumn{1}{c}{5517.94}  & 5624.82  & \multicolumn{1}{c}{6376.92}  & 5695.84  & \multicolumn{1}{c}{---}         &  ---        \\
 & & Ours                            & \multicolumn{1}{c}{367.18}   & 9.14   & \multicolumn{1}{c}{716.27}   & 17.49   & \multicolumn{1}{c}{1414.84}  & 34.57    & \multicolumn{1}{c}{2812.54}  & 72.58    & \multicolumn{1}{c}{5609.03}         &  146.14        \\
 & & Ours-Px                         & \multicolumn{1}{c}{\cellcolor{green!30}207.10}   & \cellcolor{green!30}8.34    & \multicolumn{1}{c}{\cellcolor{green!30}240.99}   & \cellcolor{green!30}9.59   & \multicolumn{1}{c}{\cellcolor{green!30}275.07}   & \cellcolor{green!30}10.91    & \multicolumn{1}{c}{\cellcolor{green!30}309.59}   & \cellcolor{green!30}12.01    & \multicolumn{1}{c}{\cellcolor{green!30}343.87}         &  \cellcolor{green!30}13.19        \\ \cline{2-13}
 & \multirow{4}{*}{6} & van Baarsen et al.~\cite{van2024fuzzy} & \multicolumn{1}{c}{9360.77}  & 5481.36 & \multicolumn{1}{c}{---}      & ---     & \multicolumn{1}{c}{---}      & ---      & \multicolumn{1}{c}{---}      & ---      & \multicolumn{1}{c}{---}      & ---      \\
 & & Dang et al.~\cite{dang2025ccs}   & \multicolumn{1}{c}{22946.21} & 7243.12 & \multicolumn{1}{c}{---}      & ---     & \multicolumn{1}{c}{---}      & ---      & \multicolumn{1}{c}{---}      & ---      & \multicolumn{1}{c}{---}      & ---      \\
 & & Ours                            & \multicolumn{1}{c}{2148.88}  & 56.24   & \multicolumn{1}{c}{4248.93}  & 115.48  & \multicolumn{1}{c}{8451.61}  & 234.19   & \multicolumn{1}{c}{16860.39} & 445.43   & \multicolumn{1}{c}{33552.39}         &    989.84      \\
 & & Ours-Px                         & \multicolumn{1}{c}{\cellcolor{green!30}604.10}   & \cellcolor{green!30}22.84   & \multicolumn{1}{c}{\cellcolor{green!30}704.70}   & \cellcolor{green!30}22.41   & \multicolumn{1}{c}{\cellcolor{green!30}804.81}   & \cellcolor{green!30}25.55    & \multicolumn{1}{c}{\cellcolor{green!30}904.76}   & \cellcolor{green!30}29.43    & \multicolumn{1}{c}{\cellcolor{green!30}1004.55}         &  \cellcolor{green!30}38.15        \\ \hline
\multirow{12}{*}{$L_2$}      & \multirow{4}{*}{2} & van Baarsen et al.~\cite{van2024fuzzy} & \multicolumn{1}{c}{211.27}   & 112.01  & \multicolumn{1}{c}{451.27}   & 221.86  & \multicolumn{1}{c}{1027.27}  & 448.16   & \multicolumn{1}{c}{2563.27}  & 891.59   & \multicolumn{1}{c}{7171.27}  & 1803.15  \\
 & & Dang et al.~\cite{dang2025ccs}   & \multicolumn{1}{c}{530.49}   & 151.32  & \multicolumn{1}{c}{642.30}   & 212.06  & \multicolumn{1}{c}{766.64}   & 285.79   & \multicolumn{1}{c}{911.44}   & 383.57   & \multicolumn{1}{c}{1092.79}         &    566.46      \\
 & & Ours                            & \multicolumn{1}{c}{\cellcolor{green!30}52.27}    & \cellcolor{green!30}1.55    & \multicolumn{1}{c}{\cellcolor{green!30}96.33}    & \cellcolor{green!30}2.58    & \multicolumn{1}{c}{183.57}   & \cellcolor{green!30}4.84     & \multicolumn{1}{c}{358.25}   & 8.96    & \multicolumn{1}{c}{707.49}         &  17.82        \\
 & & Ours-Px                         & \multicolumn{1}{c}{108.54}   & 4.33    & \multicolumn{1}{c}{126.57}   & 5.14    & \multicolumn{1}{c}{\cellcolor{green!30}144.23}   & 5.83     & \multicolumn{1}{c}{\cellcolor{green!30}162.12}   & \cellcolor{green!30}6.66     & \multicolumn{1}{c}{\cellcolor{green!30}179.96}         &  \cellcolor{green!30}7.37        \\ \cline{2-13}
 & \multirow{4}{*}{4} & van Baarsen et al.~\cite{van2024fuzzy} & \multicolumn{1}{c}{1576.27}  & 896.00  & \multicolumn{1}{c}{3160.27}  & 1795.90 & \multicolumn{1}{c}{6424.27}  & 3604.21  & \multicolumn{1}{c}{13336.27} & 7208.44  & \multicolumn{1}{c}{---}         &   ---       \\
 & & Dang et al.~\cite{dang2025ccs}   & \multicolumn{1}{c}{6995.61}  & 627.84  & \multicolumn{1}{c}{7493.30}  & 5470.09 & \multicolumn{1}{c}{8644.48}  & 5484.85  & \multicolumn{1}{c}{10379.32} & 5568.10  & \multicolumn{1}{c}{---}         & ---         \\
 & & Ours                            & \multicolumn{1}{c}{367.62}   & \cellcolor{green!30}9.61   & \multicolumn{1}{c}{717.17}   & 17.93   & \multicolumn{1}{c}{1415.77}  & 35.54    & \multicolumn{1}{c}{2813.58}  & 73.14    & \multicolumn{1}{c}{5610.19}         &   141.26       \\
 & & Ours-Px                         & \multicolumn{1}{c}{\cellcolor{green!30}263.22}   & 9.63   & \multicolumn{1}{c}{\cellcolor{green!30}306.83}   & \cellcolor{green!30}11.09   & \multicolumn{1}{c}{\cellcolor{green!30}350.20}   & \cellcolor{green!30}12.45    & \multicolumn{1}{c}{\cellcolor{green!30}394.08}   & \cellcolor{green!30}14.24    & \multicolumn{1}{c}{\cellcolor{green!30}437.73}         & \cellcolor{green!30}15.87         \\ \cline{2-13}
 & \multirow{4}{*}{6} & van Baarsen et al.~\cite{van2024fuzzy} & \multicolumn{1}{c}{9376.27}  & 5485.92 & \multicolumn{1}{c}{---}      & ---     & \multicolumn{1}{c}{---}      & ---      & \multicolumn{1}{c}{---}      & ---      & \multicolumn{1}{c}{---}         &  ---        \\
 & & Dang et al.~\cite{dang2025ccs}   & \multicolumn{1}{c}{43227.86}      & 7374.09     & \multicolumn{1}{c}{---}      & ---     & \multicolumn{1}{c}{---}      & ---      & \multicolumn{1}{c}{---}      & ---      & \multicolumn{1}{c}{---}      & ---      \\
 & & Ours                            & \multicolumn{1}{c}{2149.32}  & 52.40   & \multicolumn{1}{c}{4249.83}  & 109.84  & \multicolumn{1}{c}{8452.55}  & 218.51   & \multicolumn{1}{c}{16861.43} & 447.69   & \multicolumn{1}{c}{33553.55}         &   967.63       \\
 & & Ours-Px                         & \multicolumn{1}{c}{\cellcolor{green!30}687.98}   & \cellcolor{green!30}24.19   & \multicolumn{1}{c}{\cellcolor{green!30}802.92}   & \cellcolor{green!30}25.46   & \multicolumn{1}{c}{\cellcolor{green!30}916.94}   & \cellcolor{green!30}27.46    & \multicolumn{1}{c}{\cellcolor{green!30}1030.87}  & \cellcolor{green!30}32.82    & \multicolumn{1}{c}{\cellcolor{green!30}1144.66}         &  \cellcolor{green!30}40.48        \\ \hline
\end{tabular}}
\end{table*}

\subsection{Performance on Unique Cell Assumptions}

We compare our protocols with the state-of-the-art constructions~\cite{fss25,van2024fuzzy} for $L_\infty$ distance under the unique cell assumption in Table~\ref{Tab: performance 2delta}. 
The results show that our protocols outperform prior works in both communication and computation across most parameter settings.
Since existing works do not support general $L_p$ distance under this assumption, we present our results for $L_1$ and $L_2$ distances in Appendix~\ref{appendix: results Lp 2delta}.

\noindent\textbf{Communication.} Against van Baarsen and Pu~\cite{van2024fuzzy}, our best protocol reduces communication by up to $20\times$ at $d=2$, with the advantage growing substantially as $\delta$ increases, a direct consequence of our sublinear dependence on $\delta$ versus the linear growth of~\cite{van2024fuzzy}. At higher dimensions ($d=4, 6$), \cite{van2024fuzzy} achieves lower communication at small $\delta=32$, but our protocols overtake it as $\delta$ grows. Against Bui et al.~\cite{fss25}, our protocols achieve up to $282\times$ communication reduction at $d=6$, where the communication of~\cite{fss25} blows up due to its $O((\log\delta)^d)$ complexity. At $d=2$, the gap is smaller ($1.2$--$1.7\times$), as~\cite{fss25} benefits from low-dimensional setting.

\noindent\textbf{Computation.} The computational advantage is significant across all settings. Against van Baarsen and Pu~\cite{van2024fuzzy}, our protocols achieve speedups of $4$--$248\times$, with the largest gains at large $\delta$ and large set sizes, reaching over $200\times$ at $n = 2^{16}$ and $\delta = 512$. Against Bui et al.~\cite{fss25}, we achieve speedups of $4$--$4978\times$, with the most dramatic gains at $d=6$, where the running time of~\cite{fss25} becomes prohibitive due to its $O((\log\delta)^d)$ complexity. These gains stem primarily from our exclusive use of symmetric-key primitives, avoiding both the superlinear complexity of~\cite{fss25} and the expensive additive homomorphic operations of~\cite{van2024fuzzy}.

\noindent\textbf{Effect of Prefix Optimization.} These two variants expose a clear communication-computation tradeoff. The non-optimized protocol scales linearly in $\delta$ and achieves the lowest running time for small $\delta$ (32 and 64). The prefix-optimized protocol scales logarithmically in $\delta$, making it superior in communication at large $\delta$; for instance, at $\delta = 512$ and $d=2$, it reduces communication by up to $20\times$ over~\cite{van2024fuzzy} while maintaining a substantial computational advantage.

\subsection{Performance on Unique Block Assumptions}

{We compare our protocols with the state-of-the-art constructions~\cite{dang2025ccs,van2024fuzzy} under the unique block assumption. 
Table~\ref{Tab: performance 4delta} reports results for set size $2^{12}$ as a representative case; other settings exhibit similar trends, since both prior works and our protocols scale linearly with the set size.}

\noindent\textbf{Communication.} Compared to van Baarsen and Pu~\cite{van2024fuzzy}, the communication reduction grows substantially with the threshold: at $\delta = 32$ we achieve $4$--$31\times$ reduction, increasing to $13$--$100\times$ at $\delta = 256$ and up to $67\times$ at $\delta = 512$. This rapid scaling is a direct consequence of our logarithmic dependence on $\delta$, whereas~\cite{van2024fuzzy} incurs communication growing linearly in $\delta$ or $\delta^p$. Against Dang et al.~\cite{dang2025ccs}, our prefix-optimized protocol consistently reduces communication by $4$--$63\times$ across all metrics and parameter settings, with the largest gains at higher dimension ($d = 6$). We note that our protocol achieves better performance at larger parameter settings, owing to the underlying silent OT~\cite{boyle2019efficient} used in our implementation, which offers superior amortized communication efficiency for large batch generation.

\noindent\textbf{Computation.}  The computational advantage is even more pronounced. Against van Baarsen and Pu~\cite{van2024fuzzy}, our protocol achieves speedups of $71$--$1377\times$ across all metrics and dimensions. Against Dang et al.~\cite{dang2025ccs}, we achieve $7$--$568\times$ speedups. 
These dramatic gains stem from two factors. First, our prefix-optimized protocol achieves $O(\log\delta)$ complexity, while~\cite{van2024fuzzy} scales linearly in $\delta$ and~\cite{dang2025ccs} incurs $O((\log\delta)^d)$ complexity. Second, our protocols rely exclusively on symmetric-key primitives, whereas both~\cite{van2024fuzzy,dang2025ccs} depend on expensive public-key operations such as DDH and Paillier.

\noindent\textbf{Effect of Prefix Optimization.}
The prefix-optimized variant attains the best performance in nearly all settings, in contrast to the previous case. This is because, under the \textit{unique block} assumption, the dominant bottleneck is the receiver programming $2^d \cdot d \cdot n \cdot (2\delta+1)$ key-value pairs, while the sender evaluates at only $d \cdot m$ points. Incorporating prefix trie techniques reduces the receiver's overhead by a factor of $O(\delta/\log\delta)$, while increasing the sender's overhead by only $O(\log\delta)$, yielding a clear improvement in overall efficiency. Moreover, this advantage becomes more pronounced as the dimension $d$ and the threshold $\delta$ increase, indicating that the prefix-optimized variant is well-suited to large-scale settings.

%% file: conclusion.tex
\section{Conclusion}
In this work, we presented concretely efficient fuzzy PSI protocols for general 
$L_{p\in[1,\infty]}$ distances under one-sided assumptions, relying solely on 
lightweight symmetric-key primitives. The core technique is new multi-point fuzzy 
matching protocols, upon which we derived protocols flexibly supporting both sender-sided and receiver-sided settings. By incorporating prefix-trie techniques, 
we further achieved $O(\log\delta)$ complexity in the distance threshold $\delta$, improving upon the $O((\log\delta)^d)$ and $O(\delta)$ complexities 
of prior works. Extensive experiments confirm that our protocols significantly 
outperform all prior works under the same assumptions.

%% file: appendix.tex
\section{Open Science}

We follow the principles of the Open Science Policy. Our research artifacts, comprising the source code of the proposed protocols, benchmarking scripts, and comprehensive technical documentation, have been consolidated into an open-source repository, which is accessible at {\color{blue}\url{https://github.com/Th0masAndy/FPSI}}.

\section{Ethical Considerations}

This work provides effective solutions for fuzzy private set intersection tasks. The experiments in this paper are all based on synthetic datasets and do not contain any personal or illegal information. We firmly believe that our research was done ethically.

\section{Generative AI Usage}

This paper used LLMs (including ChatGPT and Gemini) solely to assist with grammar checking and language polishing. The LLMs did not contribute any conceptual ideas, methodological innovations, experimental designs, or analytical insights to this work. All intellectual contributions originate from the authors. The content provided to the LLMs for linguistic refinement contains no sensitive, personal, or ethically problematic information.

\section{Other Related Work}
\label{appendix: other work}

van Baarsen and Pu~\cite{van2024fuzzy} proposed the first general fuzzy PSI construction under the ``\textit{existing disjoint}'' assumption, requiring that each point in both parties' sets is at least $2\delta$ apart from all other points in some dimension. The overall complexity scales as $O((d\delta)^2)$, with no additional exponential factor. Gao et al.~\cite{gao2025efficient} subsequently improved this to $O(d\delta)$, achieving fully linear complexity for the first time. Their protocol, however, still operates under the same two-sided ``\textit{existing disjoint}'' assumption, as the symmetric structure of their construction requires both parties' inputs to satisfy the condition. Zhang et al.~\cite{zhang2025fast} adopted a stronger ``$s$-\textit{separate}'' assumption and proposed an efficient fuzzy PSI protocol from symmetric-key primitives, at the cost of $O(\delta^s)$ overhead, which becomes prohibitive for large $\delta$. Dang et al.~\cite{dang2025ccs} improved upon Gao et al.~\cite{gao2025efficient} under the same assumption via prefix trie techniques, reducing the overhead from $O(d\delta)$ to $O(d\log\delta)$. This was further improved by Yang et al.~\cite{yang2026} via secret-shared OPRF~\cite{alamati2024improved}, eliminating the reliance on public-key primitives.
Recently, van Baarsen and Pu~\cite{van2025} revisited their earlier construction~\cite{van2024fuzzy} using symmetric-key primitives. Their new protocol relies on spatial hashing to map each point to a unique cell, requiring both parties' points to be at least $2\delta$ or $4\delta$ apart to prevent additional leakage. Chongchitmate et al.~\cite{chongchitmate2024approximate} proposed a nearly linear fuzzy PSI protocol, but their construction requires a strong gap assumption: every pair of points across the two sets must be either close or far apart. This is a strict assumption, as it imposes a joint constraint across both parties' inputs.

While these works have achieved substantial efficiency improvements, their advantages rely on two-sided assumptions, which are harder to satisfy in real-world applications. Removing the assumption on either party's input would render these protocols insecure, inapplicable, or subject to dramatically increased overhead.

Two additional constructions operate under the arbitrary distribution setting. The first~\cite{garimella2024computation} incurs $O(mn(\log\delta)^d)$ complexity, which is quadratic in the set sizes $m$ and $n$. The second~\cite{van2025} achieves $O((m+n)(\log\delta)^d)$ complexity. However, both constructions suffer from the worse $O((\log\delta)^d)$ factor in both communication and computation, making them far from practical. Moreover, they support only $L_\infty$ distance with no known efficient extension to general $L_{p\in[1,\infty]}$ distances; we therefore omit a detailed discussion.

\section{Other Functionalities}
\label{appendix: func other}

\subsection{Interval Test}
Chakraborti et al.~\cite{chakraborti2023distance} introduced an efficient interval test protocol, which runs in $O(\log\delta)$ complexity with respect to the interval size $\delta$. We give the ideal functionality in Figure~\ref{Func: Interval}.

\begin{figure}[!h]
\begin{nffunc}{$\FFF^\delta_{\mathsf{Interval}}$}

\noindent \textbf{Parameters:} Sender \SSS and receiver \RRR. Threshold $\delta$.

\noindent \textbf{Functionality:}

\begin{enumerate}

    \item Wait for input $x^{\SSS}$ from \SSS and $x^{\RRR}$ from \RRR.

    \item Output 1 to \RRR if $x^{\SSS}+x^{\RRR} \le \delta$, otherwise output 0.

\end{enumerate}

\end{nffunc}
\vspace{0.5em}
\caption{Functionality of private interval test.}
\label{Func: Interval}
\end{figure}

\subsection{MUX}

The ideal functionality of MUX is given in Figure~\ref{Func:MUX}.

\begin{figure}[!h]
\begin{nffunc}{\Func[MUX]}

\noindent \textbf{Parameters:} Sender \SSS and receiver \RRR.

\noindent \textbf{Functionality:}

\begin{enumerate}

    \item Wait for input $(b_0,x_0)$ from \SSS and $(b_1,x_1)$ from \RRR.

    \item Select $r_0$ randomly, and set $r_1=x_0+x_1-r_0$ if $b_0\oplus b_1=1$ otherwise set $r_1=-r_0$.
    \item Output $r_0$ to \RRR and $r_1$ to \SSS.
    
\end{enumerate}

\end{nffunc}
\vspace{0.5em}
\caption{Functionality of MUX.}
\label{Func:MUX}
\end{figure}

\subsection{Secret-shared Private Equality Test}

The ideal functionality of secret-shared private equality test is given in Figure~\ref{Func: ssPEQT}.

\begin{figure}[!h]
\begin{nffunc}{\Func[ssPEQT]}

\noindent \textbf{Parameters:} Sender \SSS and receiver \RRR.

\noindent \textbf{Functionality:}

\begin{enumerate}

    \item Wait for input $x_0$ from \SSS and $x_1$ from \RRR.

    \item Select $r_0$ randomly, and set $r_1=1-r_0$ if $x_0=x_1$ otherwise set $r_1=-r_0$.
    \item Output $r_0$ to \RRR and $r_1$ to \SSS.
    
\end{enumerate}

\end{nffunc}
\vspace{0.5em}
\caption{Functionality of secret-shared PEQT.}
\label{Func: ssPEQT}
\end{figure}

\subsection{Secret-shared Comparison}

The ideal functionality of secret-shared comparison is given in Figure~\ref{Func: ssCMP}.

\begin{figure}[!h]
\begin{nffunc}{\Func[ssCMP]}

\noindent \textbf{Parameters:} Sender \SSS and receiver \RRR.

\noindent \textbf{Functionality:}

\begin{enumerate}

    \item Wait for input $x_0$ from \SSS and $x_1$ from \RRR.

    \item Select $r_0$ randomly, and set $r_1=1-r_0$ if $x_0\ge x_1$ otherwise set $r_1=-r_0$.
    \item Output $r_0$ to \RRR and $r_1$ to \SSS.
    
\end{enumerate}

\end{nffunc}
\vspace{0.5em}
\caption{Functionality of secret-shared comparison.}
\label{Func: ssCMP}
\end{figure}

\subsection{Private Equality Test}

The ideal functionality of private equality test is given in Figure~\ref{Func: Eq}.

\begin{figure}[!h]
\begin{nffunc}{\Func[Eq]}

\noindent \textbf{Parameters:} Sender \SSS and receiver \RRR.

\noindent \textbf{Functionality:}

\begin{enumerate}

    \item Wait for input $x_0$ from \SSS and $x_1$ from \RRR.

    \item Set $r=1$ if $x_0=x_1$ otherwise set $r=0$.
    \item Output $r$ to \RRR
    
\end{enumerate}

\end{nffunc}
\vspace{0.5em}
\caption{Functionality of private equality test.}
\label{Func: Eq}
\end{figure}

\section{Comparisons of Assumptions}
\label{appendix: analysis of assumptions}

We compare the unique cell assumption of Definition~\ref{def: unique cell} with existing assumptions as follows.

\subsection{Unique Cell}

\noindent\subsubsection{Comparison with disjoint hash assumption.} The disjoint hash assumption~\cite{piske2025distance,richardson2024fuzzy} (abbreviated as disj.\ hash in Figure~\ref{Fig: assumptions}) requires the spatial hash to map each point to a distinct cell, which is identical to our \textit{unique cell} assumption. However, while~\cite{piske2025distance,richardson2024fuzzy} rely on this assumption in their theoretical constructions, their concrete protocols require additional assumptions on the other party.

\subsubsection{Comparison with mini-universe assumption}
We revisit the \textit{mini-universe} assumption from~\cite{fss25,garimella2024computation}. Since \cite{garimella2024computation,fss25} defines $L_\infty$ balls of side length $\delta$ rather than $2\delta$, we double the radius parameter to align with our setting. 
Consider a $d$-dimensional space tiled by cells of side length $2\delta$. For any grid vertex $\veco = (o_1, o_2, \ldots, o_d)$, where coordinate $o_i$ is an integer multiple of $2\delta$, we define the \emph{mini-universe} with origin $\veco$ as the $d$-dimensional hypercube of side length $4\delta$, namely $\mathsf{Univ}_{\veco} := [o_1,\, o_1+4\delta) \times \cdots \times [o_d,\, o_d+4\delta)$.
For any point $\vecx \in \mathbb{U}^d$, we define $\mathsf{Ball}_{\vecx}$ as the $L_\infty$ ball of side length $2\delta$ with lower-left corner at $\vecx$, namely $\mathsf{Ball}_{\vecx} := [x_1,\, x_1+2\delta] \times \cdots \times [x_d,\, x_d+2\delta]$. Here, we use closed intervals to align with our definition, since our ball contains exactly $2\delta+1$ values in each dimension, while the ball in~\cite{garimella2024computation,fss25} only contains $2\delta$ values.  

\begin{definition}[Mini-Universe~\cite{garimella2024computation,fss25}]
A set $W \in \mathbb{U}^{n \times d}$ satisfies the \textit{mini-universe} assumption if for every $\vecx \in W$, there exists a unique mini-universe $\mathsf{Univ}_{\veco}$ containing $\mathsf{Ball}_{\vecx}$. In other words, for any two distinct points $\vecx, \vecx' \in W$ with $\mathsf{Ball}_{\vecx} \subseteq \mathsf{Univ}_{\veco}$ and $\mathsf{Ball}_{\vecx'} \subseteq \mathsf{Univ}_{\veco'}$, it holds that $\mathsf{Univ}_{\veco} \neq \mathsf{Univ}_{\veco'}$.
\end{definition}

{Below, we formally prove that the unique cell assumption is equivalent to the mini-universe assumption.}

\begin{lemma}
A set $W \in \mathbb{U}^{n \times d}$ satisfies the \textit{unique cell} assumption if and only if it satisfies the \textit{mini-universe} assumption.
\end{lemma}

\begin{proof}
We first show that the \textit{unique cell} assumption implies the \textit{mini-universe} assumption. Let $\vecx \in W$ lie in the grid cell with lower-left corner $\veco$, i.e., $x_i \in [o_i, o_i + 2\delta)$ for all $i \in [d]$. Then $x_i + 2\delta < o_i + 4\delta$, so $[x_i,x_i+2\delta] \subset [o_i,o_i+4\delta)$ for all $i\in[d]$. Therefore, it holds that $\mathsf{Ball}_{\vecx} \subseteq \mathsf{Univ}_{\veco}$, meaning every point in $W$ has its ball $\mathsf{Ball}_{\vecx}$ contained in the mini-universe $\mathsf{Univ}_{\veco}$ determined by its cell's lower-left corner $\veco$. For any two distinct points $\vecx, \vecx' \in W$, the \textit{unique cell} assumption gives $\mathsf{cell}_{2\delta}(\vecx) \neq \mathsf{cell}_{2\delta}(\vecx')$, so $\veco \neq \veco'$ and hence $\mathsf{Univ}_{\veco} \neq \mathsf{Univ}_{\veco'}$. Thus the mini-universe containing each ball is unique.

We next show that the \textit{mini-universe} assumption implies the \textit{unique cell} assumption. Let $\vecx \in W$ with $\mathsf{Ball}_{\vecx} \subseteq \mathsf{Univ}_{\veco}$ for some corner $\veco$. Since $\mathsf{Ball}_{\vecx} = [x_1, x_1+2\delta] \times \cdots \times [x_d, x_d+2\delta]$ and $\mathsf{Univ}_{\veco} = [o_1, o_1+4\delta) \times \cdots \times [o_d, o_d+4\delta)$, the containment implies $o_i \leq x_i < o_i + 2\delta$ for all $i \in [d]$. Since each $o_i$ is an integer multiple of $2\delta$, the region $[o_1, o_1+2\delta) \times \cdots \times [o_d, o_d+2\delta)$ is precisely the grid cell with lower-left corner $\veco$, and $\vecx$ lies in this cell. For any two distinct points $\vecx, \vecx' \in W$ with $\mathsf{Ball}_{\vecx} \subseteq \mathsf{Univ}_{\veco}$ and $\mathsf{Ball}_{\vecx'} \subseteq \mathsf{Univ}_{\veco'}$, the \textit{mini-universe} assumption gives $\mathsf{Univ}_{\veco} \neq \mathsf{Univ}_{\veco'}$, hence $\veco \neq \veco'$. By the above argument, $\vecx$ and $\vecx'$ lie in distinct cells, so $\mathsf{cell}_{2\delta}(\vecx) \neq \mathsf{cell}_{2\delta}(\vecx')$.
\end{proof}

The difference between the two assumptions stems solely from how the $L_\infty$ ball of a point is defined: our construction centers the ball at the point, whereas the structure-aware PSI~\cite{garimella2024computation,fss25} places the ball's lower-left corner at the point. Despite this difference in convention, the two assumptions impose the same distributional constraint on the input set.

\noindent\subsubsection{Comparison with $2\delta$-apart assumption.}
We recall the $2\delta$-apart assumption from~\cite{van2024fuzzy} as follows.

\begin{definition}[$2\delta$-apart]
A set $W \in \mathbb{U}^{n \times d}$ satisfies the \textit{$2\delta$-apart} assumption if for any two distinct points $\vecw, \vecw' \in W$, $\mathsf{dist}_{p}(\vecw, \vecw') \ge 2\delta d^{1/p}$. Specifically, if $p = \infty$, this reduces to $\mathsf{dist}_{\infty}(\vecw, \vecw') \ge 2\delta$.
\end{definition}

The following lemma from \cite{van2024fuzzy} illustrates that the \textit{$2\delta$-apart} assumption is strictly stronger than the \textit{unique cell} assumption. Intuitively, the \textit{unique cell} assumption permits the $L_\infty$ balls of input points to overlap as long as their centers fall into distinct cells, whereas the \textit{$2\delta$-apart} assumption ensures that these $L_\infty$ balls are pairwise disjoint.

\begin{lemma}[\cite{van2024fuzzy}]\label{lemma: 2delta}
Suppose there are multiple $L_p$ balls ($p \in [1,\infty]$) with radius $\delta$ lying in a $d$-dimensional space which is tiled by cells with side length $2\delta$. If these balls' centers are at least $2\delta d^{1/p}$ apart, then for each cell, there is at most one center of the balls lying in this cell. Specifically, if $p=\infty$, then the unique center holds for disjoint balls since $2\delta d^{1/p}$ degrades to $2\delta$ in this case.
\end{lemma}

\noindent\subsubsection{Comparison with $2\delta$-disjoint projection assumption.}
The \textit{$2\delta$-disjoint projection} assumption (abbreviated as $2\delta$-disj.\ proj.\ in Figure~\ref{Fig: assumptions}), adopted in~\cite{van2024fuzzy,gao2025efficient,dang2025ccs,van2025,garimella2022structure,yang2026}, requires that for every point $\vecw \in W$, there exists a dimension $k$ such that $|w_k - w'_k| \ge 2\delta$ for all $\vecw' \in W \setminus \{\vecw\}$. It is strictly stronger than the \textit{unique cell} assumption for general $L_p$ distances, since $|w_k - w'_k| \ge 2\delta$ for some dimension $k$ implies $\mathsf{cell}_{2\delta}(\vecw) \neq \mathsf{cell}_{2\delta}(\vecw')$.

To the best of our knowledge, all existing constructions relying on the \textit{$2\delta$-disjoint projection} assumption require it to hold for both parties' input sets, and no known technique extends it to the one-sided setting.

\subsection{Unique Block}

We compare the unique block assumption of Definition~\ref{def: unique block} with existing assumptions as follows.

\subsubsection{Comparison with unique cell assumption.}
We note that the \textit{unique block} assumption is strictly stronger than the \textit{unique cell} assumption, as it requires the $L_\infty$ balls to have no common cells and therefore their centers fall into different cells.

\subsubsection{Comparison with $4\delta$-apart assumption.}
We recall the $4\delta$-apart assumption from~\cite{van2024fuzzy} as follows.

\begin{definition}[$4\delta$-apart]
A set $W \in \mathbb{U}^{n \times d}$ satisfies the \textit{$4\delta$-apart} assumption if for any two distinct points $\vecw, \vecw' \in W$, $\mathsf{dist}_{p}(\vecw, \vecw') \ge 2\delta(d^{1/p}+1)$. In particular, for $p = \infty$, this reduces to $\mathsf{dist}_{\infty}(\vecw, \vecw') \ge 4\delta$.
\end{definition}

The following lemma from \cite{van2024fuzzy} illustrates that the \textit{$4\delta$-apart} assumption is strictly stronger than the \textit{unique block} assumption.

\begin{lemma}[\cite{van2024fuzzy}]\label{lemma: 4delta}
Suppose there are multiple $\delta$-radius $L_p$ balls ($p \in [1,\infty]$) distributed in a $d$-dimensional space which is tiled by cells of side length $2\delta$. If these balls' centers are at least $2\delta (d^{1/p}+1)$ apart from each other, then there exists at most one ball intersecting with the same cell. Specifically, if $p=\infty$, this holds for $L_\infty$ balls with $4\delta$-apart centers.
\end{lemma}

\section{Security Proof}
\label{appendix: Proof of Fuzzy PSI}

\subsection{Threat Model}
\label{appendix: Threat Model}
Similar to prior works~\cite{van2024fuzzy,dang2025ccs,gao2025efficient,van2025,fss25,yang2026}, we consider static semi-honest probabilistic polynomial-time (PPT) adversaries.
Namely, a PPT adversary \AAA passively corrupts either the sender \SSS or the receiver \RRR at the beginning of the protocol and honestly follows the protocol specification.
We use the standard simulation-based security definition for secure two-party computation.
Our construction invokes multiple sub-protocols, and we use the \textit{hybrid model} to describe them.
By convention, a protocol invoking a functionality \FFF is referred to as the \FFF-hybrid model.
We give the formal security definition as follows.

\begin{definition}
Let $\mathsf{view}_{\SSS}^{\Pi}(x, y)$ and $\mathsf{view}_{\RRR}^{\Pi}(x, y)$ be the views of $\SSS$ and $\RRR$ in a protocol $\Pi$, respectively, where $x$ is the input of $\SSS$ and  $y$ is the input of $\RRR$. Let $\mathsf{out}(x, y)$ be the protocol's output of both parties and $\mathcal{F}(x, y)$ be the functionality's output. $\Pi$ is said to securely compute a functionality $\mathcal{F}$ in the semi-honest model if for every PPT adversary \AAA there exists PPT simulators $\mathsf{Sim}_{\SSS}$ and $\mathsf{Sim}_{\RRR}$ such that for all inputs $x$ and $y$,
\begin{equation}
\begin{aligned}
    &\{\mathsf{view}_{\SSS}^{\Pi}(x, y), \mathsf{out}(x, y)\} \approx_c \{\mathsf{Sim}_{\SSS}(x, \mathcal{F}_{\SSS}(x, y)), \mathcal{F}(x, y)\}, \\
    &\{\mathsf{view}_{\RRR}^{\Pi}(x, y), \mathsf{out}(x, y)\} \approx_c \{\mathsf{Sim}_{\RRR}(y, \mathcal{F}_{\RRR}(x, y)), \mathcal{F}(x, y)\}. \nonumber
\end{aligned}
\end{equation}
\end{definition}

\subsection{$L_\infty$ Distance under Sender-sided Unique Cell Assumptions}
\label{appendix: proof sender 2delta Linf}

\begin{proof}

    \textbf{Correctness.} We first consider that for some $\vecq_j \in Q$, there exists $\vecw_i \in W$ such that $\mathsf{dist}_{\infty}(\vecq_j,\vecw_i)\le \delta$. By the correctness of spatial hashing, there exists an index $z^*\in[2^d]$ such that $\CCC^\RRR_{i,z^*}=\CCC^\SSS_{j}$. Since $\mathsf{dist}_{\infty}(\vecq_j,\vecw_i)\le \delta$, it holds that $\abs{q_{j,k}-w_{i,k}}\le \delta$ for all $k\in[d]$. Therefore, for each $k\in[d]$, there exists $t_k\in[-\delta,\delta]$ such that $q_{j,k} + t_k=w_{i,k}$. By the correctness of functionality \Func[so\text{-}OPPRF],  $e^\RRR_{i,k,z^*}\Vert r^\RRR_{i,k,z^*}$  and $e^\SSS_{i,k,z^*}\Vert r^\SSS_{i,k,z^*}$ are the shares of $0^\ell\Vert s_{j,k}$. And then the aggregated shares $e^\RRR_{i,z^*}$ and $e^\SSS_{i,z^*}$ are the shares of $0^\ell$, $r^\RRR_{i,z^*}$ and $r^\SSS_{i,z^*}$ are the shares of $s_j = \sum_{k\in[d]}s_{j,k}$. Since $e^\RRR_{i,z^*}+e^\SSS_{i,z^*}=0$, by the correctness of functionality $\FFF^{\mathsf{Eq}}_{\mathsf{ConRand}}$, it holds that $v^\RRR_{i,z^*}+v^\SSS_{i,z^*}=s_j$. Then, the \RRR reconstructs $v_{i,z^*}=s_j$. After shuffling, there must be $x_{j^*}\in X$ and $x_{j^*} = H_{\lambda'+du}(s_j)\oplus(0^{\lambda'}\Vert \vecq_j)$. Finally, the \RRR unmasks $x_{j^*}$ by $H_{\lambda'+du}(v_{i,z^*})\oplus x_{j^*}$ and correctly gets $\vecq_j$.

    We then consider that for some $\vecq_j \in Q$, it holds $\mathsf{dist}_{\infty}(\vecq_j,\vecw_i) > \delta$ for any $\vecw_i \in W$. We consider the following two cases. 
    
    (1) The first case, there exists an index $z^*\in[2^d]$ such that $\CCC^\RRR_{i,z^*}=\CCC^\SSS_{j}$. Since $\mathsf{dist}_{\infty}(\vecq_j,\vecw_i) > \delta$, it holds that there exists an index $k^*\in[d]$ such that $\abs{q_{j,k^*}-w_{i,k^*}}>\delta$. By the correctness of functionality \Func[so\text{-}OPPRF], the outputs $e^\RRR_{i,k^*,z^*}\Vert r^\RRR_{i,k^*,z^*}$ and $e^\SSS_{i,k^*,z^*}\Vert r^\SSS_{i,k^*,z^*}$ are uniformly random. Then, after aggregation, the shares $e^\RRR_{i,z^*}\Vert r^\RRR_{i,z^*}$ and $e^\SSS_{i,z^*}\Vert r^\SSS_{i,z^*}$ are also uniformly random. By the correctness of functionality $\FFF^{\mathsf{Eq}}_{\mathsf{ConRand}}$, the new shares $v^\RRR_{i,z^*}$ and $v^\SSS_{i,z^*}$ are uniformly random. After reconstruction, $v_{i,z*}$ is uniformly random.
    
    (2) The second case, for any $z\in[2^d]$, it holds that $\CCC^\RRR_{i,z}\neq \CCC^\SSS_{j}$. By the correctness of functionality \Func[so\text{-}OPPRF], the outputs $e^\RRR_{i,k,z}\Vert r^\RRR_{i,k,z}$ and $e^\SSS_{i,k,z}\Vert r^\SSS_{i,k,z}$ are uniformly random for all $z\in [2^d],k\in [d]$. Similarly, $v_{i,z}$ is uniformly random.

    In both cases, $v_{i,z}$ is uniformly random for all $i\in[n],z\in[2^d]$. By setting the bit length of $s_j$ as $\kappa > \lambda +d+ \log mn$, the probability of $v_{i,z}=s_j$ for some $j\in[m]$ is at most $2^{-\lambda}$, which is negligible. On the other hand, if $v_{i,z}\neq s_j$, the decryption $H_{\lambda^\prime+du}(v_{i,z}) \oplus x_{j}$ is uniformly random for all $i\in[n],j\in[m],z\in[2^d]$. By setting $\lambda^\prime = \lambda + d + \log mn$, the probability of valid decryption is at most $2^{-\lambda}$, which is negligible. A union bound shows that the receiver can obtain $\vecq_j$ with the probability at most $2^{-\lambda}(2-2^{-\lambda})<2^{-\lambda+1}$, which is negligible. 
    This completes the correctness proof.
    
\textbf{Corrupted sender.}
The simulator $\Sim[\SSS]$ receives the sender's input $Q$ and proceeds as follows:
\begin{enumerate}

    \item \Sim[\SSS] maintains a table for queries of the random oracle $H$. For each query to $H$, \Sim[\SSS] returns a consistent random value.

    \item $\Sim[\SSS]$ emulates $\Func[so\text{-}OPPRF]$ by sampling uniformly random shares $e_{i,k,z}^\SSS, r_{i,k,z}^\SSS$ for $i \in [n], k \in [d], z \in [2^d]$, and sampling a random function $F$ such that $F(\CCC^{\SSS}_{j}\Vert k \Vert q_{j,k} + t) = {0^\ell} \Vert s_{j,k}$ for $j\in[m],k\in[d], t\in [-\delta,\delta]$, and sends $e_{i,k,z}^\SSS, r_{i,k,z}^\SSS$ and $\OOO^F$ to $\AAA$.
    
    \item $\Sim[\SSS]$ emulates \ensuremath{\FFF^{\mathsf{Eq}}_{\mathsf{ConRand}}} by sampling uniformly random shares $v_{i,z}^\SSS$ for $i \in [n], z \in [2^d]$, and sends them to $\AAA$.
\end{enumerate}

\textit{Indistinguishability.}
The adversary's view in the real execution is identically distributed to its view in the simulated execution. This follows because the values $e_{i,k,z}^\SSS, r_{i,k,z}^\SSS$ received from $\Func[so\text{-}OPPRF]$ and $v_{i,z}^\SSS$ from \ensuremath{\FFF^{\mathsf{Eq}}_{\mathsf{ConRand}}} are uniformly random in the real world and are sampled from the same distribution by the simulator in the ideal world. Moreover, the random function oracle $\OOO^F$ subject to $F(\CCC^{\SSS}_{j}\Vert k \Vert q_{j,k} + t) = {0^\ell} \Vert s_{j,k}$, both of which are independent of \SSS's actual input $Q$.
Therefore, the view in the real-world execution is identically distributed to its view in the simulation.

\textbf{Corrupted receiver.} 
We show the simulator $\Sim[\RRR](W, I)$: 
\begin{enumerate}

    \item \Sim[\RRR] maintains a table for queries of the random oracle $H$. For each query to $H$, \Sim[\RRR] returns a consistent random value.

    \item $\Sim[\RRR]$ emulates $\Func[so\text{-}OPPRF]$ by sampling uniformly random shares $e_{i,k,z}^\RRR, r_{i,k,z}^\RRR$ for $i \in [n], k \in [d], z \in [2^d]$, and sends $e_{i,k,z}^\RRR, r_{i,k,z}^\RRR$ to $\AAA$.
    
    \item $\Sim[\RRR]$ emulates \ensuremath{\FFF^{\mathsf{Eq}}_{\mathsf{ConRand}}} by sampling uniformly random shares $v_{i,z}^\RRR$ for $i \in [n], z \in [2^d]$, and sends them to $\AAA$.
    
    \item For each $\vecq_j \in I$, \Sim[\RRR] samples uniformly random $s_j$ and sets $v_{i,z}^\SSS = s_j - v_{i,z}^\RRR$ for all $i \in [n]$ and $z \in [2^d]$ such that $\CCC_{i, z}^\RRR = \CCC_j^\SSS$ and $\mathsf{dist}(\vecq_j, \vecw_i) \leq \delta$. For the remaining $v_{i,z}^\SSS$'s, \Sim[\RRR] samples uniformly random values. \Sim[\RRR] sends these $v_{i,z}^\SSS$ to \AAA.
    
    \item \Sim[\RRR] initializes an empty set $X$. For each $\vecq_j \in I$, where there exists $\vecw_i$ and $z \in [2^d]$ such that $\CCC_{i, z}^\RRR = \CCC_j^\SSS$ and $\mathsf{dist}(\vecq_j, \vecw_i) \leq \delta$, it computes $x = H(s_j) \oplus (0^{\lambda'} \| \vecq_j)$ and adds $x$ to $X$. For the remaining $m - |I|$ points, \Sim[\RRR] samples uniformly random strings of length $\lambda' + d \cdot u$ and adds them to $X$. \Sim[\RRR] shuffles $X$ and sends them to \AAA.
    
\end{enumerate}

\textit{Indistinguishability.} We prove security via a sequence of hybrids.

\begin{description}
    \item[Hybrid $\mathsf{H}_0$.]
    This is the real execution of the protocol.

    \item[Hybrid $\mathsf{H}_1$.] 
    This hybrid is identical to $\mathsf{H}_0$, except that \Sim[\RRR] acts as \Func[so\text{-}OPPRF] and \ensuremath{\FFF^{\mathsf{Eq}}_{\mathsf{ConRand}}}, but follows the specification of these functionalities and uses the real sender's input $Q$. This is only a conceptual change, and hence the adversary's view in Hybrid $\mathsf{H}_0$ and Hybrid $\mathsf{H}_1$ are identically distributed.

    \item[Hybrid $\mathsf{H}_2$.] 
    This hybrid is identical to $\mathsf{H}_1$, except for the following modifications. 
    For each $\vecq_j \in I$, \Sim[\RRR] sets $v_{i,z}^\SSS = s_j - v_{i,z}^\RRR$ for all $i \in [n]$ and $z \in [2^d]$ such that $\CCC_{i, z}^\RRR = \CCC_j^\SSS$ and $\mathsf{dist}(\vecq_j, \vecw_i) \leq \delta$. For the remaining $v_{i,z}^\SSS$'s, \Sim[\RRR] samples uniformly random values. Besides, for each $\vecq_j \in I$, 
    it computes $x = H(s_j) \oplus (0^{\lambda'} \| \vecq_j)$ and adds $x$ to $X$. For the remaining $m - |I|$ points, \Sim[\RRR] samples uniformly random strings of length $\lambda' + d \cdot u$ and adds them to $X$. 

    In $\mathsf{H}_1$, for each $\vecw_i$ and $z\in[2^d]$ such that there is no $\vecq_j \in Q$ satisfying $\CCC_{i, z}^\RRR = \CCC_j^\SSS$ and $\mathsf{dist}(\vecq_j, \vecw_i) \leq \delta$, there exists at least one $k^*_i$ such that $e_{i, k^*_i, z}^\SSS$ and $e_{i, k^*_i, z}^\RRR$ are uniformly random by the correctness of so-OPPRF. This implies $e_{i, z}^\SSS$ and $e_{i, z}^\RRR$ are uniformly random. By setting the bit length as $\ell = \lambda + d + \log n$, a union bound shows that $e_{i, z}^\SSS \neq - e_{i, z}^\RRR$ with overwhelming probabilities. According to the correctness of \ensuremath{\FFF^{\mathsf{Eq}}_{\mathsf{ConRand}}}, $v_{i,z}^\SSS$ for the above $\vecw_i$ and $z\in[2^d]$ is uniformly random and hence identical in both worlds. 
    Besides, the above analysis implies that $s_{j}$ for $\vecq_j \notin I$ is hidden from \AAA.
    Moreover, due to $|s_{j}| = \kappa$, the probability that the adversary queries $H$ on $s_j$ is negligible. 
    Then, for each $\vecq_j \notin I$, $H(s_j) \xor (0^{\lambda'} \| \vecq_j)$ yields a string identically distributed to a uniform random string.
    Therefore, the adversary's view of Hybrid $\mathsf{H}_2$ is computationally indistinguishable from that of Hybrid $\mathsf{H}_1$.

    \item[Hybrid $\mathsf{H}_3$.] 
    This hybrid is identical to $\mathsf{H}_2$, except that \Sim[\RRR] simulates \Func[so\text{-}OPPRF] and \ensuremath{\FFF^{\mathsf{Eq}}_{\mathsf{ConRand}}} as follows. In particular, \Sim[\RRR] samples uniformly random shares $e_{i,k,z}^\RRR, r_{i,k,z}^\RRR$, $v_{i,z}^\RRR$ for $i \in [n], k \in [d], z \in [2^d]$, and sends them to \AAA.
    
    Since they are uniformly random shares from the adversary's view, the adversary's view in Hybrid $\mathsf{H}_2$ and Hybrid $\mathsf{H}_3$ are identically distributed.
    
\end{description}

Observe that Hybrid $\mathsf{H}_3$ is exactly the simulated execution produced by $\Sim[\RRR]$.
Therefore, the real and ideal executions are computationally indistinguishable, completing the proof.

\end{proof}

\subsection{$L_p$ Distance under Sender-sided Unique Cell Assumptions}
\label{appendix: proof sender 2delta Lp}

\begin{proof}

\textbf{Correctness.} We first consider that for some $\vecq_j \in Q$, there exists $\vecw_i \in W$ such that $\mathsf{dist}_{p}(\vecq_j,\vecw_i)\le \delta$ and therefore $\mathsf{dist}_{\infty}(\vecq_j,\vecw_i)\le \delta$. By the correctness of spatial hashing, there exists an index $z^*\in[2^d]$ such that $\CCC^\RRR_{i,z^*}=\CCC^\SSS_{j}$. Since $\mathsf{dist}_{p}(\vecq_j,\vecw_i)\le \delta$, it holds that $\abs{q_{j,k}-w_{i,k}}\le \delta$ for all $k\in[d]$. Therefore, for each $k\in[d]$, there exists $t_k\in[-\delta,\delta]$ such that $q_{j,k} + t_k=w_{i,k}$. By the correctness of functionality \Func[so\text{-}OPPRF],  $e^\RRR_{i,k,z^*}\Vert r^\RRR_{i,k,z^*}$  and $e^\SSS_{i,k,z^*}\Vert r^\SSS_{i,k,z^*}$ are the shares of $\abs{t_k}^p \Vert s_{j,k}$. And by the correctness of functionality \Func[B2A], $d^\RRR_{i,z^*}$ and $d^\SSS_{i,z^*}$ are the shares of $\sum_{k\in[d]}\abs{t_k}^p = \sum_{k\in[d]}(\abs{q_{j,k}-w_{i,k}})^p$. Since $\sum_{k\in[d]}(\abs{q_{j,k}-w_{i,k}})^p\le \delta^p$, by the correctness of functionality $\FFF^{\mathsf{Cmp}}_{\mathsf{ConRand}}$, it holds that $v^\RRR_{i,z^*}+v^\SSS_{i,z^*}=s_j$. Then, the \RRR reconstructs $v_{i,z^*}=s_j$. After shuffling, there must be $x_{j^*}\in X$ and $x_{j^*} = H_{\lambda'+du}(s_j)\oplus(0^{\lambda'}\Vert \vecq_j)$. Finally, the \RRR unmasks $x_{j^*}$ by $H_{\lambda'+du}(v_{i,z^*})\oplus x_{j^*}$ and correctly gets $\vecq_j$.

We then consider that for some $\vecq_j \in Q$, it holds $\mathsf{dist}_{p}(\vecq_j,\vecw_i) > \delta$ for any $\vecw_i \in W$. We consider the following three cases. 

(1) The first case: for any $z\in[2^d]$, it holds that $\CCC^\RRR_{i,z} \neq \CCC^\SSS_{j}$. By the correctness of functionality \Func[so\text{-}OPPRF], the outputs $e^\RRR_{i,k,z}\Vert r^\RRR_{i,k,z}$ and $e^\SSS_{i,k,z}\Vert r^\SSS_{i,k,z}$ are uniformly random for all $z\in [2^d],k\in [d]$. By the correctness of \Func[B2A], after aggregation, the shares $d^\RRR_{i,z}$ and $d^\SSS_{i,z}$ are also uniformly random. 
By the correctness of $\FFF^{\mathsf{Cmp}}_{\mathsf{ConRand}}$, the new shares $v^\RRR_{i,z}$ and $v^\SSS_{i,z}$ are uniformly random. Therefore, after reconstruction, $v_{i,z}$ is uniformly random.

(2) The second case: there exists an index $z^*\in[2^d]$ such that $\CCC^\RRR_{i,z^*}=\CCC^\SSS_{j}$ but $\mathsf{dist}_{\infty}(\vecq_j,\vecw_i)>\delta$. Then there exists an index $k^*\in[d]$ such that $\abs{q_{j,k^*}-w_{i,k^*}} > \delta$. By the correctness of functionality \Func[so\text{-}OPPRF], the outputs $e^\RRR_{i,k^*,z^*}\Vert r^\RRR_{i,k^*,z^*}$ and $e^\SSS_{i,k^*,z^*}\Vert r^\SSS_{i,k^*,z^*}$ are uniformly random. By the correctness of \Func[B2A], after aggregation, the shares $d^\RRR_{i,z^*}$ and $d^\SSS_{i,z^*}$ are also uniformly random. Similarly, $v_{i,z^*}$ is uniformly random.

(3) The third case: there exists an index $z^*\in[2^d]$ such that $\CCC^\RRR_{i,z^*}=\CCC^\SSS_{j}$ and $\mathsf{dist}_{\infty}(\vecq_j,\vecw_i)\le\delta$. By the correctness of \Func[so\text{-}OPPRF] and \Func[B2A], it holds that $d^\RRR_{i,z^*}+d^\SSS_{i,z^*} = {\sum}_{k\in[d]} (\abs{w_{i,k}-q_{j,k}})^p$.
Since $\mathsf{dist}_{p}(\vecq_j,\vecw_i) > \delta$, therefore $d^\RRR_{i,z^*}+d^\SSS_{i,z^*}>\delta^p$.
By the correctness of functionality $\FFF^{\mathsf{Cmp}}_{\mathsf{ConRand}}$, the outputs $v^\RRR_{i,z^*}$ and $v^\SSS_{i,z^*}$ are uniformly random. After reconstruction, $v_{i,z*}$ is also uniformly random.

In all cases, $v_{i,z}$ is uniformly random for all $i\in[n],z\in[2^d]$. By setting the bit length of $s_j$ as $\kappa > \lambda +d+ \log nm$, the probability of $v_{i,z}=s_j$ for some $j\in[m]$ is at most $2^{-\lambda}$, which is negligible. On the other hand, if $v_{i,z}\neq s_j$, the decryption $H_{\lambda^\prime+du}(v_{i,z}) \oplus x_{j}$ is uniformly random for all $i\in[n],j\in[m],z\in[2^d]$. By setting $\lambda^\prime = \lambda + d + \log mn$, the probability of valid decryption is at most $2^{-\lambda}$, which is negligible. A union bound shows that the receiver can obtain valid $\vecq_j$ with the probability at most $2^{-\lambda}+(1-2^{-\lambda})2^{-\lambda}<2^{-\lambda+1}$, which is negligible.

This completes the correctness proof.

\end{proof}

\subsection{$L_\infty$ Distance under Receiver-sided Unique Cell Assumptions}
\label{appendix: proof receiver 2delta Linf}

\begin{proof}

\textbf{Correctness.} We first consider that for some $\vecq_j \in Q$, there exists $\vecw_i \in W$ such that $\mathsf{dist}_{\infty}(\vecq_j,\vecw_i)\le \delta$. By the correctness of spatial hashing, there exists an index $z^*\in[2^d]$ such that $\CCC^\SSS_{j,z^*}=\CCC^\RRR_{i}$. Since $\mathsf{dist}_{\infty}(\vecq_j,\vecw_i)\le \delta$, it holds that $\abs{q_{j,k}-w_{i,k}}\le \delta$ for all $k\in[d]$. Therefore, for each $k\in[d]$, there exists $t_k\in[-\delta,\delta]$ such that $\CCC^\SSS_{j}\Vert k\Vert q_{j,k} + t_k=w_{i,k}$. By the correctness of functionality \Func[so\text{-}OPPRF],  $e^\RRR_{j,k,z^*}$ and $e^\SSS_{j,k,z^*}$ are the shares of $0^\ell$. And then the aggregated shares $r^\RRR_{j,z^*}$ and $r^\SSS_{j,z^*}$ are still the shares of $0^\ell$. By the correctness of \Func[Eq], it holds that $b_{j,z^*} =1$ and therefore $b_j = 1$.
By the correctness of \Func[OT], \RRR correctly gets $\vecq_j$. 

We then consider that for some $\vecq_j\in Q$, it holds $\mathsf{dist}_{\infty}(\vecq_j,\vecw_i)> \delta$ for any $\vecw_i \in W$. We consider the following two cases. 

(1) The first case: there exists an index $z^*$ such that $\CCC^\SSS_{j,z^*} = \CCC^\RRR_i$. Since $\mathsf{dist}_{\infty}(\vecq_j,\vecw_i)> \delta$, it holds that there exists an index $k^*\in [d]$ such that $\abs{q_{j,k^*}-w_{i,k^*}}>\delta$. By the correctness of functionality \Func[so\text{-}OPPRF], the outputs $e^\RRR_{j,k^*,z^*}$ and $e^\SSS_{j,k^*,z^*}$ are uniformly random. Then, after aggregation, $r^\RRR_{j,z^*}$ and $r^\SSS_{j,z^*}$ are also uniformly random. 

(2) The second case: for any $z\in [2^d]$, it holds that $\CCC^\SSS_{j,z} \neq \CCC^\RRR_i$. By the correctness of functionality \Func[so\text{-}OPPRF], the outputs $e^\RRR_{j,k,z}$ and $e^\SSS_{j,k,z}$ are uniformly random. Then, $r^\RRR_{j,z}$ and $r^\SSS_{j,z}$ are also uniformly random.

In both cases, $r^\RRR_{j,z}$ and $r^\SSS_{j,z}$ are uniformly random for all $j\in [m],z\in [2^d]$. By the correctness of functionality \Func[Eq], the probability of $b_{j,z}=1$ is at most $1/\abs{\FF}$. A union bound shows that the probability of there existing $j\in [m]$ such that $b_{j}=1$ is at most $m2^d/\abs{\FF}$. By setting $  \log \abs{\FF} =\ell= \lambda+d+\log m$, the probability $m2^d/\abs{\FF} = 2^{-\lambda}$ is negligible. By the correctness of \Func[OT], \RRR learns nothing about $\vecq_j$. This completes the correctness proof.

\textbf{Corrupted sender.}
The simulator $\Sim[\SSS]$ receives the sender's input $Q$ and proceeds as follows:
\begin{enumerate}
    \item $\Sim[\SSS]$ emulates $\Func[so\text{-}OPPRF]$ by sampling uniformly random shares $e_{j,k,z}^{\SSS}$ for all $j \in [m]$, $k \in [d]$, and $z \in [2^d]$, and sends them to $\AAA$.

    \item $\Sim[\SSS]$ emulates $\Func[Eq]$ and $\Func[OT]$, both of which output nothing to $\AAA$.
\end{enumerate}

\textit{Indistinguishability.}
The adversary's view in the real execution is identically distributed to its view in the simulated execution. This follows because the values $e_{j,k,z}^{\SSS}$ received from $\Func[so\text{-}OPPRF]$ are uniformly random in the real world and are sampled from the same distribution by the simulator in the ideal world.

\textbf{Corrupted receiver.}
The simulator $\Sim[\RRR]$ receives the receiver's input $W$ and output $I$, and proceeds as follows:
\begin{enumerate}
    \item $\Sim[\RRR]$ emulates $\Func[so\text{-}OPPRF]$ by sampling uniformly random shares $e_{j,k,z}^{\RRR}$ for all $j \in [m]$, $k \in [d]$, and $z \in [2^d]$, and by sampling a random function $F$ such that $F(\CCC^{\RRR}_{i}\Vert k \Vert w_{i,k}+t)=0$ for all $i\in[n]$, $k\in[d]$, and $t\in[-\delta,\delta]$. It then sends these shares $e_{j,k,z}^{\RRR}$ and $\OOO^F$ to $\AAA$.

    \item $\Sim[\RRR]$ emulates $\Func[Eq]$ by padding $I$ with $\bot$ to size $m$ and randomly permuting the padded $I$. For each $j \in [m]$ and $z \in [2^d]$, it sets
    \[
    b_{j,z} =
    \begin{cases}
        1, & \text{if } \vecq_j = I[j] \neq \bot \text{ and there exists } \vecw_i \text{ such that } \\
        & \CCC_i^{\RRR}=\CCC_{j,z}^{\SSS} 
            \text{ and } \mathsf{dist}(\vecq_j,\vecw_i)\le \delta,\\
        0, & \text{otherwise.}
    \end{cases}
    \]
    $\Sim[\RRR]$ sends all values $b_{j,z}$ to $\AAA$.

    \item $\Sim[\RRR]$ emulates $\Func[OT]$ and sends the padded and permuted $I$ to $\AAA$.
\end{enumerate}

\textit{Indistinguishability.} We prove security via a sequence of hybrids.

\begin{description}
    \item[Hybrid $\mathsf{H}_0$.]
    This is the real execution of the protocol.

    \item[Hybrid $\mathsf{H}_1$.] 
    This hybrid is identical to $\mathsf{H}_0$, except that \Sim[\RRR] acts as \Func[so\text{-}OPPRF], \Func[Eq], and \Func[OT], but follows the specification of these functionalities and uses the real sender's input $Q$. This is only a conceptual change, and hence the adversary's view in Hybrid $\mathsf{H}_0$ and Hybrid $\mathsf{H}_1$ are identically distributed.

    \item[Hybrid $\mathsf{H}_2$.]
    This hybrid is identical to $\mathsf{H}_1$, except that instead of computing $b_{j, z}$'s (resp. $\vecu_j$'s) according to \Func[Eq] (resp. \Func[OT]), \Sim[\RRR] derives them utilizing the receiver's input $W$ and output $I$. Specifically, it pads $I$ with $\bot$ to size $m$, applies a random permutation,
    and for every $j\in[m],z\in[2^d]$ sets
    \[
    b_{j,z}=
    \begin{cases}
        1, &\text{if } \vecq_j = I[j] \neq \bot \text{ and there exists }\vecw_i
        \text{ such that } \\ &\CCC_i^{\RRR}=\CCC_{j,z}^{\SSS}
        \text{ and }\mathsf{dist}(\vecq_j,\vecw_i)\le\delta,\\
        0,&\text{otherwise.}
    \end{cases}
    \]
    The values $b_{j, z}$'s (resp.  $\vecu_j$'s) are returned to $\AAA$ as the outputs of $\Func[Eq]$ (resp. $\Func[OT]$).

    The adversary's views in Hybrid $\mathsf{H}_1$ and Hybrid $\mathsf{H}_2$ are identically distributed, because
    functionality Eq outputs exactly these bits $b_{j,z}$ and the outputs of functionality OT are exactly determined by $b_{j,z}$ and $I$.

    \item[Hybrid $\mathsf{H}_3$.]
    This hybrid is identical to $\mathsf{H}_2$, except that we replace the sender's interaction in $\Func[so\text{-}OPPRF]$ with a simulation as follows:
    sample random shares
    $e_{j,k,z}^{\RRR}$ for all $j\in[m],k\in[d],z\in[2^d]$,
    and sample a random function $F$ conditioned on $F(\CCC_i^{\RRR}\Vert k\Vert w_{i,k}+t)=0$
    for all $i\in[n]$, $k\in[d]$, and $t\in[-\delta,\delta]$.
    The simulator sends ${e_{j,k,z}^{\RRR}}$ and $\OOO^F$ to $\AAA$.

    By the security of $\Func[so\text{-}OPPRF]$, the adversarial views of Hybrid $\mathsf{H}_2$ and Hybrid $\mathsf{H}_3$ are identically distributed.
\end{description}

Observe that Hybrid $\mathsf{H}_3$ is exactly the simulated execution produced by $\Sim[\RRR]$.
Therefore, the real and ideal executions are perfectly indistinguishable, completing the proof.

\end{proof}

\subsection{$L_p$ Distance under Receiver-sided Unique Cell Assumptions}
\label{appendix: proof receiver 2delta Lp}

\begin{proof}
\textbf{Correctness.} We first consider that for some $\vecq_j \in Q$, there exists $\vecw_i \in W$ such that $\mathsf{dist}_{p}(\vecq_j,\vecw_i)\le \delta$. By the correctness of spatial hashing, there exists an index $z^*\in[2^d]$ such that $\CCC^\SSS_{j,z^*}=\CCC^\RRR_{i}$. Since $\mathsf{dist}_{p}(\vecq_j,\vecw_i)\le \delta$, it holds that $\abs{q_{j,k}-w_{i,k}}\le \delta$ for all $k\in[d]$. Therefore, for each $k\in[d]$, there exists $t_k\in[-\delta,\delta]$ such that $q_{j,k} + t_k=w_{i,k}$. By the correctness of functionality \Func[so\text{-}OPPRF],  $e^\RRR_{j,k,z^*}$ and $e^\SSS_{j,k,z^*}$ are the shares of $\abs{t_k}^p$. By the correctness of \Func[B2A], $d^\RRR_{j,k,z^*}+d^\SSS_{j,k,z^*} = \abs{t_k}^p$. Then, after aggregation, $r^\RRR_{j,z^*}+r^\SSS_{j,z^*} =\sum_{k\in[d]}\abs{t_k}^p = \sum_{k\in[d]}\abs{q_{j,k}-w_{i,k}}^p$. Since $\mathsf{dist}_p(\vecq_j,\vecw_i) \le \delta$, by the correctness of $\FFF^{\delta^p}_{\mathsf{Interval}}$, it holds that $b_{j,z^*} =1$ and therefore $b_j = 1$.
By the correctness of \Func[OT], \RRR correctly gets $\vecq_j$. 

We then consider that for some $\vecq_j\in Q$, it holds $\mathsf{dist}_{p}(\vecq_j,\vecw_i)> \delta$ for any $\vecw_i \in W$. We consider the following three cases. 

(1) The first case, for any $z\in [2^d]$, it holds that $\CCC^\SSS_{j,z} \neq \CCC^\RRR_i$. By the correctness of functionality \Func[so\text{-}OPPRF], the outputs $e^\RRR_{j,k,z}$ and $e^\SSS_{j,k,z}$ are uniformly random for all $k\in[d],z\in[2^d]$. 
By the correctness of \Func[B2A], the shares $d^\RRR_{j,k,z}$ and $d^\SSS_{j,k,z}$ are still uniformly random.
Then, after aggregation, $r^\RRR_{j,z}$ and $r^\SSS_{j,z}$ are also uniformly random.

(2) The second case, there exists an index $z^*$ such that $\CCC^\SSS_{j,z^*} = \CCC^\RRR_i$ but $\mathsf{dist}_{\infty}(\vecq_j,\vecw_i) > \delta$. Therefore, it holds that there exists an index $k^*\in [d]$ such that $\abs{q_{j,k^*}-w_{i,k^*}}>\delta$. By the correctness of functionality \Func[so\text{-}OPPRF], the outputs $e^\RRR_{j,k^*,z^*}$ and $e^\SSS_{j,k^*,z^*}$ are uniformly random. By the correctness of functionality \Func[B2A], $d^\RRR_{j,k^*,z^*}$ and $d^\SSS_{j,k^*,z^*}$ are also uniformly random. Then, after aggregation, $r^\RRR_{j,z^*}$ and $r^\SSS_{j,z^*}$ are also uniformly random. 

In the first two cases, $r^\RRR_{j,z}$ and $r^\SSS_{j,z}$ are uniformly random for all $j\in [m],z\in [2^d]$. By the correctness of functionality $\FFF^{\delta^p}_{\mathsf{Interval}}$, the probability of $b_{j,z}=1$ is at most $\delta^p/\abs{\FF}$. A union bound shows that the probability of there existing $j\in [m]$ such that $b_{j}=1$ is at most $m2^d\delta^p/\abs{\FF}$. By setting $\log\abs{\FF} = \ell = \lambda+d+\log m + p\log\delta$, the probability $m2^d\delta^p/\abs{\FF} = 2^{-\lambda}$ is negligible. By the correctness of \Func[OT], \RRR learns nothing about $\vecq_j$.

(3) The third case, there exists an index $z^*$ such that $\CCC^\SSS_{j,z^*} = \CCC^\RRR_i$ and $\mathsf{dist}_{\infty}(\vecq_j,\vecw_i) \le \delta$. By the correctness of functionality \Func[so\text{-}OPPRF] and \Func[B2A], $r^\RRR_{j,z^*}+r^\SSS_{j,z^*}=\mathsf{dist}_p(\vecq_j,\vecw_i)^p > \delta^p$. By the correctness of functionality $\FFF^{\delta^p}_{\mathsf{Interval}}$ and \Func[OT], \RRR learns nothing about $\vecq_j$. 

This completes the correctness proof.

\end{proof}

\subsection{$L_\infty$ Distance under Receiver-sided Unique Cell Assumptions with Prefix Optimization}
\label{appendix: proof receiver 4delta Linf prefix}

\begin{proof}

\textbf{Correctness.} We first consider that for some $\vecq_j \in Q$, there exists $\vecw_i \in W$ such that $\mathsf{dist}_{\infty}(\vecq_j,\vecw_i)\le \delta$. By the correctness of spatial hashing, there exists an index $z^*\in[2^d]$ such that $\CCC^\SSS_{j,z^*}=\CCC^\RRR_{i}$. Since $\mathsf{dist}_{\infty}(\vecq_j,\vecw_i)\le \delta$, it holds that $\abs{q_{j,k}-w_{i,k}}\le \delta$ for all $k\in[d]$. 

Therefore, for each $k\in[d]$, $q_{j,k} \in [w_{i,k}-\delta, w_{i,k}+\delta]$. By Lemma~\ref{lemma: prefix}, there exists exactly one index $h^*\in [\mu]$ such that $\tilde{q}_{j,k,h^*}\in \{ \tilde{w}_{i,k,h}\}_{h\in[\mu]}$.
By the correctness of functionality \Func[so\text{-}OPPRF],  $e^\RRR_{j,k,z^*,h^*}$ and $e^\SSS_{j,k,z^*,h^*}$ are the shares of $0^\ell$. By the correctness of functionality $\FFF^{\mathsf{Eq},\mu}_{\mathsf{ConSel}}$, $v^\RRR_{j,k,z^*}+v^\SSS_{j,k,z^*}=e^\RRR_{j,k,z^*,h^*}+e^\SSS_{j,k,z^*,h^*}=0$. Then, after aggregation, $r^\RRR_{j,z^*}+r^\SSS_{j,z^*}=0$. By the correctness of functionality \Func[Eq], $b_{j,z^*}=1$ and therefore $b_j=1$. By the correctness of functionality \Func[OT], \RRR correctly obtains $\vecq_j$.

We then consider that for some $\vecq_j\in Q$, it holds $\mathsf{dist}_{\infty}(\vecq_j,\vecw_i)> \delta$ for any $\vecw_i \in W$. We consider the following two cases. 

(1) The first case, there exists an index $z^*$ such that $\CCC^\SSS_{j,z^*} = \CCC^\RRR_i$. Since $\mathsf{dist}_{\infty}(\vecq_j,\vecw_i)> \delta$, it holds that there exists an index $k^*\in [d]$ such that $\abs{q_{j,k^*}-w_{i,k^*}}>\delta$, therefore $q_{j,k^*} \notin [w_{i,k^*}-\delta,w_{i,k^*}+\delta]$. By Lemma~\ref{lemma: prefix}, it holds that $\{ \tilde{q}_{j,k,h} \}_{h\in[\mu]} \cap \{ \tilde{w}_{i,k,h}\}_{h\in[\mu]} = \emptyset$. By the correctness of functionality \Func[so\text{-}OPPRF], the outputs $e^\RRR_{j,k^*,z^*,h}$ and $e^\SSS_{j,k^*,z^*,h}$ are uniformly random for all $h\in [\mu]$. Then, by the correctness of functionality $\FFF^{\mathsf{Eq},\mu}_{\mathsf{ConSel}}$, $v^\RRR_{j,k,z^*}$ and $v^\SSS_{j,k,z^*}$ are uniformly random. After aggregation, $r^\RRR_{j,z^*}$ and $r^\SSS_{j,z^*}$ are also uniformly random. 

(2) The second case, for any $z\in [2^d]$, it holds that $\CCC^\SSS_{j,z} \neq \CCC^\RRR_i$. By the correctness of functionality \Func[so\text{-}OPPRF], the outputs $e^\RRR_{j,k,z}$ and $e^\SSS_{j,k,z}$ are uniformly random. Then, $r^\RRR_{j,z}$ and $r^\SSS_{j,z}$ are also uniformly random.

In both cases, $r^\RRR_{j,z}$ and $r^\SSS_{j,z}$ are uniformly random for all $j\in [m],z\in [2^d]$. By the correctness of functionality \Func[Eq], the probability of $b_{j,z}=1$ is at most $1/\abs{\FF}$. A union bound shows that the probability of there existing $j\in [m]$ such that $b_{j}=1$ is at most $m2^d/\abs{\FF}$. By setting $\log\abs{\FF} = \ell = \lambda+d+\log m$, the probability $m2^d/\abs{\FF} = 2^{-\lambda}$ is negligible. By the correctness of \Func[OT], \RRR learns nothing about $\vecq_j$.
This completes the correctness proof.

\textbf{Corrupted sender.}
The simulator $\Sim[\SSS]$ receives the sender's input $Q$ and proceeds as follows:
\begin{enumerate}
    \item $\Sim[\SSS]$ emulates $\Func[so\text{-}OPPRF]$ by sampling uniformly random shares $e_{j,k,z,h}^{\SSS}$ for all $j \in [m]$, $k \in [d]$, $z \in [2^d]$, and $h \in [\mu]$ and sends them to $\AAA$.

    \item $\Sim[\SSS]$ emulates $\FFF^{\mathsf{Eq},\mu}_{\mathsf{ConSel}}$ by sampling uniformly random shares $v_{j,k,z}^{\SSS}$ for all $j \in [m]$, $k \in [d]$, and $z \in [2^d]$, and sends them to $\AAA$.

    \item $\Sim[\SSS]$ emulates $\Func[Eq]$ and $\Func[OT]$, both of which output nothing to $\AAA$.
\end{enumerate}

\textit{Indistinguishability.}
The adversary's view in the real execution is identically distributed to its view in the simulated execution. This follows because the values $e_{j,k,z,h}^{\SSS}$ received from $\Func[so\text{-}OPPRF]$ and $v_{j,k,z}^{\SSS}$ from $\FFF^{\mathsf{Eq},\mu}_{\mathsf{ConSel}}$ are uniformly random in the real world and are sampled from the same distribution by the simulator in the ideal world.

\textbf{Corrupted receiver.}
The simulator $\Sim[\RRR]$ receives the receiver's input $W$ and output $I$, and proceeds as follows:
\begin{enumerate}
    \item $\Sim[\RRR]$ emulates $\Func[so\text{-}OPPRF]$ by sampling uniformly random shares $e_{j,k,z,h}^{\RRR}$ for all $j \in [m]$, $k \in [d]$, $z \in [2^d]$, $h \in [\mu]$, and by sampling a random function $F$ such that $F(\CCC^{\RRR}_{i}\Vert k \Vert \tilde{w}_{i,k,h})=0$ for all $i\in[n]$, $k\in[d]$, and $h\in[\mu]$. It then sends these shares $e_{j,k,z,h}^{\RRR}$ and $\OOO^F$ to $\AAA$.

    \item $\Sim[\RRR]$ emulates $\FFF^{\mathsf{Eq},\mu}_{\mathsf{ConSel}}$ by sampling uniformly random shares $v_{j,k,z}^{\RRR}$ for all $j \in [m]$, $k \in [d]$, and $z \in [2^d]$, and sends them to $\AAA$.

    \item $\Sim[\RRR]$ emulates $\Func[Eq]$ by padding $I$ with $\bot$ to size $m$ and randomly permuting the padded $I$. For each $j \in [m]$ and $z \in [2^d]$, it sets
    \[
    b_{j,z} =
    \begin{cases}
        1, & \text{if } \vecq_j = I[j] \neq \bot \text{ and there exists } \vecw_i \text{ such that } \\
        & \CCC_i^{\RRR}=\CCC_{j,z}^{\SSS} 
            \text{ and } \mathsf{dist}(\vecq_j,\vecw_i)\le \delta,\\
        0, & \text{otherwise.}
    \end{cases}
    \]
    $\Sim[\RRR]$ sends all values $b_{j,z}$ to $\AAA$.

    \item $\Sim[\RRR]$ emulates $\Func[OT]$ and sends the padded and permuted $I$ to $\AAA$.
\end{enumerate}

\textit{Indistinguishability.} We prove security via a sequence of hybrids.

\begin{description}
    \item[Hybrid $\mathsf{H}_0$.]
    This is the real execution of the protocol.

    \item[Hybrid $\mathsf{H}_1$.] 
    This hybrid is identical to $\mathsf{H}_0$, except that \Sim[\RRR] acts as \Func[so\text{-}OPPRF], $\FFF^{\mathsf{Eq},\mu}_{\mathsf{ConSel}}$, \Func[Eq], and \Func[OT], but follows the specification of these functionalities and uses the real sender's input $Q$. This is only a conceptual change, and hence the adversary's view in Hybrid $\mathsf{H}_0$ and Hybrid $\mathsf{H}_1$ are identically distributed.

    \item[Hybrid $\mathsf{H}_2$.]
    This hybrid is identical to $\mathsf{H}_1$, except that instead of computing $b_{j, z}$'s (resp. $\vecu_j$'s) according to \Func[Eq] (resp. \Func[OT]), \Sim[\RRR] derives them utilizing the receiver's input $W$ and output $I$. Specifically, it pads $I$ with $\bot$ to size $m$, applies a random permutation,
    and for every $j\in[m],z\in[2^d]$ sets
    \[
    b_{j,z}=
    \begin{cases}
        1, &\text{if } \vecq_j = I[j] \neq \bot \text{ and there exists }\vecw_i
        \text{ such that } \\ &\CCC_i^{\RRR}=\CCC_{j,z}^{\SSS}
        \text{ and }\mathsf{dist}(\vecq_j,\vecw_i)\le\delta,\\
        0,&\text{otherwise.}
    \end{cases}
    \]
    The values $b_{j, z}$'s (resp.  $\vecu_j$'s) are returned to $\AAA$ as the outputs of $\Func[Eq]$ (resp. $\Func[OT]$).

    The adversary's views in Hybrid $\mathsf{H}_1$ and Hybrid $\mathsf{H}_2$ are identically distributed, because
    functionality Eq outputs exactly these bits $b_{j,z}$ and the outputs of functionality OT are exactly determined by $b_{j,z}$ and $I$.

    \item[Hybrid $\mathsf{H}_3$.]
    This hybrid is identical to $\mathsf{H}_2$, except that we replace the sender's interaction in $\Func[so\text{-}OPPRF]$ and $\FFF^{\mathsf{Eq},\mu}_{\mathsf{ConSel}}$ with a simulation as follows.
    $\Sim[\RRR]$ emulates $\Func[so\text{-}OPPRF]$ by sampling uniformly random shares $e_{j,k,z,h}^{\RRR}$ for all $j \in [m]$, $k \in [d]$, $z \in [2^d]$, $h \in [\mu]$, and by sampling a random function $F$ such that $F(\CCC^{\RRR}_{i}\Vert k \Vert \tilde{w}_{i,k,h})=0$ for all $i\in[n]$, $k\in[d]$, and $h\in[\mu]$. It then sends these shares $e_{j,k,z,h}^{\RRR}$ and $\OOO^F$ to $\AAA$. Then, $\Sim[\RRR]$ emulates $\FFF^{\mathsf{Eq},\mu}_{\mathsf{ConSel}}$ by sampling uniformly random shares $v_{j,k,z}^{\RRR}$ for all $j \in [m]$, $k \in [d]$, and $z \in [2^d]$, and sends them to $\AAA$.

    By the security of $\Func[so\text{-}OPPRF]$ and $\FFF^{\mathsf{Eq},\mu}_{\mathsf{ConSel}}$, the adversarial views of Hybrid $\mathsf{H}_2$ and Hybrid $\mathsf{H}_3$ are identically distributed.
\end{description}

Observe that Hybrid $\mathsf{H}_3$ is exactly the simulated execution produced by $\Sim[\RRR]$.
Therefore, the real and ideal executions are perfectly indistinguishable, completing the proof.

\end{proof}

\subsection{$L_\infty$ Distance under Receiver-sided Unique Block Assumptions}
\label{appendix: proof receiver 4delta Linf}

\begin{proof}
\textbf{Correctness.} We first consider that for some $\vecq_j \in Q$, there exists $\vecw_i \in W$ such that $\mathsf{dist}_{\infty}(\vecq_j,\vecw_i)\le \delta$. By the correctness of spatial hashing, there exists an index $z^*\in[2^d]$ such that $\CCC^\SSS_{j}=\CCC^\RRR_{i,z^*}$. Since $\mathsf{dist}_{\infty}(\vecq_j,\vecw_i)\le \delta$, it holds that $\abs{q_{j,k}-w_{i,k}}\le \delta$ for all $k\in[d]$. Therefore, for each $k\in[d]$, there exists $t_k\in[-\delta,\delta]$ such that $q_{j,k} + t_k=w_{i,k}$. By the correctness of functionality \Func[so\text{-}OPPRF],  $e^\RRR_{j,k}$ and $e^\SSS_{j,k}$ are the shares of $0^\ell$. And then the aggregated shares $r^\RRR_{j}$ and $r^\SSS_{j}$ are still the shares of $0^\ell$. By the correctness of \Func[Eq], it holds that $b_{j} =1$. By the correctness of \Func[OT], \RRR correctly obtains $\vecq_j$. 

We then consider that for some $\vecq_j\in Q$, it holds $\mathsf{dist}_{\infty}(\vecq_j,\vecw_i)> \delta$ for any $\vecw_i \in W$. We consider the following two cases. 

(1) The first case: there exists an index $z^*$ such that $\CCC^\SSS_{j} = \CCC^\RRR_{i,z^*}$. Since $\mathsf{dist}_{\infty}(\vecq_j,\vecw_i)> \delta$, it holds that there exists an index $k^*\in [d]$ such that $\abs{q_{j,k^*}-w_{i,k^*}}>\delta$. By the correctness of functionality \Func[so\text{-}OPPRF], the outputs $e^\RRR_{j,k^*}$ and $e^\SSS_{j,k^*}$ are uniformly random. Then, after aggregation, $r^\RRR_{j}$ and $r^\SSS_{j}$ are also uniformly random. 

(2) The second case: for any $z\in [2^d]$, it holds that $\CCC^\SSS_{j} \neq \CCC^\RRR_{i,z}$. By the correctness of functionality \Func[so\text{-}OPPRF], the outputs $e^\RRR_{j,k}$ and $e^\SSS_{j,k}$ are uniformly random. Then, $r^\RRR_{j}$ and $r^\SSS_{j}$ are also uniformly random.

In both cases, $r^\RRR_{j}$ and $r^\SSS_{j}$ are uniformly random for all $j\in [m]$. By the correctness of functionality \Func[Eq], the probability of $b_{j}=1$ is at most $1/\abs{\FF}$. A union bound shows that the probability of there existing $j\in [m]$ such that $b_{j}=1$ is at most $m/\abs{\FF}$. By setting $\log \abs{\FF}= \ell = \lambda+\log m$, the probability $m/\abs{\FF} = 2^{-\lambda}$ is negligible. By the correctness of \Func[OT], \RRR learns nothing about $\vecq_j$.
This completes the correctness proof.

\textbf{Corrupted sender.}
The simulator $\Sim[\SSS]$ receives the sender's input $Q$ and proceeds as follows:
\begin{enumerate}
    \item $\Sim[\SSS]$ emulates $\Func[so\text{-}OPPRF]$ by sampling uniformly random shares $e_{j,k}^{\SSS}$ for all $j \in [m]$ and $k \in [d]$, and sends them to $\AAA$.

    \item $\Sim[\SSS]$ emulates $\Func[Eq]$ and $\Func[OT]$, both of which output nothing to $\AAA$.
\end{enumerate}

\textit{Indistinguishability.}
The adversary's view in the real execution is identically distributed to its view in the simulated execution. This follows because the values $e_{j,k}^{\SSS}$ received from $\Func[so\text{-}OPPRF]$ are uniformly random in the real world and are sampled from the same distribution by the simulator in the ideal world.

\textbf{Corrupted receiver.}
The simulator $\Sim[\RRR]$ receives the receiver's input $W$ and output $I$, and proceeds as follows:
\begin{enumerate}
    \item $\Sim[\RRR]$ emulates $\Func[so\text{-}OPPRF]$ by sampling uniformly random shares $e_{j,k}^{\RRR}$ for all $j \in [m]$ and $k \in [d]$, and by sampling a random function $F$ such that $F(\CCC^{\RRR}_{i, z}\Vert k \Vert w_{i,k}+t)=0$ for all $i\in[n]$, $k\in[d]$, $z \in [2^d]$, and $t\in[-\delta,\delta]$. It then sends these shares $e_{j,k}^{\RRR}$ and $\OOO^F$ to $\AAA$.

    \item $\Sim[\RRR]$ emulates $\Func[Eq]$ by padding $I$ with $\bot$ to size $m$ and randomly permuting the padded $I$. For each $j \in [m]$, it sets
    \[
    b_{j} =
    \begin{cases}
        1, & \text{if } \vecq_j = I[j] \neq \bot \text{ and there exists } \vecw_i \text{ and } z \in [2^d] \\ & \text{ such that } \CCC_{i,z}^{\RRR}=\CCC_{j}^{\SSS} 
            \text{ and } \mathsf{dist}(\vecq_j,\vecw_i)\le \delta,\\
        0, & \text{otherwise.}
    \end{cases}
    \]
    $\Sim[\RRR]$ sends all values $b_{j}$ to $\AAA$.

    \item $\Sim[\RRR]$ emulates $\Func[OT]$ and sends the padded and permuted $I$ to $\AAA$.
\end{enumerate}

\textit{Indistinguishability.} We prove security via a sequence of hybrids.

\begin{description}
    \item[Hybrid $\mathsf{H}_0$.]
    This is the real execution of the protocol.

    \item[Hybrid $\mathsf{H}_1$.] 
    This hybrid is identical to $\mathsf{H}_0$, except that \Sim[\RRR] acts as \Func[so\text{-}OPPRF], \Func[Eq], and \Func[OT], but follows the specification of these functionalities and uses the real sender's input $Q$. This is only a conceptual change, and hence the adversary's view in Hybrid $\mathsf{H}_0$ and Hybrid $\mathsf{H}_1$ are identically distributed.

    \item[Hybrid $\mathsf{H}_2$.]
    This hybrid is identical to $\mathsf{H}_1$, except that instead of computing $b_{j}$'s (resp. $\vecu_j$'s) according to \Func[Eq] (resp. \Func[OT]), \Sim[\RRR] derives them utilizing the receiver's input $W$ and output $I$. Specifically, it pads $I$ with $\bot$ to size $m$, applies a random permutation,
    and for every $j\in[m]$ sets
    \[
    b_{j}=
    \begin{cases}
        1, &\text{if } \vecq_j = I[j] \neq \bot \text{ and there exists }\vecw_i \text{ and } z\in[2^d] \\
        & \text{ such that }\CCC_{i,z}^{\RRR}=\CCC_{j}^{\SSS} \text{ and }\mathsf{dist}(\vecq_j,\vecw_i)\le\delta,\\
        0,&\text{otherwise.}
    \end{cases}
    \]
    The values $b_{j}$'s (resp.  $\vecu_j$'s) are returned to $\AAA$ as the outputs of $\Func[Eq]$ (resp. $\Func[OT]$).

    The adversary's views in Hybrid $\mathsf{H}_1$ and Hybrid $\mathsf{H}_2$ are identically distributed, because
    functionality Eq outputs exactly these bits $b_{j}$ and the outputs $\vecu_{j}$ of functionality OT are exactly determined by $b_{j}$ and $I$.

    \item[Hybrid $\mathsf{H}_3$.]
    This hybrid is identical to $\mathsf{H}_2$, except that we replace the sender's interaction in $\Func[so\text{-}OPPRF]$ with a simulation as follows: sample random shares
    $e_{j,k}^{\RRR}$ for all $j\in[m],k\in[d]$,
    and sample a random function $F$ conditioned on $F(\CCC_{i,z}^{\RRR}\Vert k\Vert w_{i,k}+t)=0$ for all $i\in[n]$, $k\in[d]$, $z\in[2^d]$, and $t\in[-\delta,\delta]$.
    The simulator sends $\{e_{j,k}^{\RRR}\}$ and $\OOO^F$ to $\AAA$.

    By the security of $\Func[so\text{-}OPPRF]$, the adversary’s view of Hybrid $\mathsf{H}_2$ and Hybrid $\mathsf{H}_3$ are identically distributed.
\end{description}

Observe that Hybrid $\mathsf{H}_3$ is exactly the simulated execution produced by $\Sim[\RRR]$.
Therefore, the real and ideal executions are perfectly indistinguishable, completing the proof.

\end{proof}

\subsection{Security Proof of $\mathcal{F}_{\mathsf{ConRand}}$}
\label{appendix: proof con rand}

\begin{proof}
\textbf{Correctness.} For any $e^\SSS$ and $e^\RRR$ such that $f(e^\SSS, e^\RRR) = 1$, the outputs $b^\SSS$ and $b^\RRR$ are secret shares of $1$. By the correctness of \Func[MUX], the MUX outputs $m^\SSS$ and $m^\RRR$ are secret shares of $0$, and hence $z^\SSS + z^\RRR = v^\SSS + v^\RRR$.

Conversely, for any $e^\SSS$ and $e^\RRR$ such that $f(e^\SSS, e^\RRR) = 0$, the outputs $b^\SSS$ and $b^\RRR$ are secret shares of $0$. Since $r^\SSS$ and $r^\RRR$ are sampled uniformly at random, by the correctness of \Func[MUX], $m^\SSS$ and $m^\RRR$ are uniformly random, and consequently $v^\SSS + m^\SSS$ and $v^\RRR + m^\RRR$ are each uniformly random.
This completes the correctness proof.

\textbf{Corrupted sender.}
The simulator $\Sim[\SSS]$ receives the sender's input $(e^\SSS, v^\SSS)$ and output $z^\SSS$, and proceeds as follows:
\begin{enumerate}
    \item $\Sim[\SSS]$ emulates $\Func[ssPEQT]$/$\Func[ssCMP]$ by sampling uniformly random shares $b^{\SSS}$ and sends it to $\AAA$.

    \item $\Sim[\SSS]$ emulates $\Func[MUX]$ by computing $m^{\SSS} = z^{\SSS} - v^{\SSS}$ and sends it to $\AAA$.
\end{enumerate}

\textit{Indistinguishability.}
The adversary's view in the real execution is identically distributed to its view in the simulated execution. This follows because the values $b^{\SSS}$ received from $\Func[ssPEQT]$/$\Func[ssCMP]$ are uniformly random in the real execution and are sampled from the same distribution by the simulator in the ideal execution. Moreover, the value $m^{\SSS}$ is identically distributed in both worlds conditioned on the relation
$z^{\SSS} = v^{\SSS} + m^{\SSS}$.
Therefore, the joint distribution of the adversary's view is identical in the real and ideal executions.

\textbf{Corrupted receiver.} Since the protocol execution is symmetric, the security of the corrupted receiver follows similarly. This completes the proof.

\end{proof}

\subsection{Security Proof of $\mathcal{F}_{\mathsf{ConSel}}$}
\label{appendix: proof con sel}

\begin{proof}

\textbf{Corrupted sender.}
The simulator $\Sim[\SSS]$ receives the sender's input $\{e_i^\SSS\}_{i \in [\mu]}$ and output $m^\SSS$, and proceeds as follows:
\begin{enumerate}
    \item $\Sim[\SSS]$ emulates $\Func[ssPEQT]$/$\Func[ssCMP]$ by sampling uniformly random shares $b_i^{\SSS}$ for $i \in [\mu]$ and sends them to $\AAA$.

    \item $\Sim[\SSS]$ emulates $\Func[MUX]$ and sends $m^\SSS$ to $\AAA$.
\end{enumerate}

\textit{Indistinguishability.}
The adversary's view in the real execution is identically distributed to its view in the simulated execution. This follows because the values $b_i^{\SSS}$ received from $\Func[ssPEQT]$/$\Func[ssCMP]$ are uniformly random in the real execution and are sampled from the same distribution by the simulator in the ideal execution. Moreover, the output $m^{\SSS}$ is identically distributed in both worlds.
Therefore, the joint distribution of the adversary's view is identical in the real and ideal executions.

\textbf{Corrupted receiver.} Since the protocol execution is symmetric, the security of the corrupted receiver follows similarly. This completes the proof.

\end{proof}

\subsection{Proof of Lemma~\ref{lemma: prefix}}
\label{appendix: proof of prefix}

\begin{proof}
By definition 3.1 of~\cite{dang2025ccs}, let $\{ b_j \}_{j\in[\mu]} = \mathsf{PxTrie}(q-\delta,q+\delta)$. Then:
\[
[q-\delta,q+\delta] = \bigcup_{j\in[\mu]} \mathsf{Interval}(b_j)
\]
and for any $b\in \{ b_j \}_{j\in[\mu]}$, if a binary string $b'$ is a prefix of $b$, then:
$$\mathsf{Interval}(b') \nsubseteq [q-\delta,q+\delta]$$

Therefore, for any $b_j$, we have $\mathsf{Interval}(b_j) \subseteq [q-\delta, q+\delta]$. Let $k_j$ be the number of wildcards in $b_j$. Since the length of the interval $|\mathsf{Interval}(b_j)| = 2^{k_j}$ cannot exceed the total length of the range $2\delta+1$, it holds that $2^{k_j} \le 2\delta+1$, which implies $k_j \le \lfloor \log_2(2\delta+1) \rfloor$.

{\textbf{Case 1}: $w \in [q-\delta, q+\delta]$.}
First, we show \textbf{existence}. Since $w \in [q-\delta, q+\delta]$, the completeness property of definition 3.1 of~\cite{dang2025ccs} ensures there exists an index $j^* \in [\mu]$ such that $w \in \mathsf{Interval}(b_{j^*})$. By the definition of $\mathsf{Interval}(\cdot)$, $b_{j^*}$ must be a prefix of $w$. 
As derived above, the wildcard count of $b_{j^*}$ satisfies $k_{j^*} \le \lfloor \log_2(2\delta+1) \rfloor$. By the definition of $\mathsf{PxPath}(w, \delta)$ provided in section~\ref{Sec: prefix-trie prelim}, $b_{j^*}$ is a valid prefix on the path of root-to-$w$ within this wildcard bound. Thus, $b_{j^*} \in \mathsf{PxPath}(w, \delta) \cap \mathsf{PxTrie}(q-\delta, q+\delta)$, and the intersection is non-empty.

Next, we show \textbf{uniqueness} by contradiction. Suppose there exist two distinct nodes $x, y$ in the intersection. Since $x, y \in \mathsf{PxPath}(w, \delta)$, both are prefixes of $w$, implying one must be a prefix of the other. However, definition 3.1 of~\cite{dang2025ccs} states that for any $b \in \mathsf{PxTrie}$, its proper prefixes $b'$ must satisfy $\mathsf{Interval}(b') \not\subseteq [q-\delta, q+\delta]$. This prefix-free property ensures that no node in $\mathsf{PxTrie}$ can be a prefix of another, contradicting the existence of both $x$ and $y$. Thus, the intersection size is exactly 1.

{\textbf{Case 2}: $w \notin [q-\delta, q+\delta]$.} We show the intersection is empty by contradiction. Suppose there exists some $b_i \in \mathsf{PxPath}(w,\delta) \cap \mathsf{PxTrie}(q-\delta,q+\delta)$. 
The membership $b_i \in \mathsf{PxPath}(w,\delta)$ implies $b_i$ is a prefix of $w$, so $w \in \mathsf{Interval}(b_i)$. 
The membership $b_i \in \mathsf{PxTrie}(q-\delta,q+\delta)$ implies $\mathsf{Interval}(b_i) \subseteq [q-\delta, q+\delta] $ per definition 3.1 of~\cite{dang2025ccs}. 
It follows that $w \in [q-\delta, q+\delta]$, which contradicts the assumption. Hence, the intersection must be empty.

Combining the above two cases completes the proof.

\end{proof}

\begin{figure}[!h]
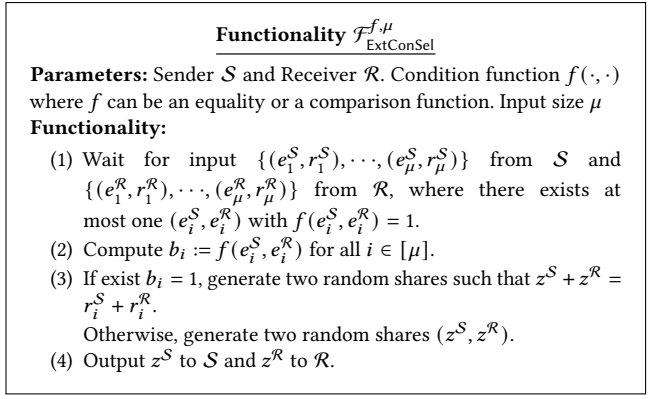

\begin{nffunc}{\ensuremath{\FFF^{f,\mu}_{\mathsf{ExtConSel}}}}

\noindent \textbf{Parameters:} Sender \SSS and Receiver \RRR. Condition function $f(\cdot,\cdot)$ where $f$ can be an equality or a comparison function. Input size $\mu$

\noindent \textbf{Functionality:}

\begin{enumerate}

    \item Wait for input $\{(e^\SSS_1,r^\SSS_1), \cdots\!,(e^\SSS_\mu,r^\SSS_\mu)\}$ from \SSS and $\{(e^\RRR_1,r^\RRR_1), \cdots\!,(e^\RRR_\mu,r^\RRR_\mu)\}$ from \RRR, where there exists at most one $(e^\SSS_i,e^\RRR_i)$ with $f(e^\SSS_i,e^\RRR_i) = 1$.

    \item Compute $b_i:= f(e^\SSS_i,e^\RRR_i)$ for all $i\in[\mu]$.

    \item If exist $b_i=1$, generate two random shares such that $z^\SSS + z^\RRR = r_i^\SSS + r_i^\RRR$. 
    
    Otherwise, generate two random shares $(z^\SSS,z^\RRR)$.
    \item Output $z^\SSS$ to \SSS and $z^\RRR$ to \RRR.
    
\end{enumerate}

\end{nffunc}
\caption{Functionality of Extended Conditional Selection.}
\label{Func: eq sel ext}
\end{figure}

\begin{figure}[!h]
    \begin{nfprot}{$\Pi^{\mathsf{Eq},\mu}_{\mathsf{ExtConSel}}$}
    \noindent \textbf{Parameters:} Ideal functionality \Func[MUX] and \Func[ssPEQT].
    
    \noindent \textbf{Input:} \SSS inputs $\{(e_i^\SSS,r_i^\SSS)\}_{i \in [\mu]}$ and \RRR inputs $\{(e_i^\RRR,r_i^\RRR)\}_{i \in [\mu]}$.
    
    \noindent \textbf{Protocol:}
    \begin{enumerate}
        \item \SSS and \RRR invoke \Func[ssPEQT], where \SSS inputs $\{e_i^\SSS\}_{i \in [\mu]}$ and \RRR inputs $\{e_i^\RRR\}_{i \in [\mu]}$. \SSS receives $\{b_i^\SSS\}_{i \in [\mu]} \in \bool^\mu$, and \RRR receives $\{b_i^\RRR\}_{i \in [\mu]} \in \bool^\mu$. 

        \item \SSS and \RRR invoke \Func[MUX], where \SSS inputs $\{(b_i^\SSS, r_i^\SSS)\}_{i \in [\mu]}$ and \RRR inputs $\{(b_i^\RRR, r_i^\RRR)\}_{i \in [\mu]}$. \SSS receives $\{t_i^\SSS\}_{i \in [\mu]}$ and \RRR receives $\{t_i^\RRR\}_{i \in [\mu]}$.

        \item \SSS computes $t^\SSS := \sum_{i \in [\mu]} t_i^\SSS$ and \RRR computes $t^\RRR := \sum_{i \in [\mu]} t_i^\RRR$

        \item \SSS computes $b^\SSS := \Xor_{i \in [\mu]} b_i^\SSS$ and \RRR computes $b^\RRR := \Xor_{i \in [\mu]} b_i^\RRR$.
        
        \item \SSS randomly samples $r^\SSS$ and \RRR randomly samples $r^\RRR$.
        
        \item \SSS and \RRR invoke \Func[MUX], where \SSS inputs $(1-b^\SSS, r^\SSS)$ and \RRR inputs $(b^\RRR, r^\RRR)$. \SSS receives $m^\SSS$ and \RRR receives $m^\RRR$.

        \item \SSS outputs $z^\SSS = m^\SSS+t^\SSS$ and \RRR outputs $z^\RRR=m^\RRR+t^\RRR$.

    \end{enumerate}
    
    \end{nfprot}
    \vspace{0.5em}
    \caption{Protocol of Extended Equality-conditional Selection.}
    \label{Prot: eq sel ext}
\end{figure}

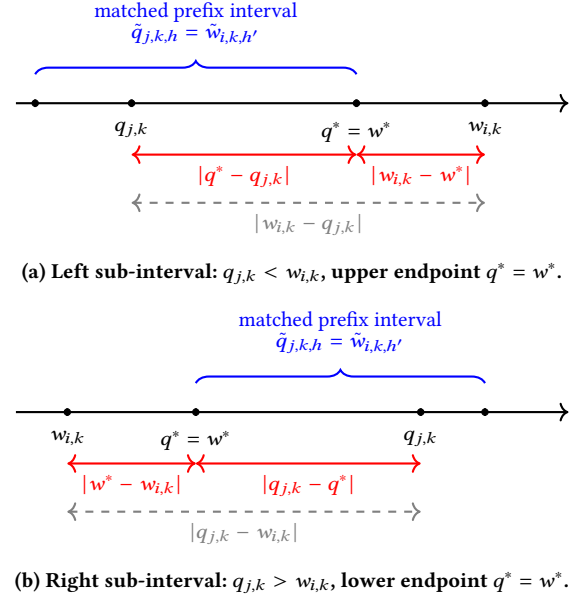
\begin{figure}[t]
\centering

\begin{subfigure}[t]{\columnwidth}
\centering
\begin{tikzpicture}[x=0.85cm, y=0.85cm,
                    every node/.style={font=\small}]
  \draw[->, thick] (-0.3, 0) -- (8.3, 0);

  \foreach \x in {0, 1.5, 5, 7}{
    \draw[thick] (\x, -0.06) -- (\x, 0.06);
  }
  \foreach \x in {1.5, 5, 7, 0}{
    \filldraw (\x, 0) circle (1.2pt);
  }
  \node[below=3pt] at (1.5, -0.03) {$q_{j,k}$};
  \node[below=3pt] at (5,   -0.03) {$q^{*}=w^{*}$};
  \node[below=3pt] at (7,   -0.03) {$w_{i,k}$};

  \draw[blue, thick, decorate, decoration={brace, amplitude=5pt}]
        (0, 0.45) -- (5, 0.45);
  \node[blue, above=8pt, align=center] at (2.5, 0.45)
        {matched prefix interval\\[-2pt]
         $\tilde{q}_{j,k,h}=\tilde{w}_{i,k,h'}$};

  \draw[<->, red, thick] (1.5, -0.8) -- (5, -0.8);
  \node[red, below=1pt] at (3.25, -0.8) {$|q^{*}-q_{j,k}|$};

  \draw[<->, red, thick] (5, -0.8) -- (7, -0.8);
  \node[red, below=1pt] at (6, -0.8) {$|w_{i,k}-w^{*}|$};

  \draw[<->, dashed, gray, thick] (1.5, -1.55) -- (7, -1.55);
  \node[gray, below=1pt] at (4.25, -1.55) {$|w_{i,k}-q_{j,k}|$};
\end{tikzpicture}
\caption{Left sub-interval: $q_{j,k}<w_{i,k}$, upper endpoint $q^{*}=w^{*}$.}
\label{fig:prefix-decomp-left}
\end{subfigure}

\vspace{0.6em}

\begin{subfigure}[t]{\columnwidth}
\centering
\begin{tikzpicture}[x=0.85cm, y=0.85cm,
                    every node/.style={font=\small}]
  \draw[->, thick] (-0.3, 0) -- (8.3, 0);
  \foreach \x in {7, 0.5, 2.5, 6}{
    \draw[thick] (\x, -0.06) -- (\x, 0.06);
  }
  \foreach \x in {0.5, 2.5, 6, 7}{
    \filldraw (\x, 0) circle (1.2pt);
  }
  \node[below=3pt] at (0.5, -0.03) {$w_{i,k}$};
  \node[below=3pt] at (2.5, -0.03) {$q^{*}=w^{*}$};
  \node[below=3pt] at (6,   -0.03) {$q_{j,k}$};
  \draw[blue, thick, decorate, decoration={brace, amplitude=5pt}]
        (2.5, 0.45) -- (7, 0.45);
  \node[blue, above=8pt, align=center] at (4.75, 0.45)
        {matched prefix interval\\[-2pt]
         $\tilde{q}_{j,k,h}=\tilde{w}_{i,k,h'}$};
  \draw[<->, red, thick] (0.5, -0.8) -- (2.5, -0.8);
  \node[red, below=1pt] at (1.5, -0.8) {$|w^{*}-w_{i,k}|$};
  \draw[<->, red, thick] (2.5, -0.8) -- (6, -0.8);
  \node[red, below=1pt] at (4.25, -0.8) {$|q_{j,k}-q^{*}|$};
  \draw[<->, dashed, gray, thick] (0.5, -1.55) -- (6, -1.55);
  \node[gray, below=1pt] at (3.25, -1.55) {$|q_{j,k}-w_{i,k}|$};
\end{tikzpicture}
\caption{Right sub-interval: $q_{j,k}>w_{i,k}$, lower endpoint $q^{*}=w^{*}$.}
\label{fig:prefix-decomp-right}
\end{subfigure}

\caption{An example to compute the distances with prefix trie techniques}
\label{Fig: prefix distance}
\end{figure}

\begin{figure}[h]
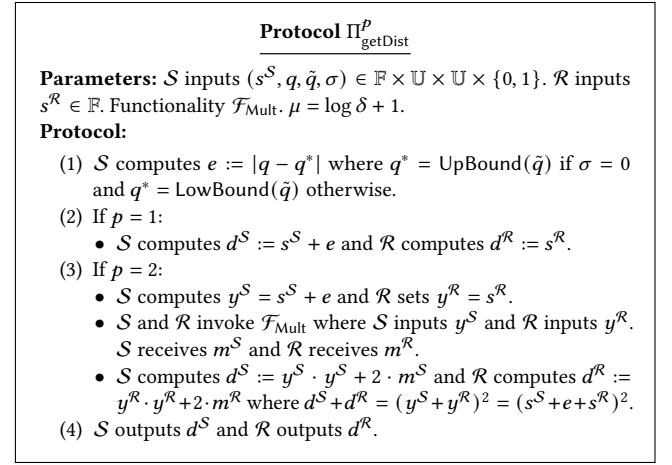

    \begin{nfprot}{$\Pi^p_{\mathsf{getDist}}$} 

    \noindent \textbf{Parameters:} \SSS inputs $(s^\SSS,q,\tilde{q},\sigma)\in \FF\times \UU \times \UU \times \{0,1\}$. \RRR inputs $s^\RRR \in \FF$. Functionality $\Func[Mult]$. $\mu = \log\delta+1$.
    
    \noindent \textbf{Protocol:}
    \begin{enumerate}
        \item \SSS computes $e := \abs{q-q^*}$ where $q^* = \mathsf{UpBound}(\tilde{q})$ if $\sigma=0$ and $q^* = \mathsf{LowBound}(\tilde{q})$ otherwise.
        \item If $p=1$: \begin{itemize}
                            \item \SSS computes $d^\SSS := s^\SSS + e$ and \RRR computes $d^\RRR := s^\RRR$. 
                        \end{itemize}
        \item If $p=2$: \begin{itemize}
                    \item \SSS computes $y^\SSS = s^\SSS+e$ and \RRR sets $y^\RRR = s^\RRR$.
                    \item \SSS and \RRR invoke $\Func[Mult]$ where \SSS inputs $y^\SSS$ and \RRR inputs $y^\RRR$. \SSS receives $m^\SSS$ and \RRR receives $m^\RRR$.
                    \item \SSS computes $d^\SSS :=  y^\SSS\cdot y^\SSS + 2\cdot m^\SSS$ and \RRR computes $d^\RRR := y^\RRR\cdot y^\RRR + 2\cdot m^\RRR$ where $d^\SSS + d^\RRR = (y^\SSS + y^\RRR)^2 = (s^\SSS +e+ s^\RRR)^2 $.
                    \end{itemize}
        \item \SSS outputs $d^\SSS$ and \RRR outputs $d^\RRR$.
        
    \end{enumerate}
    
    \end{nfprot}
    \vspace{0.5em}
    \caption{Protocol for computing distance for $p=1$ and $p=2$. The protocol can be extended to arbitrary $p$.}
    \label{Prot: sum distance}
\end{figure}

\section{Additional Protocols under Unique Cell Assumptions}

\subsection{Sender-sided Setting}
\label{appendix: other sender 2delta Lp}
The protocol for $L_p$ distances under sender-sided unique cell assumption is given in Figure~\ref{Prot: sender 2delta Lp}. The proof is given in Appendix~\ref{appendix: proof sender 2delta Lp}.

\begin{figure*}[t]
    \begin{nfprot}{\ensuremath{\Pi_\mathsf{FPSI}^{L_p}}\xspace}
        \noindent \textbf{Parameters:}  Sender \SSS and receiver \RRR with input sizes $m$ and $n$. Threshold \( \delta \). Bit length $\lambda^\prime = \lambda + d+\log mn$, $\ell = \lambda + d+\log n+p\log\delta$, and $u = \log \abs{\UU}$.

        \noindent \textbf{Input:} \SSS inputs \( Q=\{\vecq_j\}_{j\in [m]} \in \mathbb{U}^{d\times m} \) and \RRR inputs \( W=\{\vecw_i\}_{i\in [n]} \in \mathbb{U}^{d\times n} \).

        \noindent \textbf{Protocol:}
        \begin{enumerate}
            \item For $j\in[m]$, \SSS computes $\CCC_j^\SSS = \mathsf{cell}_{2\delta}(\vecq_j)$. For $i\in [n]$, \RRR computes $\{ \CCC^{\RRR}_{i,z}\}_{z\in [2^d]} = \mathsf{neigh}_{2\delta}(\vecw_i)$.
            
            \item \SSS sets $\mathsf{List} = \left\{ \left( \CCC^{\SSS}_{j} \Vert k \Vert q_{j,k} + t, \abs{t}^p \big\Vert s_{j,k}  \right) \right\}_{j\in[m],k\in[d], t\in[\text{-}\delta,\delta]}$ and computes $s_{j} = {\sum}_{k\in[d]}s_{j,k}$ where $s_{j,k}\overset{\$}{\leftarrow} \FF_{2^\kappa}$.
            \item \RRR and \SSS invoke functionality \Func[so\text{-}OPPRF] where \SSS inputs $\mathsf{List}$ and \RRR inputs $\{\CCC^{\RRR}_{i,z} \Vert k \Vert w_{i,k}\}_{i\in[n],k\in[d],z\in[2^d]}$. \RRR receives $\{e^\RRR_{i,k,z}\Vert r^\RRR_{i,k,z}\}_{i\in [n],k\in [d],z\in[2^d]}$ and \SSS receives $\{e^\SSS_{i,k,z}\Vert r^\SSS_{i,k,z}\}_{i\in [n],k\in [d],z\in[2^d]}$ and $\OOO^F$.

            \item For $i\in [n],k\in [d],z\in[2^d]$, \SSS and \RRR invoke {\Func[B2A], where \SSS inputs $e^{\SSS}_{i,k,z}$ and \RRR inputs $e^{\RRR}_{i,k,z}$. \SSS receives $d^{\SSS}_{i,k,z}$ and \RRR receives $d^{\RRR}_{i,k,z}$}.

            \item For $i\in [n],z\in[2^d]$, \RRR computes $d_{i,z}^{\RRR} = {\sum}_{k\in[d]}d^{\RRR}_{i,k,z} $ and $r_{i,z}^{\RRR} = {\sum}_{k\in[d]}r^{\RRR}_{i,k,z} $, and \SSS computes $d_{i,z}^{\SSS} = {\sum}_{k\in[d]}d^{\SSS}_{i,k,z} $ and $r_{i,z}^{\SSS} = {\sum}_{k\in[d]}r^{\SSS}_{i,k,z}$.

            \item For $i\in [n],z\in[2^d]$, \RRR and \SSS invoke \ensuremath{\FFF^{\mathsf{Cmp}}_{\mathsf{ConRand}}} where \RRR inputs $(d_{i,z}^\RRR , r_{i,z}^\RRR )$ and \SSS inputs $(d_{i,z}^\SSS , r_{i,z}^\SSS )$. \RRR receives $v_{i,z}^\RRR$ and \SSS receives $v_{i,z}^\SSS$.

            \item For $i\in[n],z\in[2^d]$, \SSS sends $v^\SSS_{i,z}$ and \RRR reconstructs $v_{i,z} = v^\SSS_{i,z} + v^\RRR_{i,z}$.

            \item \SSS computes $X = \{ H_{\lambda^\prime+du}(s_j)\oplus(0^{\lambda^\prime}\Vert \vecq_j)\}_{j\in[m]}$ and sends shuffled $X$ to \RRR.

            \item \RRR sets $I = \emptyset$ and parses $X$ as $\{ x_{j}\}_{j\in [m]}$.  For $j\in [m],i\in[n],z\in[2^d]$, if $H_{\lambda^\prime+du}(v_{i,z}) \oplus x_{j} = 0^{\lambda^\prime} \Vert \vecq$, update $I = I \cup \{\vecq\}$.
        \end{enumerate}

    \end{nfprot}
    \vspace{0.5em}
    \caption{Protocol of fuzzy PSI for $L_p$ distance under sender-sided unique cell assumption.}
    \label{Prot: sender 2delta Lp}
\end{figure*}

\subsection{Receiver-sided Setting}
\label{appendix: other receiver 2delta Lp}
The protocol for $L_p$ distances under the receiver-sided unique cell assumption is given in Figure~\ref{Prot: receiver 2delta Lp}. The proof is given in Appendix~\ref{appendix: proof receiver 2delta Lp}.

\begin{figure*}[t]
    \begin{nfprot}{\ensuremath{\Pi_\mathsf{FPSI}^{L_p}}\xspace}
        \noindent \textbf{Parameters:}  Sender \SSS and receiver \RRR with input sizes $m$ and $n$. Threshold \( \delta \). Bit length $\ell = \lambda +d+ \log m + p\log\delta$.

        \noindent \textbf{Input:} \SSS inputs \( Q=\{\vecq_j\}_{j\in [m]} \in \mathbb{U}^{d\times m} \) and \RRR inputs \( W=\{\vecw_i\}_{i\in [n]} \in \mathbb{U}^{d\times n} \).

        \noindent \textbf{Protocol:}
        \begin{enumerate}
            \item For $i\in[n]$, \RRR computes $\CCC_i^\RRR = \mathsf{cell}_{2\delta}(\vecw_i)$. \SSS shuffles $Q$ and computes $\{ \CCC^{\SSS}_{j,z}\}_{z\in [2^d]} = \mathsf{neigh}_{2\delta}(\vecq_j)$ for $j\in [m]$.
            
            \item \RRR sets $\mathsf{List} = \{(\CCC^{\RRR}_{i} \Vert k \Vert w_{i,k}+t,\abs{t}^p) \}_{i\in[n],k\in[d],t\in[\text{-}\delta,\delta]}$
            \item \RRR and \SSS invoke functionality \Func[so\text{-}OPPRF] where \RRR inputs $\mathsf{List}$ and \SSS inputs $\{\CCC^{\SSS}_{j,z} \Vert k\Vert q_{j,k}\}_{j\in [m],k\in [d],z\in [2^d]}$. \RRR receives $\{e^\RRR_{j,k,z}\}_{j\in [m],k\in [d],z\in [2^d]}$ and $\OOO^F$ and \SSS receives $\{e^\SSS_{j,k,z}\}_{j\in [m],k\in [d],z\in [2^d]}$.

            \item For $j\in [m],k\in [d], z\in[2^d]$,  \RRR and \SSS invoke functionality \Func[B2A], where \RRR inputs $e^{\RRR}_{j,k,z}$ and \SSS inputs $e^{\SSS}_{j,k,z}$. \RRR receives $d^{\RRR}_{j,k,z}$ and \SSS receives $d^{\SSS}_{j,k,z}$.

            \item For $j\in [m],z\in [2^d]$, \RRR computes $r^{\RRR}_{j,z}={\sum}_{k\in[d]}d^{\RRR}_{j,k,z}$ and \SSS computes $r^{\SSS}_{j,z}={\sum}_{k\in[d]}d^{\SSS}_{j,k,z}$.

            \item For $j\in [m],z\in[2^d]$, \RRR and \SSS invoke functionality \ensuremath{\FFF^{\delta^p}_{\mathsf{Interval}}} where \RRR inputs $r^{\RRR}_{j,z}$ and \SSS inputs $r^{\SSS}_{j,z}$. \RRR receives $b_{j,z} = \mathbf{1}\{ r^{\RRR}_{j,z} + r^{\SSS}_{j,z} \le \delta^p \}$.

            \item For $j\in [m]$, \RRR computes $b_j = {\bigvee}_{z\in[2^d]}b_{j,z}$.

            \item For $j\in [m]$, \RRR and \SSS invoke functionality \Func[OT] where \RRR inputs $b_j$ and \SSS inputs $(\bot,\vecq_j)$. \RRR receives $\vecu_j$.

            \item \RRR outputs $I = \{ \vecu_j \; \vert \; b_j=1 \;\text{for} \;  j\in[m] \}$.

        \end{enumerate}

    \end{nfprot}
    \vspace{0.5em}
    \caption{Protocol of fuzzy PSI for $L_p$ distance under receiver-sided unique cell assumption.}
    \label{Prot: receiver 2delta Lp}
\end{figure*}

\section{Additional Prefix-optimized Protocols}
\label{appendix: other protocols}

\subsection{Sender-sided Setting}
\label{appendix: other protocols for sender}

The prefix-optimized protocol under the sender-sided \textit{unique cell} assumption for $L_\infty$ distances is given in Figure~\ref{Prot: sender 2delta Linf px}. The protocol for $L_p$ distances can be optimized analogously, and we omit the details for brevity.
The core idea follows Section~\ref{Sec: overview prefix trie}. Here, we extend $\FFF^{f,\mu}_{\mathsf{ConSel}}$ to a more general functionality $\FFF^{f,\mu}_{\mathsf{ExtConSel}}$ that handles additional inputs. Specifically, $\FFF^{f,\mu}_{\mathsf{ExtConSel}}$ takes as input a set of tuples, each consisting of a condition and an associated value to be randomized, and randomizes each value depending on whether its condition is satisfied. $\FFF^{f,\mu}_{\mathsf{ConSel}}$ can be viewed as a special case where the condition itself serves as the value. The ideal functionality of $\FFF^{f,\mu}_{\mathsf{ExtConSel}}$ is given in Figure~\ref{Func: eq sel ext}, and its instantiation, which follows the same structure as $\FFF^{f,\mu}_{\mathsf{ConSel}}$, is presented in Figure~\ref{Prot: eq sel ext}.

We note that when invoking $\FFF^{f,\mu}_{\mathsf{ExtConSel}}$, there is no need to additionally invoke $\mathsf{EqRand}$, as the output shares $v^\RRR_{i,z}$ and $v^\SSS_{i,z}$ are already re-randomized for non-matching points $\vecw_i$.

\begin{figure*}[t]
    \begin{nfprot}{\ensuremath{\Pi_\mathsf{FPSI\text{-}Px}^{L_\infty}}\xspace}
        \noindent \textbf{Parameters:}  Sender \SSS and receiver \RRR with input sizes $m$ and $n$. Threshold \( \delta \). Bit length $\lambda^\prime = \lambda + d + \log mn$, $\ell = \lambda +d+ \log nd\mu$ and $u=\log\abs{\UU}$.

        \noindent \textbf{Input:} \SSS inputs \( Q=\{\vecq_j\}_{j\in [m]} \in \mathbb{U}^{d\times m} \) and \RRR inputs \( W=\{\vecw_i\}_{i\in [n]} \in \mathbb{U}^{d\times n} \).

        \noindent \textbf{Protocol:}
        \begin{enumerate}
            \item For $j\in[m]$, \SSS computes $\CCC_j^\SSS = \mathsf{cell}_{2\delta}(\vecq_j)$. For $i\in [n]$, \RRR computes $\{ \CCC^{\RRR}_{i,z}\}_{z\in [2^d]} = \mathsf{neigh}_{2\delta}(\vecw_i)$.

            \item \SSS encodes $\mathsf{List} = \{ (\CCC^{\SSS}_{j}\Vert k \Vert \tilde{q}_{j,k,h},0^\ell \Vert s_{j,k})\}_{j\in[m],k\in[d], h\in[\mu]}$ where $s_{j,k} \overset{\$}{\leftarrow} \FF_{2^\kappa}$ and $\{\tilde{q}_{j,k,h}\}_{h\in[\mu]} = \textsf{PxTrie}(q_{j,k}\!-\!\delta,q_{j,k}\!+\delta)$.
            
            \item \RRR computes $\mathsf{query} = \{\CCC^{\RRR}_{i,z} \Vert k \Vert \tilde{w}_{i,k,h}\}_{i\in[n],k\in[d],z\in[2^d],h\in [\mu]}$ where $\{\tilde{w}_{i,k,h}\}_{h\in[\mu]} = \textsf{PxPath}(w_{i,k},\delta)$
            \item \RRR and \SSS invoke functionality \Func[so\text{-}OPPRF] where \SSS inputs $\mathsf{List}$ and \RRR inputs $\mathsf{query}$. \RRR receives $\{e^\RRR_{i,k,z,h}\Vert r^\RRR_{i,k,z,h}\}_{i\in[n],k\in[d],z\in[2^d],h\in [\mu]}$ and \SSS receives $\{e^\SSS_{i,k,z,h}\Vert r^\SSS_{i,k,z,h} \}_{i\in[n],k\in[d],z\in[2^d],h\in [\mu]}$.

            \item For $i\in [n],k\in [d],z\in[2^d]$, \RRR and \SSS invoke $\FFF^{\mathsf{Eq},\mu}_{\mathsf{ExtConSel}}$ where \RRR inputs $\{(e^\RRR_{i,k,z,h}, r^\RRR_{i,k,z,h})\}_{h\in[\mu]}$ and \SSS inputs $\{(-e^\SSS_{i,k,z,h},r^\SSS_{i,k,z,h})\}_{h\in[\mu]}$. \RRR receives $v^\RRR_{i,k,z}$ and \SSS receives $v^\SSS_{i,k,z}$.

            \item For $i\in[n],z\in[2^d]$, \RRR computes $v^\RRR_{i,z}={\sum}_{k\in[d]}v^\RRR_{i,k,z}$. \SSS computes $v^\SSS_{i,z}={\sum}_{k\in[d]}v^\SSS_{i,k,z}$.

            \item For $i\in[n],z\in[2^d]$, \SSS sends $v^\SSS_{i,z}$ and \RRR reconstructs $v_{i,z} = v^\SSS_{i,z} + v^\RRR_{i,z}$.

            \item For $j\in[m]$, \SSS computes $s_{j} = {\sum}_{k\in[d]}s_{j,k}$.

            \item \SSS computes $X = \{ H_{\lambda^\prime+du}(s_j)\oplus(0^{\lambda^\prime}\Vert \vecq_j)\}_{j\in[m]}$ and sends shuffled $X$ to \RRR.

            \item \RRR sets $I = \emptyset$ and parses $X$ as $\{ x_{j}\}_{j\in [m]}$.  For $j\in [m],i\in[n],z\in[2^d]$, if $H_{\lambda^\prime+du}(v_{i,z}) \oplus x_{j} = 0^{\lambda^\prime} \Vert \vecq$, update $I = I \cup \{\vecq\}$.
            
        \end{enumerate}

    \end{nfprot}
    \vspace{0.5em}
    \caption{Protocol of prefix-optimized fuzzy PSI for $L_\infty$ distance under sender-sided unique cell assumption.}
    \label{Prot: sender 2delta Linf px}
\end{figure*}

\subsection{Receiver-sided Setting}
\label{appendix: receiver 2delta Lp px}

The prefix-optimized protocol under the receiver-sided \textit{unique cell} assumption for $L_p$ distances is given in Figure~\ref{Prot: receiver 2delta Lp px}. Since each prefix covers a range of points rather than a single value, the receiver cannot program the exact $p$-th power distance $|t|^p$ as in the $L_\infty$ construction. We therefore adopt the approach of~\cite{dang2025ccs}, where the receiver programs the distance from $w_{i,k}$ to the endpoint of the interval covered by its prefix. 
The sender then evaluates the so-OPPRF to obtain shares of this distance and computes the total distance via MPC multiplication and addition.

Specifically, let $\{\tilde{w}_{i,k,h'}\}_{h'\in[\mu]}$ and $\{\tilde{q}_{j,k,h}\}_{h\in[\mu]}$ denote the prefixes of $w_{i,k}$ and $q_{j,k}$, respectively, with upper endpoints $w^*$ and $q^*$. The receiver programs $|w_{i,k} - w^*|$ via the so-OPPRF, and the sender evaluates at $\tilde{q}_{j,k,h}$. In addition, the sender knows $q_{j,k}$ and $q^*$ locally, therefore it can compute $|q^* - q_{j,k}|$ directly. If $\tilde{q}_{j,k,h}$ matches some $\tilde{w}_{i,k,h'}$, then $q^* = w^*$, and the distance decomposes as $|w_{i,k} - q_{j,k}| = |w_{i,k} - w^*| + |q^* - q_{j,k}|$. Since $|w_{i,k} - w^*|$ is secret-shared between the two parties and $|q^* - q_{j,k}|$ is known to the sender, the $p$-th power can be expanded via the binomial theorem:
\[
|w_{i,k} - q_{j,k}|^p = \sum_{c=0}^{p} \binom{p}{c} |w_{i,k} - w^*|^{c} \cdot |q^* - q_{j,k}|^{p-c}.
\]
This decomposition holds when $q_{j,k} \le w_{i,k}$; the symmetric case $q_{j,k} > w_{i,k}$ is handled using the lower endpoint instead. To cover both cases, we split $[w_{i,k}-\delta,\, w_{i,k}+\delta]$ into $[w_{i,k}-\delta,\, w_{i,k}]$ and $[w_{i,k}+1,\, w_{i,k}+\delta]$, applying the upper endpoint to the left sub-interval and the lower endpoint to the right. We use $\sigma \in \{0, 1\}$ to index the two sub-intervals. The two cases are illustrated in Figure~\ref{Fig: prefix distance}, and the detailed sub-protocol for computing the distance is given in Figure~\ref{Prot: sum distance}.

\begin{figure*}[t]
    \begin{nfprot}{\ensuremath{\Pi_\mathsf{FPSI\text{-}Px}^{L_p}}\xspace}
        \noindent \textbf{Parameters:}  Sender \SSS and receiver \RRR with input sizes $m$ and $n$. Threshold \( \delta \). Bit length $\ell = \lambda +d+ \log m + p\log\delta$. Prefix length $\mu = 1+ \log\delta$

        \noindent \textbf{Input:} \SSS inputs \( Q=\{\vecq_j\}_{j\in [m]} \in \mathbb{U}^{d\times m} \) and \RRR inputs \( W=\{\vecw_i\}_{i\in [n]} \in \mathbb{U}^{d\times n} \).

        \noindent \textbf{Protocol:}
        \begin{enumerate}
            \item For $i\in[n]$, \RRR computes $\CCC_i^\RRR = \mathsf{cell}_{2\delta}(\vecw_i)$. \SSS shuffles $Q$ and computes $\{ \CCC^{\SSS}_{j,z}\}_{z\in [2^d]} = \mathsf{neigh}_{2\delta}(\vecq_j)$ for $j\in [m]$.

            \item \RRR sets $\mathsf{List} = \{(\CCC^{\RRR}_{i} \Vert k \Vert \sigma \Vert \tilde{w}_{i,k,h,\sigma},0^\ell\Vert\abs{w^*-w_{i,k}}) \}_{i\in[n],k\in[d],h\in[\mu],\sigma\in [0,1]}$ where $\{\tilde{w}_{i,k,h,0}\}_{h\in[\mu]} = \mathsf{PxTrie}(w_{i,k}-\delta,w_{i,k})$ and $\{\tilde{w}_{i,k,h,1}\}_{h\in[\mu]} = \mathsf{PxTrie}(w_{i,k}+1,w_{i,k}+\delta)$ and $w^*=\mathsf{UpBound}(\tilde{w}_{i,k,h,\sigma})$ if $\sigma=0$ otherwise $w^*=\mathsf{LowBound}(\tilde{w}_{i,k,h,\sigma})$.

            \item \SSS computes $\mathsf{query} = \{\CCC^{\SSS}_{j,z} \Vert k\Vert \sigma \Vert \tilde{q}_{j,k,h}\}_{j\in [m],k\in [d],z\in [2^d],h\in [\mu],\sigma\in [0,1]}$ where $\{\tilde{q}_{j,k,h} \}_{h\in [\mu]} = \mathsf{PxPath}(q_{j,k},\delta/2)$.

            \item \RRR and \SSS invoke functionality \Func[so\text{-}OPPRF] where \RRR inputs $\mathsf{List}$ and \SSS inputs $\mathsf{query}$. \RRR receives $\{e^\RRR_{j,k,z,h,\sigma}\Vert r^\RRR_{j,k,z,h,\sigma}\}_{j\in [m],k\in [d],z\in [2^d],h\in [\mu],\sigma\in [0,1]}$ and \SSS receives $\{e^\SSS_{j,k,z,h,\sigma}\Vert r^\SSS_{j,k,z,h,\sigma}\}_{j\in [m],k\in [d],z\in [2^d],h\in [\mu],\sigma\in [0,1]}$.

            \item For $j\in [m],k\in [d], z\in[2^d],h\in [\mu],\sigma\in [0,1]$,  \RRR and \SSS invoke functionality \Func[B2A], where \RRR inputs $r^{\RRR}_{j,k,z,h,\sigma}$ and \SSS inputs $r^{\SSS}_{j,k,z,h,\sigma}$. \RRR receives $s^{\RRR}_{j,k,z,h,\sigma}$ and \SSS receives $s^{\SSS}_{j,k,z,h,\sigma}$.

            \item For $j\in [m],k\in [d], z\in[2^d],h\in [\mu],\sigma\in [0,1]$, \RRR and \SSS invoke $\Pi^p_{\mathsf{getDist}}$ where \RRR inputs $s^{\RRR}_{j,k,z,h,\sigma}$ and \SSS inputs $(s^{\SSS}_{j,k,z,h,\sigma},q_{j,k},\tilde{q}_{j,k,h},\sigma)$. \RRR receives $v^{\RRR}_{j,k,z,h,\sigma}$ and \SSS receives $v^{\SSS}_{j,k,z,h,\sigma}$.

            \item For $j\in [m],k\in [d], z\in[2^d]$, \RRR and \SSS invoke functionality $\FFF^{\mathsf{Eq},2\cdot\mu}_{\mathsf{ExtConSel}}$ where \RRR inputs $\{(e^\RRR_{j,k,z,h,\sigma},v^\RRR_{j,k,z,h,\sigma})\}_{h\in [\mu],\sigma \in [0,1]}$ and \SSS inputs $\{(-e^\SSS_{j,k,z,h,\sigma},v^\SSS_{j,k,z,h,\sigma})\}_{h\in [\mu],\sigma \in [0,1]}$. \RRR receives $d^\RRR_{j,k,z}$ and \SSS receives $d^\SSS_{j,k,z}$.

            \item For $j\in [m],z\in [2^d]$, \RRR computes $r^{\RRR}_{j,z}={\sum}_{k\in[d]}d^{\RRR}_{j,k,z}$ and \SSS computes $r^{\SSS}_{j,z}={\sum}_{k\in[d]}d^{\SSS}_{j,k,z}$.

            \item For $j\in [m],z\in[2^d]$, \RRR and \SSS invoke functionality \ensuremath{\FFF^{\delta^p}_{\mathsf{Interval}}} where \RRR inputs $r^{\RRR}_{j,z}$ and \SSS inputs $r^{\SSS}_{j,z}$. \RRR receives $b_{j,z} = \mathbf{1}\{ r^{\RRR}_{j,z} + r^{\SSS}_{j,z} \le \delta^p \}$.

            \item For $j\in [m]$, \RRR computes $b_j = {\bigvee}_{z\in[2^d]}b_{j,z}$.

            \item For $j\in [m]$, \RRR and \SSS invoke functionality \Func[OT] where \RRR inputs $b_j$ and \SSS inputs $(\bot,\vecq_j)$. \RRR receives $\vecu_j$.

            \item \RRR outputs $I = \{ \vecu_j \; \vert \; b_j=1 \;\text{for} \;  j\in[m] \}$.
        \end{enumerate}

    \end{nfprot}
    \vspace{0.5em}
    \caption{Protocol of prefix-optimized fuzzy PSI for $L_p$ distance under receiver-sided unique cell assumption.}
    \label{Prot: receiver 2delta Lp px}
\end{figure*}

\section{Additional Protocols under Unique Block Assumptions}
\label{appendix: other protocols 4delta}

The protocol under the receiver-sided \textit{unique block} assumption for $L_p$ distances is given in Figure~\ref{Prot: receiver 4delta Lp}. The key difference from the $L_\infty$ construction is that the receiver programs the $p$-th power of the per-dimension distance rather than $\mathbf{0}$, and invokes $\FFF^{\delta^p}_{\mathsf{Interval}}$ to test whether the aggregated distance falls within the threshold $\delta^p$. The prefix-optimized variant for $L_\infty$ distance follows the approach of Section~\ref{Sec: overview prefix trie}, and the detailed protocol is given in Figure~\ref{Prot: receiver 4delta Linf prefix}. 

The construction for general $L_p$ distances follows the same idea as in Appendix~\ref{appendix: receiver 2delta Lp px}, with the prefix-optimized protocol given in Figure~\ref{Prot: receiver 4delta Lp prefix}. The protocols under the sender-sided unique block assumption follow analogously, and we omit the details for brevity.

\begin{figure*}[t]
    \begin{nfprot}{\ensuremath{\Pi_\mathsf{FPSI}^{L_p}}\xspace}
        \noindent \textbf{Parameters:}  Distance threshold \( \delta \). Bit length $\ell = \lambda+\log m+p\log \delta$.

        \noindent \textbf{Input:} \SSS inputs \( Q=\{\vecq_j\}_{j\in [m]} \in \mathbb{U}^{d\times m} \) and \RRR inputs \( W=\{\vecw_i\}_{i\in [n]} \in \mathbb{U}^{d\times n} \).

        \noindent \textbf{Protocol:}
        \begin{enumerate}
            \item For $i\in[n]$, \RRR computes $\{\CCC_{i,z}^\RRR\}_{z\in[2^d]} = \mathsf{neigh}_{2\delta}(\vecw_i)$. \SSS shuffles $Q$ and computes $ \CCC^{\SSS}_{j} = \mathsf{cell}_{2\delta}(\vecq_j)$ for $j\in [m]$.

            \item \RRR sets $\mathsf{List}\! =\! \{(\CCC^{\RRR}_{i,z} \Vert k \Vert w_{i,k}+t, \!\abs{t}^p) \}_{i\in[n],k\in[d],z\in[2^d],t\in[\text{-}\delta,\delta]}$
            \item \RRR and \SSS invoke functionality \Func[so\text{-}OPPRF] where \RRR inputs $\mathsf{List}$ and \SSS inputs $\{\CCC^{\SSS}_j \Vert k\Vert q_{j,k}\}_{j\in[m],k\in[d]}$. \RRR receives $\{e^\RRR_{j,k}\}_{j\in [m],k\in [d]}$ and \SSS receives $\{e^\SSS_{j,k}\}_{j\in [m],k\in [d]}$.

            \item For $j\in [m],k\in [d]$, \SSS and \RRR invoke functionality {\Func[B2A], where \SSS inputs $e^{\SSS}_{j,k}$ and \RRR inputs $e^{\RRR}_{j,k}$. \SSS receives $s^{\SSS}_{j,k}$ and \RRR receives $s^{\RRR}_{j,k}$}.

            \item For $j\in[m]$, \RRR computes $r^\RRR_j={\sum}_{k\in[d]}s^\RRR_{j,k}$ and \SSS computes $r^\SSS_j={\sum}_{k\in[d]}s^\SSS_{j,k}$

            \item For $j\in [m]$, \RRR and \SSS invoke functionality \ensuremath{\FFF^{\delta^p}_{\mathsf{Interval}}} where \RRR inputs $r^{\RRR}_{j}$ and \SSS inputs $r^{\SSS}_{j}$. \RRR receives $b_{j} = \mathbf{1}\{ r^{\RRR}_{j} + r^{\SSS}_{j} \le \delta^p \}$.

            \item For $j\in [m]$, \RRR and \SSS invoke functionality \Func[OT] where \RRR inputs $b_j$ and \SSS inputs $(\bot,\vecq_j)$. \RRR receives $\vecu_j$.

            \item \RRR outputs $I = \{ \vecu_j \; \vert \; b_j=1 \;\text{for} \;  j\in[m] \}$.

        \end{enumerate}

    \end{nfprot}
    \vspace{0.5em}
    \caption{Protocol of fuzzy PSI for $L_p$ distance under receiver-sided unique block assumption.}
    \label{Prot: receiver 4delta Lp}
\end{figure*}

\begin{figure*}[t]
    \begin{nfprot}{\ensuremath{\Pi_\mathsf{FPSI\text{-}Px}^{L_\infty}}\xspace}
        \noindent \textbf{Parameters:}  Distance threshold \( \delta \). Bit length $\mu = \log\delta+2$ and $\ell = \lambda+\log(md\mu) $

        \noindent \textbf{Input:} \SSS inputs \( Q=\{\vecq_j\}_{j\in [m]} \in \mathbb{U}^{d\times m} \) and \RRR inputs \( W=\{\vecw_i\}_{i\in [n]} \in \mathbb{U}^{d\times n} \).

        \noindent \textbf{Protocol:}
        \begin{enumerate}

            \item For $i\in[n]$, \RRR computes $\{\CCC_{i,z}^\RRR\}_{z\in[2^d]} = \mathsf{neigh}_{2\delta}(\vecw_i)$. \SSS shuffles $Q$ and computes $ \CCC^{\SSS}_{j} = \mathsf{cell}_{2\delta}(\vecq_j)$ for $j\in [m]$.

            \item \RRR sets $\mathsf{List} = \{(\CCC^{\RRR}_{i,z}\Vert k \Vert \tilde{w}_{i,k,h},0^{\ell})\}_{i\in[n],z\in[2^d],k\in[d],h\in [\mu]}$ where $\{\tilde{w}_{i,k,h}\}_{h\in[\mu]} = \mathsf{PxTrie}(w_{i,k}-\delta,w_{i,k}+\delta)$
            \item \SSS computes $\{\CCC^{\SSS}_j \Vert k \Vert \tilde{q}_{j,k,h}\}_{j\in[m],k\in[d],h\in[\mu]}$ where $\{\tilde{q}_{j,k,h}\}_{h\in[\mu]} = \mathsf{PxPath}(q_{j,k},\delta)$
            \item \RRR and \SSS invoke functionality \Func[so\text{-}OPPRF] where \RRR inputs $\mathsf{List}$ and \SSS inputs $\{\CCC^{\SSS}_j \Vert k \Vert \tilde{q}_{j,k,h}\}_{j\in[m],k\in[d],h\in[\mu]}$. \RRR receives $\{e^\RRR_{j,k,h}\}_{j\in[m],k\in[d],h\in[\mu]}$ and \SSS receives $\{e^\SSS_{j,k,h}\}_{j\in[m],k\in[d],h\in[\mu]}$.

            \item For $j\in[m],k\in[d]$, \RRR and \SSS invoke functionality $\FFF^{\mathsf{Eq},\mu}_{\mathsf{ConSel}}$ where \RRR inputs $\{e^\RRR_{j,k,h}\}_{h\in[\mu]}$ and \SSS inputs $\{-e^\SSS_{j,k,h}\}_{h\in[\mu]}$. \RRR receives $r^\RRR_{j,k}$ and \SSS receives $r^\SSS_{j,k}$.
            
            \item For $j\in[m]$, \RRR computes $r^\RRR_j={\sum}_{k\in[d]}r^\RRR_{j,k}$. \SSS computes $r^\SSS_j={\sum}_{k\in[d]}r^\SSS_{j,k}$
            \item For $j\in[m]$, \RRR and \SSS invoke functionality \Func[Eq] where \RRR inputs $r^\RRR_j$ and \SSS inputs $-r^\SSS_j$. \RRR receives $b_j=\mathbf{1}\{r^\RRR_j=-r^\SSS_j\}$.
            \item For $j\in[m]$, \RRR and \SSS invoke functionality \Func[OT] where \RRR inputs $b_j$ and \SSS inputs $(\perp,\vecq_j)$. \RRR receives ${\vecu_j}$.
            \item \RRR outputs $\{\vecu_j\; \vert \; b_j=1 \;\text{for} \; j\in[m] \}$.
        \end{enumerate}

    \end{nfprot}
    \vspace{0.5em}
    \caption{Protocol of prefix-optimized fuzzy PSI for $L_\infty$ distance under receiver-sided unique block assumption.}
    \label{Prot: receiver 4delta Linf prefix}
\end{figure*}

\begin{figure*}[t]
    \begin{nfprot}{\ensuremath{\Pi_\mathsf{FPSI\text{-}Px}^{L_p}}\xspace}
        \noindent \textbf{Parameters:}  Distance threshold \( \delta \). Bit length $\ell = \lambda+\log m+p\log \delta$. 

        \noindent \textbf{Input:} \SSS inputs \( Q=\{\vecq_j\}_{j\in [m]} \in \mathbb{U}^{d\times m} \) and \RRR inputs \( W=\{\vecw_i\}_{i\in [n]} \in \mathbb{U}^{d\times n} \).

        \noindent \textbf{Protocol:}
        \begin{enumerate}

            \item For $i\in[n]$, \RRR computes $\{\CCC_{i,z}^\RRR\}_{z\in[2^d]} = \mathsf{neigh}_{2\delta}(\vecw_i)$. \SSS shuffles $Q$ and computes $ \CCC^{\SSS}_{j} = \mathsf{cell}_{2\delta}(\vecq_j)$ for $j\in [m]$.
            
            \item \RRR sets $\mathsf{List} = \{(\CCC^{\RRR}_{i,z} \Vert k \Vert \sigma \Vert \tilde{w}_{i,k,h,\sigma},0^\ell\Vert\abs{w^*-w_{i,k}} )\}_{i\in[n],z\in[2^d],k\in[d],h\in [\mu],\sigma\in[0,1]}$ where $\{\tilde{w}_{i,k,h,0}\}_{h\in[\mu]} = \mathsf{PxTrie}(w_{i,k}-\delta,w_{i,k})$ and $\{\tilde{w}_{i,k,h,1}\}_{h\in[\mu]} = \mathsf{PxTrie}(w_{i,k}+1,w_{i,k}+\delta)$ and $w^*=\mathsf{UpBound}(\tilde{w}_{i,k,h,\sigma})$ if $\sigma=0$ otherwise $w^*=\mathsf{LowBound}(\tilde{w}_{i,k,h,\sigma})$.
            \item \SSS computes $\{\CCC^{\SSS}_j \Vert k \Vert \sigma \Vert \tilde{q}_{j,k,h}\}_{j\in[m],k\in[d],h\in[\mu],\sigma\in[0,1]}$ where $\{\tilde{q}_{j,k,h}\}_{h\in[\mu]} = \mathsf{PxPath}(q_{j,k},\delta/2)$.
            \item \RRR and \SSS invoke functionality \Func[so\text{-}OPPRF] where \RRR inputs $\mathsf{List}$ and \SSS inputs $\{\CCC^{\SSS}_j \Vert k\Vert \sigma \Vert \tilde{q}_{j,k,h}\}_{j\in[m],k\in[d],h\in[\mu],\sigma\in[0,1]}$. \RRR receives $\{e^\RRR_{j,k,h,\sigma}\Vert r^\RRR_{j,k,h,\sigma}\}_{j\in [m],k\in [d],h\in [\mu],\sigma \in [0,1]}$ and \SSS receives $\{e^\SSS_{j,k,h,\sigma}\Vert r^\SSS_{j,k,h,\sigma}\}_{j\in [m],k\in [d],h\in [\mu],\sigma \in [0,1]}$.

            \item For $j\in [m],k\in [d],h\in [\mu],\sigma \in [0,1]$, \SSS and \RRR invoke functionality {\Func[B2A], where \RRR inputs $r^{\RRR}_{j,k,h,\sigma}$ and \SSS inputs $r^{\SSS}_{j,k,h,\sigma}$.  \RRR receives $s^{\RRR}_{j,k,h,\sigma}$} and \SSS receives $s^{\SSS}_{j,k,h,\sigma}$.

            \item For $j\in [m],k\in [d],h\in[\mu],\sigma \in [0,1]$, \SSS and \RRR invoke $\Pi^p_{\mathsf{getDist}}$, where \RRR inputs $s^{\RRR}_{j,k,h,\sigma}$ and \SSS inputs $( s^{\SSS}_{j,k,h,\sigma},q_{j,k}, \tilde{q}_{j,k,h},\sigma)$.
            \RRR receives $d^\RRR_{j,k,h,\sigma}$ and \SSS receives $d^\SSS_{j,k,h,\sigma}$.

            \item For $j\in [m],k\in [d]$, \RRR and \SSS invoke functionality $\FFF^{\mathsf{Eq},2\cdot\mu}_{\mathsf{ExtConSel}}$ where \RRR inputs $\{(e^\RRR_{j,k,h,\sigma},d^\RRR_{j,k,h,\sigma})\}_{h\in[\mu],\sigma\in[0,1]}$ and \SSS inputs $\{(-e^\SSS_{j,k,h,\sigma},d^\SSS_{j,k,h,\sigma})\}_{h\in[\mu],\sigma\in[0,1]}$. \RRR receives $r^\RRR_{j,k}$ and \SSS receives $r^\SSS_{j,k}$
            \item For $j\in[m]$, \RRR computes $r^\RRR_j={\sum}_{k\in[d]}r^\RRR_{j,k}$ and \SSS computes $r^\SSS_j={\sum}_{k\in[d]}r^\SSS_{j,k}$.
            
            \item For \( j \in [m] \), \SSS and \RRR invoke functionality { $\FFF^{\delta^p}_{\mathsf{Interval}}$, where \RRR inputs $r^{\RRR}_{j}$ and \SSS inputs $r^{\SSS}_{j}$. \RRR receives $b_j := \mathbf{1}\{ r^{\SSS}_{j}+r^{\RRR}_{j}\le \delta^p\}$}.
            
            \item For \( j \in [m] \), \RRR and \SSS invoke functionality \Func[OT], where \RRR inputs $b_j$ and \SSS inputs $(\perp ,\vecq_j)$. \RRR receives OT outputs $\vecu_j$.
            
            \item \RRR outputs $\{ \vecu_j \;\vert \; b_j = 1\;\text{for}\; j\in [m]\}$.
        \end{enumerate}

    \end{nfprot}
    \vspace{0.5em}
    \caption{Protocol of prefix-optimized fuzzy PSI for $L_p$ distance under receiver-sided unique block assumption.}
    \label{Prot: receiver 4delta Lp prefix}
\end{figure*}

\section{Additional Experimental Results}
\label{appendix: results Lp 2delta}
\label{appendix: results sender 2delta}
We present the results of our protocols for $L_1$ and $L_2$ distances under the receiver-sided unique cell assumption in Table~\ref{Tab: 2delta Lp}. As no existing work supports $L_p$ distances under this assumption, we omit a direct comparison. 

We present the results of our protocols for $L_\infty$, $L_1$, and $L_2$ distances under the sender-sided unique cell assumption in Table~\ref{Tab: 2delta sender}. As no existing work operates under this assumption, we omit a direct comparison. 

\newpage

\begin{table*}[t]
\caption{The communication (MB) and running time (s) of our protocols for $L_1$ and $L_2$ distances under the receiver-sided unique cell assumption.}
\label{Tab: 2delta Lp}
\begin{tabular}{ccccccccccccc}
\hline
\multirow{2}{*}{\begin{tabular}[c]{@{}c@{}}Size\\ $m\!=\!n$\end{tabular}} & \multirow{2}{*}{\begin{tabular}[c]{@{}c@{}}Dim.\\ $d$\end{tabular}} & \multirow{2}{*}{Metric} & \multicolumn{2}{c}{$\delta=32$} & \multicolumn{2}{c}{$\delta=64$} & \multicolumn{2}{c}{$\delta=128$} & \multicolumn{2}{c}{$\delta=256$} & \multicolumn{2}{c}{$\delta=512$} \\ \cline{4-13} 
                                                                          &                                                                     &                           & Comm           & Time           & Comm           & Time           & Comm            & Time           & Comm            & Time           & Comm            & Time           \\ \cline{1-13} 
\multirow{6}{*}{$2^{8}$}  & \multirow{2}{*}{2} & $L_1$
  & {3.52}  & {0.22}
  & {4.23}  & {0.30}
  & {5.61}  & {0.37}
  & {8.36}  & {0.41}
  & {13.83} & {0.47} \\
                          &                    & $L_2$
  & 3.66  & 0.29
  & 4.41  & 0.30
  & 5.82  & {0.39}
  & 8.93  & 0.42
  & 14.39 & 0.52 \\ \cline{2-13}
                          & \multirow{2}{*}{4} & $L_1$
  & {14.23} & 0.59
  & {15.63} & 0.60
  & {18.38} & {0.72}
  & {23.95} & {0.91}
  & {34.95} & {1.00} \\
                          &                    & $L_2$
  & 14.67 & {0.51}
  & 16.53 & {0.59}
  & 19.31 & 0.74
  & 24.99 & 0.83
  & 36.11 & 1.15 \\ \cline{2-13}
                          & \multirow{2}{*}{6} & $L_1$
  & {69.89}  & {2.17}
  & {72.42}  & {2.34}
  & {77.00}  & 2.38
  & {85.65}  & {2.47}
  & {102.52} & {2.86} \\
                          &                    & $L_2$
  & 72.37  & 2.16
  & 75.41  & 2.35
  & 80.47  & {2.54}
  & 89.63  & 2.65
  & 106.99 & 3.34 \\ \cline{1-13}

\multirow{6}{*}{$2^{12}$} & \multirow{2}{*}{2} & $L_1$
  & {36.67}  & {1.20}
  & {48.04}  & {1.32}
  & {70.31}  & {1.81}
  & {114.37} & {2.50}
  & {202.04} & {3.89} \\
                           &                    & $L_2$
  & 39.15  & {1.27}
  & 51.03  & 1.37
  & 73.78  & 1.82
  & 118.35 & 2.75
  & 206.52 & 4.67 \\ \cline{2-13}
                           & \multirow{2}{*}{4} & $L_1$
  & {204.68} & 6.00
  & {228.42} & {6.21}
  & {273.96} & {6.71}
  & {363.14} & {9.38}
  & {539.56} & {11.75} \\
                           &                    & $L_2$
  & 214.59 & {5.81}
  & 240.33 & 6.37
  & 287.85 & 7.30
  & 380.10 & 9.62
  & 558.65 & 14.17 \\ \cline{2-13}
                           & \multirow{2}{*}{6} & $L_1$
  & {1098.30} & {32.16}
  & {1139.17} & {32.28}
  & {1212.78} & {33.11}
  & {1352.03} & {34.02}
  & {1622.35} & {39.18} \\
                           &                    & $L_2$
  & 1139.32 & 31.48
  & 1188.40 & 35.79
  & 1270.15 & 36.54
  & 1417.61 & 38.15
  & 1696.18 & 51.41 \\ \cline{1-13}

\multirow{6}{*}{$2^{16}$} & \multirow{2}{*}{2} & $L_1$
  & {566.83}  & {13.85}
  & {749.43}  & {16.79}
  & {1106.61} & {23.28}
  & {1813.41} & {35.78}
  & {3219.26} & {62.07} \\
                           &                    & $L_2$
  & 607.85  & 15.78
  & 798.66  & 19.80
  & 1163.98 & 27.17
  & 1878.99 & 46.99
  & 3293.09 & 81.30 \\ \cline{2-13}
                           & \multirow{2}{*}{4} & $L_1$
  & {3259.14} & 96.95
  & {3640.86} & 99.99
  & {4372.04} & {111.70}
  & {5802.69} & {138.92}
  & {8632.00} & {195.31} \\
                           &                    & $L_2$
  & 3423.32 & {101.09}
  & 3837.89 & {104.03}
  & 4601.68 & 122.70
  & 6082.17 & 158.80
  & 8946.50 & 251.82 \\ \cline{2-13}
                           & \multirow{2}{*}{6} & $L_1$
  & {17575.29} & 540.47
  & {18234.60} & 554.85
  & {19418.77} & {589.16}
  & {21655.27} & {617.48}
  & {25993.49} & {704.26} \\
                           &                    & $L_2$
  & 18252.49 & {538.05}
  & 19047.23 & {566.22}
  & 20365.92 & 591.20
  & 22737.88 & 668.20
  & 27212.29 & 868.49 \\ \cline{1-13}
\end{tabular}
\end{table*}

\begin{table*}[t]
\caption{The communication (MB) and running time (s) of our protocols under the sender-sided unique cell assumption.}
\label{Tab: 2delta sender}
\begin{tabular}{ccccccccccccc}
\hline
\multirow{2}{*}{\begin{tabular}[c]{@{}c@{}}Size\\ $m\!=\!n$\end{tabular}} & \multirow{2}{*}{\begin{tabular}[c]{@{}c@{}}Dim.\\ $d$\end{tabular}} & \multirow{2}{*}{Metric} & \multicolumn{2}{c}{$\delta=32$} & \multicolumn{2}{c}{$\delta=64$} & \multicolumn{2}{c}{$\delta=128$} & \multicolumn{2}{c}{$\delta=256$} & \multicolumn{2}{c}{$\delta=512$} \\ \cline{4-13} 
                                                                          &                                                                     &                         & Comm           & Time           & Comm           & Time           & Comm            & Time           & Comm            & Time           & Comm            & Time           \\ \hline
\multirow{9}{*}{$2^{8}$}                                                  & \multirow{3}{*}{2}                                                  & $L_\infty$              & 2.34           & 0.31           & 3.02           & 0.32           & 4.38            & 0.34           & 7.10            & 0.42           & 12.53           & 0.49           \\
                                                                          &                                                                     & $L_1$                   & 3.37           & 0.41           & 4.05           & 0.47           & 5.41            & 0.48           & 8.13            & 0.53           & 13.57           & 0.70           \\
                                                                          &                                                                     & $L_2$                   & 3.37           & 0.44           & 4.05           & 0.48           & 5.41            & 0.49           & 8.13            & 0.59           & 13.57           & 0.65           \\ \cline{2-13} 
                                                                          & \multirow{3}{*}{4}                                                  & $L_\infty$              & 8.54           & 0.49           & 9.89           & 0.60           & 12.61           & 0.63           & 18.05           & 0.66           & 28.93           & 0.82           \\
                                                                          &                                                                     & $L_1$                   & 14.11          & 0.86           & 15.47          & 0.87           & 18.19           & 1.00           & 23.63           & 1.05           & 34.51           & 1.20           \\
                                                                          &                                                                     & $L_2$                   & 14.11          & 0.83           & 15.47          & 0.84           & 18.19           & 0.87           & 23.63           & 1.06           & 34.51           & 1.19           \\ \cline{2-13} 
                                                                          & \multirow{3}{*}{6}                                                  & $L_\infty$              & 40.03          & 1.74           & 42.07          & 1.74           & 46.15           & 1.98           & 54.30           & 1.96           & 70.67           & 2.17           \\
                                                                          &                                                                     & $L_1$                   & 70.33          & 2.77           & 72.37          & 2.81           & 76.45           & 3.07           & 84.60           & 3.27           & 100.97           & 3.36           \\
                                                                          &                                                                     & $L_2$                   & 70.33          & 2.86           & 72.37          & 2.89           & 76.45           & 2.92           & 84.60           & 3.03           & 100.97           & 3.04           \\ \hline
\multirow{9}{*}{$2^{12}$}                                                 & \multirow{3}{*}{2}                                                  & $L_\infty$              & 25.33          & 1.13           & 36.21          & 1.34           & 57.98           & 1.71           & 101.55          & 2.45           & 188.72          & 3.71           \\
                                                                          &                                                                     & $L_1$                   & 37.12          & 1.74           & 48.01          & 2.00           & 69.78           & 2.37           & 113.34          & 2.99           & 200.51          & 4.54           \\
                                                                          &                                                                     & $L_2$                   & 37.12          & 1.88           & 48.01          & 2.07           & 69.78           & 2.34           & 113.34          & 3.15           & 200.51          & 4.59           \\ \cline{2-13} 
                                                                          & \multirow{3}{*}{4}                                                  & $L_\infty$              & 124.07         & 4.64           & 145.84         & 5.01           & 189.41          & 5.37           & 276.58          & 7.47           & 451.02          & 10.20          \\
                                                                          &                                                                     & $L_1$                   & 206.53         & 8.31           & 228.30         & 8.32           & 271.86          & 9.61           & 359.04          & 10.45          & 533.47          & 13.39          \\
                                                                          &                                                                     & $L_2$                   & 206.53         & 7.75           & 228.30         & 8.24           & 271.86          & 9.41           & 359.04          & 10.70          & 533.47          & 13.18          \\ \cline{2-13} 
                                                                          & \multirow{3}{*}{6}                                                  & $L_\infty$              & 628.55         & 23.60          & 661.26         & 24.16          & 726.72          & 25.26          & 857.69          & 26.89          & 1119.79         & 30.60          \\
                                                                          &                                                                     & $L_1$                   & 1103.94        & 40.78          & 1136.65        & 40.92          & 1202.10         & 41.55          & 1333.08         & 42.68          & 1595.18         & 46.24          \\
                                                                          &                                                                     & $L_2$                   & 1103.94        & 41.75          & 1136.65        & 39.70          & 1202.10         & 42.01          & 1333.08         & 44.09          & 1595.18         & 46.32          \\ \hline
\multirow{9}{*}{$2^{16}$}                                                 & \multirow{3}{*}{2}                                                  & $L_\infty$              & 393.99         & 13.67          & 568.43         & 17.38          & 917.44          & 24.18          & 1615.98         & 42.35          & 3013.60         & 73.70          \\
                                                                          &                                                                     & $L_1$                   & 573.36         & 24.34          & 747.80         & 27.94          & 1096.82         & 32.89          & 1795.35         & 48.12          & 3192.98         & 73.14          \\
                                                                          &                                                                     & $L_2$                   & 573.36         & 26.34          & 747.80         & 27.75          & 1096.82         & 32.99          & 1795.35         & 45.87          & 3192.98         & 72.37          \\ \cline{2-13} 
                                                                          & \multirow{3}{*}{4}                                                  & $L_\infty$              & 1976.03        & 77.42          & 2325.08        & 77.85          & 3023.61         & 93.59         & 4421.18         & 117.71         & 7217.55         & 186.85         \\
                                                                          &                                                                     & $L_1$                   & 3281.97        & 141.79         & 3631.02        & 142.41         & 4329.55         & 157.54         & 5727.12         & 186.23         & 8523.50         & 228.06         \\
                                                                          &                                                                     & $L_2$                   & 3281.97        & 125.84         & 3631.02        & 132.50         & 4329.55         & 157.35         & 5727.12         & 174.60         & 8523.50         & 243.04         \\ \cline{2-13} 
                                                                          & \multirow{3}{*}{6}                                                  & $L_\infty$              & 10056.07       & 380.90         & 10580.55       & 386.97         & 11630.04        & 408.73         & 13730.14        & 442.44         & 17932.58        & 518.47         \\
                                                                          &                                                                     & $L_1$                   & 17638.65       & 732.56         & 18163.14       & 731.06         & 19212.62        & 762.97         & 21312.72        & 768.98         & 25515.16        & 864.11        \\
                                                                          &                                                                     & $L_2$                   & 17638.65       & 677.65         & 18163.14       & 695.06         & 19212.62        & 714.14         & 21312.72        & 807.86         & 25515.16        & 864.07        \\ \hline
\end{tabular}
\end{table*}